\documentclass[12pt,reqno,hidelinks]{amsart}
\usepackage{amssymb}
\usepackage{amsmath, amssymb, verbatim, url, mathrsfs}
\usepackage{mathrsfs}
\usepackage{dutchcal}
\usepackage{mathtools}
\usepackage{verbatim}
\usepackage{amsthm}
\usepackage{framed}
\usepackage{cite}
\usepackage{wasysym}
\usepackage{upgreek}
\usepackage{color}
\usepackage[dvipsnames]{xcolor}
\usepackage{tensor}
\usepackage{accents}
\usepackage{dsfont}
\usepackage[colorlinks,linkcolor=blue,citecolor=blue]{hyperref}
\usepackage{enumerate}
\usepackage[normalem]{ulem}
\usepackage{longtable}
\usepackage{mathtools}

\numberwithin{equation}{section}

\newtheorem{proposition}{Proposition}[section]
\newtheorem{lemma}{Lemma}[section]
\newtheorem{corollary}{Corollary}[section]
\newtheorem{theorem}{Theorem}[section]

\theoremstyle{definition}
\newtheorem{definition}{Definition}[subsection]
\newtheorem{remark}{Remark}[section]

\newcommand{\tta}{\breve{\mathtt{A}}}
\newcommand{\ttb}{\breve{\mathtt{B}}}
\newcommand{\ttj}{\breve{\mathtt{J}}}
\newcommand{\ttk}{\breve{\mathtt{K}}}
\newcommand{\tte}{\breve{\mathtt{E}}}
\newcommand{\ttf}{\breve{\mathtt{F}}}
\newcommand{\he}{\hat{e}}
\newcommand{\ha}{\hat{a}}
\newcommand{\hb}{\hat{b}}
\newcommand{\hc}{\hat{c}}
\newcommand{\hd}{\hat{d}}
\newcommand{\vertiiii}[1]{{\left\vert\kern-0.25ex\left\vert\kern-0.25ex\left\vert\kern-0.25ex\left\vert #1 \right\vert\kern-0.25ex\right\vert\kern-0.25ex\right\vert\kern-0.25ex\right\vert}}
\newcommand{\vertiii}[1]{{\left\vert\kern-0.25ex\left\vert\kern-0.25ex\left\vert #1 \right\vert\kern-0.25ex\right\vert\kern-0.25ex\right\vert}}
\newcommand{\norm}[1]{\|#1\|}

\newcommand{\tE}{\mathcal{E}}

\newcommand{\mfu}{\mathfrak{u}}

\newcommand{\nnb}{\nonumber}
\newcommand{\ulh}{\underline{h}}
\newcommand{\ulg}{\underline{g}}

\newcommand{\Rbb}{\mathbb{R}}
\newcommand{\Zbb}{\mathbb{Z}}

\newcommand{\Pbb}{\mathbb{P}}

\newcommand{\udl}{\underline}
\newcommand{\tp}[2]{\tensor{p}{^{#1}_{#2}}}
\newcommand{\tss}[2]{\tensor{s}{^{#1}_{#2}}}
\newcommand{\udn}[1]{\underline{\nabla}_{#1}}
\newcommand{\tsh}[2]{\tensor{h}{^{#1}_{#2}}}

\newcommand{\del}[1]{{\partial_{#1}}}

\newcommand{\AND}{{\quad\text{and}\quad}}

\newcommand{\Li}{L^\infty}
\newcommand{\la}{\langle}
\newcommand{\ra}{\rangle}
\newcommand{\ula}{\underline{\alpha}}

\newcommand{\A}{\mathcal{A}}

\newcommand{\al}[2]{
	\begin{align}\label{E:#1}
		#2
	\end{align}
}

\newcommand{\ali}[1]{
	\begin{align}
		#1
	\end{align}
}
\newcommand{\gat}[1]{
	\begin{gather}
		#1
	\end{gather}
}
\newcommand{\als}[1]{
	\begin{align*}
		#1
	\end{align*}
}\newcommand{\gas}[1]{
	\begin{gather*}
		#1
	\end{gather*}
}
\newcommand{\p}[1]{
	\begin{pmatrix}
		#1
	\end{pmatrix}
}

\newcommand{\ts}{\tensor}

\newcommand{\nb}{\underline{\nabla}}
\newcommand{\tr}{\text{tr}}

\newcommand{\Ab}{\tilde A}

\newcommand{\B}{\mathcal{B}}

\newcommand{\Pb}{\mathbb{P}}

\newcommand{\Pbp}{\mathbb{P}^{\perp}}

\newcommand{\be}{\begin{equation}}
	\newcommand{\ee}{\end{equation}}

\allowdisplaybreaks

\DeclareMathOperator{\diag}{diag}

\DeclareMathOperator{\Ima}{Im}
\DeclareMathOperator{\gd}{gd}

\begin{document}

	\title[Global existence and stability of solutions to the EYM equations]{Future global existence and stability of de Sitter-like solutions to the Einstein--Yang--Mills equations in spacetime dimensions $n\geq 4$}

	\author{Chao Liu, Todd A. Oliynyk and Jinhua Wang}
	
	\address[Chao Liu]{School of Mathematics and Statistics and Center for Mathematical Sciences, Huazhong University of Science and Technology, Wuhan 430074, Hubei Province, China.}
	\email{chao.liu.math@foxmail.com}

	\address[Todd A. Oliynyk]{School of Mathematical Sciences, 9 Rainforest Walk, Monash University, Clayton, VIC 3800, Australia}
	\email{todd.oliynyk@monash.edu}

	\address[Jinhua Wang]{School of Mathematical Science, Xiamen University,
		Xiamen 361005, Fujian Province, China. }
	\email{wangjinhua@xmu.edu.cn}
	
	\begin{abstract}
		We establish the global existence and stability to the future of non-linear perturbation of de Sitter-like solutions to the  Einstein--Yang--Mills system in $n\geq 4$ spacetime dimension. This generalizes Friedrich's Einstein--Yang--Mills stability results in dimension $n=4$ [J Differ Geom $34$ ($1991$), $275$–$345$] to all higher dimensions. 
		
		 \vspace{2mm}
		
		
		\vspace{2mm}
		
	\end{abstract}

	\maketitle
	
	\setcounter{tocdepth}{2}
	
	\pagenumbering{roman} \pagenumbering{arabic}

	\section{Introduction}\label{S:INTRO}

In general relativity, the $n$-dimensional de Sitter (dS$_n$) solution 
is an exact, maximally symmetric solution to the vacuum Einstein equations with a positive cosmological constant $\Lambda>0$. It is one of the simplest solutions that exhibits \textit{accelerated expansion}. 
While realistic cosmologies are not homogeneous and isotropic, de Sitter spacetime can still be physically relevant in situations where it is non-linearly stable under perturbations globally to the future. This is because stability will guarantee that the de Sitter spacetime can be used as a model to understand the future asymptotic behavior of cosmological spacetimes that are undergoing accelerated expansions. 
This is of direct physical relevance due to cosmological observations of accelerated expansion in our universe and the assertion that this expansion is due to a cosmological constant \cite{Riess_1998}.

The future stability of de Sitter spacetime  in $n=4$ spacetime dimensions was first established by Friedrich \cite{Friedrich1986} using a conformal representation of the vacuum Einstein field equations with a positive cosmological constant. Friedrich, in the article \cite{Friedrich1991}, subsequently extended this stability result to the Einstein--Yang--Mills equations in $n=4$ spacetime dimension. More precisely, c.f. \cite[Thm.~9.8]{Friedrich1991}, Friedrich established the global existence to the future of solutions to the Einstein--Yang--Mills equations that are generated from initial data sufficiently close to de Sitter initial data. Friedrich also established that these solutions are all asympotically simple. It is worth noting here that, in the article  \cite{Friedrich1991}, Friedrich also considered the case of vanishing cosmological constant $\Lambda =0$ and 
prescribed initial data on past timelike infinity for $\Lambda>0$. 

The main aim of this article is to extend Friedrich's future global existence and stability result for the Einstein--Yang--Mills equations with a positive cosmological constant $\Lambda>0$ to all higher spacetime dimensions $n>4$ and provide a new global existence and stability proof for $n=4$.

\subsection{Einstein--Yang--Mills equations}
Before discussing the main results of this article, we first briefly recall some key concepts from Yang--Mills theory that will be needed to fix our notations and formulate the Einstein--Yang--Mills equations; see also \cite{ChoquetBruhat1992}. 

Let $G$ denote a compact and connected Lie group with Lie algebra $\mathcal{G}$. Due to the compactness, we lose no generality in taking $G$ to be a matrix group.  
On $\mathcal{G}$, we fix an Ad-invariant, positive definite inner-product that we denote by a dot, i.e. $\phi \cdot \psi$ for $\phi,\psi \in \mathcal{G}$, and we use $|\cdot|$ to denote the associated norm. 

Given a $n$-dimensional, connected Lorentzian  manifold\footnote{In this article, we use \textit{abstract index notations}, see \S \ref{s:AIN}. } $(\widetilde{\mathcal{M}}^n, \, \tilde g_{a b})$, a  connection $\tilde{\omega}$ on a	 $G$-principal bundle over $\widetilde{\mathcal{M}}^n$ can be expressed in a gauge as a local   $\mathcal{G}$-valued $1$-form $\Ab_a$ on $\widetilde{\mathcal{M}}^n$, which is referred to as a \textit{gauge potential}. The \textit{curvature} of the connection $\tilde{\omega}$ is then determined locally by the $\mathcal{G}$-valued $2$-form $\tilde F_{ab}$ on $\widetilde{\mathcal{M}}^n$ defined by\footnote{The Yang--Mills curvature $\tilde F_{ab}$ is globally defined when viewed as taking values in the adjoint bundle.} 
\be\label{def-F}
\tilde F_{ab} = \tilde\nabla_{a} \Ab_b-\tilde\nabla_{b} \Ab_a + [\Ab_a, \Ab_b ],
\ee
where $[\cdot,\cdot]$ is the Lie bracket on $\mathcal{G}$, i.e. the matrix commutator bracket, and
$\tilde\nabla_a$ is the covariant derivative associated to $\tilde{g}_{a b}$. We recall also that the Yang--Mills curvature $\tilde F_{ab}$ automatically satisfies the  \textit{Bianchi identities}
\begin{equation*}
	\tilde D_{[a} \tilde F_{b c]} = 0,  
\end{equation*}
where $\tilde D_a = \tilde \nabla_a + [\tilde A_a, \, \cdot]$ denotes the gauge covariant derivative of a $\mathcal{G}$-valued tensor.

The \textit{Einstein--Yang--Mills equations} with a positive cosmological constant $\Lambda>0$ are then given by
\begin{align}
	\tilde G_{a b} + \Lambda \tilde g_{a b}& = \tilde T_{a b}, \label{eq-einstein} \\
	\tilde D^a \tilde F_{a b} & =0,  \label{eq-maxwell-div-o} \\
	\tilde D_{[a}\tilde F_{bc]}&= 0,  \label{eq-maxwell-bianchi-o}
\end{align}
where $\tilde{G}_{ab}=\tilde{R}_{ab}-\frac{1}{2}\tilde{R} \tilde{g}_{ab}$ is the Einstein tensor of the metric $\tilde{g}_{ab}$ and the stress energy tensor of a Yang--Mills field is defined by
\begin{align*}
	\tilde T_{a b} = \tilde F_{a}{}^{ c} \cdot \tilde F_{b c} - \frac{1}{4} \tilde g_{a b}  \tilde F^{c d}\cdot \tilde F_{c d}.
\end{align*}
In this article, we will restrict our attention to trivial principal bundles, i.e. $\widetilde{\mathcal{M}}^n \times G$. The purpose of this simplification is to allows us to work with gauge potentials $A_a$ that exists globally on $\widetilde{\mathcal{M}}^n$. This, in turns, allows us to view the Einstein--Yang--Mills equations as equations for the fields $(\tilde{g}_{ab},\tilde{A}_a)$ on $\widetilde{\mathcal{M}}^n$. 

\subsection{de Sitter spacetime}\label{s:ds}
For the remainder of the article, we fix the physical spacetime manifold by setting\footnote{In fact, we will only make use of the future half $[0,\infty)\times \Sigma$ of  $\widetilde{\mathcal{M}}^n$ given by $[0,\infty)\times \Sigma$.}
\begin{equation*}
	\widetilde{\mathcal{M}}^n =\Rbb\times \Sigma
\end{equation*}
with
\begin{equation*}
	\Sigma=\mathbb{S}^{n-1}.
\end{equation*}
Then de Sitter spacetime\footnote{See the references \cite{Hawking2010,Spradlin2001} for a more detailed introduction to de Sitter spacetime.}  $(\widetilde{\mathcal{M}}^n, \, \tilde{\ulg}_{ab} )$ is obtained by equipping $\widetilde{\mathcal{M}}^n$ with the de Sitter metric defined by
\al{DESMETR}{
	\tilde{\ulg}_{ab}=-(d\tau)_a(d\tau)_b+H^2\cosh^2(H^{-1}\tau)\ulh_{ab}
}
where the constant $H$ is determined by \al{H}{
	H=\sqrt{\frac{(n-2)(n-1)}{2\Lambda},}
}
$\ulh_{ab}$ is the standard metric on $\mathbb{S}^{n-1}$, and $\tau$ is a Cartesian coordinate function on $\Rbb$.
This spacetime plays two roles in our subsequent arguments. First, it provides an ambient spacetime manifold on which we can formulate our global existence and stability results, and second, it provides the background geometric quantities that we use to quantify the size of the metric perturbations away from the de Sitter metric.    

In the analysis carried out in this article, we find it advantageous to work with a conformally rescaled version of the de Sitter metric, which we refer to as the \textit{conformal de Sitter metric}, rather than the de Sitter metric itself. To define the conformal de Sitter metric, we introduce a new time function $t$ via
\al{COOR1}{
	t= \frac{1}{H}\left(\frac{\pi}{2}-\gd(H^{-1} \tau)\right)  
}
where $\gd(x)$, known as the  \textit{Gudermannian function}, is defined by 
\begin{equation*} 
	\gd(x)=\int^{x}_0\frac{1}{\cosh s} ds=\arctan\bigl(\sinh(x)\bigr), \quad x\in \Rbb.
\end{equation*}
Inverting \eqref{E:COOR1} gives
\al{GDINV}{
	\tau=H \gd^{-1}\left(\frac{\pi}{2}-Ht\right)
}
where 
\begin{equation*}
	\gd^{-1}(x)=\int^x_0\frac{1}{\cos t}dt=\text{arctanh}(\sin x), \quad x\in \Bigl(-\frac{\pi}{2},\frac{\pi}{2}\Bigr).
\end{equation*}
From \eqref{E:GDINV}, we observe that $\tau(t)$ is a monotonic, decreasing  and analytic for $t\in\bigl(0,\frac{\pi}{H}\bigr)$. We further observe that \eqref{E:COOR1} maps the infinite interval
$\tau \in (-\infty, +\infty)$ into the finite interval $t \in (0, \frac{\pi}{H})$ with a change of time orientation where the future lies in the direction of decreasing $t$ due to the monotonic, decreasing behavior of $t(\tau)$ on $\Rbb$. Noting that 
$t(0)=\frac{\pi}{2H}$ and $t(+\infty)=0$, we conclude that $t=0$ corresponds to future timelike infinity in de Sitter spacetime while $t=\frac{\pi}{2H}$ corresponds to $\tau=0$.

Differentiating \eqref{E:GDINV}, we find, with the help of  the identity $\frac{d}{d x}(\gd^{-1} x)=\sec x$, that
\begin{equation}\label{e:dtaut}
	(d\tau)_a=-H^2 \sec\left(\frac{\pi}{2}-Ht\right) (dt)_a.
\end{equation}
Then, noting the identity  $\cosh (\gd^{-1}(x))=\sec(x)$, we see from \eqref{e:dtaut} that the de Sitter metric \eqref{E:DESMETR} can be written as
\al{DESIT2}{
	\tilde{\ulg}_{ab}  
	=e^{2\Psi}\ulg_{ab}
}
where
\al{CONFFAC}{
	\Psi=
	-\ln\left(\frac{\sin(Ht)}{H}\right)
}
and
\begin{equation}\label{e:cfds}
	\udl{g}_{ab}=-H^2(dt)_a(dt)_b+\ulh_{ab}
\end{equation}
is the \textit{conformal de Sitter metric}. The pair $(\mathcal{M}^n=\bigl(0, \frac{\pi}{H}\bigr)\times \Sigma, \, \ulg_{ab} )$ defines a spacetime that is conformal to the de Sitter spacetime where, by construction, future timelike infinity of the de Sitter spacetime is mapped to the boundary component $\{0\}\times \Sigma$ of $\mathcal{M}^n$.

Next, we define a future directed unit normal to the spatial hypersurface
\begin{equation*}
	\Sigma_t = \{t\} \times \Sigma 
\end{equation*} 
with respect to the conformal metric $\ulg_{ab}$
by
\al{NORMT}{
	\nu_a 
	=H(dt)_a
	\AND
	\nu^a=\ulg^{ab}\nu_b. 
}
Using $\nu^a$, we set 
\al{PRO}{
	\tensor{\ulh}{^a_b}=\tensor{\delta}{^a_b}+\nu^a\nu_b,
}
which we note defines a projection operator that projects onto the $\ulg_{ab}$-orthogonal subspace to the vector field $\nu^a$. We define, in an analogous fashion, the future directed unit normals to the spatial hypersurface $\Sigma_{\tau}$ with respect to the de Sitter metric $\tilde{\ulg}_{ab}$ and the physical metric $\tilde{g}_{ab}$ by
\begin{equation} \label{E:NORMT-a}
	\tilde{\nu}_a=(d\tau)_a\AND \tilde{\nu}^a=\tilde{\ulg}^{ab}\tilde{\nu}_b,
\end{equation}  
and by
\be\label{def-T}
\tilde{T}_a=(-\tilde{\lambda})^{-\frac{1}{2}} \tilde{\nu}_a\AND
\tilde{T}^a = (-\tilde{\lambda})^{-\frac{1}{2}}\tilde{g}^{ab} \tilde{\nu}_b, \quad \text{where} \quad \tilde{\lambda} = \tilde{g}^{ab} \tilde{\nu}_a \tilde{\nu}_b,
\ee
respectively, where  we  note that $(-\tilde{\lambda})^{-\frac{1}{2}}$ is the lapse function of associated to the physical metric $\tilde{g}_{ab}$.
We denote the corresponding spatial projection operators by
\begin{align}\label{e:def-h2}
	\tensor{\tilde{\ulh}}{^c_d} = \tensor{\delta}{^{c}_{ d}} + \tilde{\nu}^c \tilde{\nu}_d,
\end{align}
and
\begin{align}\label{e:def-h2-phy}
	\tensor{\tilde{h}}{^c_d} = \tensor{\delta}{^{c}_{ d}} + \tilde{T}^c \tilde{T}_d,
\end{align}
and we use $\tensor{\tilde{h}}{^c_d}$ and $\tilde{T}^a$ to decompose the physical Yang--Mills curvature $\tilde{F}_{a p}$ into its electric and magnetic components according to 
\begin{align} \label{Et-Ht-def}
	\tilde{E}_b= \tensor{\tilde{h}}{^a_b}  \tilde{F}_{a p} \tilde{T}^p \AND \tilde{H}_{db}=\tensor{\tilde{h}}{^c_d} \tilde{F}_{c a} \tensor{\tilde{h}}{^a_b}.
\end{align}

\begin{remark}
	Our global existence proof relies on two different $3+1$ decompositions. The local existence results from
	\cite{LW2021b} are formulated in the physical picture with $3+1$ decomposition \eqref{def-T} and \eqref{e:def-h2-phy} and in order to apply these results we use the $3+1$ decomposition defined by \eqref{E:NORMT-a}  and \eqref{e:def-h2}. On the other hand, the global estimates established
	in this article are derived using
	the $3+1$ decomposition defined by \eqref{E:NORMT}--\eqref{E:PRO}.
\end{remark}

\subsection{Gauge conditions\label{sec:gauge-conditions}}
Gauge conditions for both the gravitational and Yang--Mills fields play an essential role in our stability proof. For the physical Yang--Mills, we employ an adapted temporal gauge defined by
\begin{equation}\label{phys-temp-gauge}
	\tilde{A}_a \tilde{T}^a = 0
\end{equation}
where the vector field $\tilde{T}^a$ is defined above by \eqref{def-T}, while we employ a wave gauge 
given by
\begin{equation} \label{e:prewg}
	\tilde{Z}^a =0
\end{equation}
for the physical metric where, recall $n$ is the dimension of spacetime, 
\begin{align*}
	\tilde{Z}^{a} &= \tilde{X}^a -  2 (\tilde{\ulg}^{ac}-\tilde{g}^{ac})(d \Psi)_c  + n   (\tilde{g}^{ac}-\tilde{\ulg}^{ac})(d \Psi)_c-(\tilde{g}^{fe}-\tilde{\ulg}^{fe})\tilde{\ulg}^{ac}\tilde{\ulg}_{fe}(d \Psi)_c 
	\intertext{and}
	\tilde{X}^a 
	&=-\tilde{\nb}_e \tilde{g}^{ae}+\frac{1}{2}\tilde{g}^{ae}\tilde{g}_{df}\tilde{\nb}_e\tilde{g}^{df}. 
\end{align*}
Here,  $\tilde{\nb}_e$ denotes the covariant derivative associated to the de Sitter metric $\tilde{\ulg}_{ab}$, see \eqref{E:DESMETR}, and 
$\Psi$ is the scalar function defined previously by  \eqref{E:CONFFAC}.

\subsection{Initial data and the constraint equations}\label{s:dt0}
Solutions to the Einstein--Yang--Mills equations will be generate from the initial data
\begin{equation} \label{idata}
	(\tilde{g}_{ab},\mathcal{L}_{\tilde{\nu}}\tilde{g}_{ab},\tilde{A}_a, \mathcal{L}_{\tilde{\nu}}\tilde{A}_a)|_{\Sigma_0} =  (\acute{g}_{ab},\grave{g}_{ab},\acute{A}_a, \grave{A}_a)
\end{equation}
that is specified on the initial hypersurface $\Sigma_0=\{0\}\times \Sigma$, that is, at time $\tau=0$. As is well known, e.g.~see \cite[Ch.~VI \& VII]{Choquet-Bruhat2009},  this initial data cannot be chosen freely, but must satisfy the following constraint equations:
\begin{align}
	\tilde \nu_a(\tilde{G}^{ab}+\Lambda \tilde{g}^{ab}-\tilde{T}^{ab})|_{\Sigma_0} & =0, \label{E:constraintA} \\
	\tilde{h}^{ab}(\tilde{\nabla}_a \tilde E_{b} + [\tilde A_{a}, \tilde E_{b}])|_{\Sigma_0} & =0,
	\label{E:constraintB}
\end{align}
where $\tilde{h}^{ab}=\tensor{\tilde{\ulh}}{^a_c}\tensor{\tilde{\ulh}}{^b_d}\tilde{g}^{cd}$.
We will always assume that our initial data satisfies these constraints. In addition,
to enforce the gauge conditions, we will further assume that initial data is chosen so that the following gauge constraints hold:
\begin{align}
	\tilde{Z}^a|_{\Sigma_0} & =0, \label{E:constraintC} \\
	\tilde{A}_a \tilde{T}^a|_{\Sigma_0} &= 0. \label{E:constraintD}
\end{align}

\subsection{Main theorem}
We are now in a position to state the main result of this article in the following theorem.

\begin{theorem}\label{t:mainthm}
	Suppose  $\Lambda>0$, $s\in\Zbb_{>\frac{n+1}{2}}$, and the initial data $\acute{g}_{ab}\in H^{s+1}(\Sigma)$, $\grave{g}_{ab}\in H^s(\Sigma)$, $\acute{A}_a \in H^{s}(\Sigma)$ with $\ulh^c{}_a\ulh^d{}_b(d\acute{A})_{cd} \in H^{s}(\Sigma)$, and $\grave{A}_a \in H^{s}(\Sigma)$ satisfy the constraint equations\footnote{See \S\ref{s:norm} for a definition of the $H^s$ norms.} \eqref{E:constraintA}--\eqref{E:constraintD}. Then there exists a constant $\sigma>0$ such that if the initial data 
	satisfy smallness condition	
	\begin{align*}
		\lVert (\tilde{g}^{ab}(0)-\tilde{\ulg}^{ab}(0),\,\tilde{\nb}_{d}\tilde{g}^{ab}(0),\, \tilde{A}_a(0),\, \tilde{E}_{a}(0), \, \tilde{H}_{ab}(0)) \rVert_{H^s} \leq \sigma,
	\end{align*}
	then there exists a unique solution $(\tilde{g}^{ab},\tilde{A}_a)$ to the  Einstein--Yang--Mills equations \eqref{eq-einstein}--\eqref{eq-maxwell-bianchi-o}
	on $[0,\infty)\times \Sigma$ with regularity
	\begin{equation*}
		(\tilde{g}^{ab}, \, \tilde{\nb}_{d}\tilde{g}^{ab}, \tilde{A}_a,\tilde{E}_{a}, \tilde{H}_{ab} ) \in C^0\bigl([0, \infty) ,H^{s}(\Sigma)\bigr) \cap C^1\bigl([0, \infty), H^{s-1}(\Sigma)\bigr)
	\end{equation*}
	that satisfies the initial conditions \eqref{idata}, and the temporal and wave gauge constraints \eqref{phys-temp-gauge}--\eqref{e:prewg} on $[0,\infty)\times \Sigma$. Moreover, there exists a constant $C>0$ such that the estimates
	\begin{align*}
		\lVert \tilde{A}_a(\tau) \rVert_{H^s} + \lVert \tilde{E}_a(\tau) \rVert_{H^s} + \lVert \tilde{H}_{a b}(\tau) \rVert_{H^s}& \leq C\sigma
		\intertext{and}
		\lVert \tilde{g}^{ab}(\tau)-\tilde{\ulg}^{ab}(\tau)\rVert_{H^s}+\lVert \tilde{\nb}_{d}\tilde{g}^{ab}(\tau)\rVert_{H^s }  &\leq C \left(\frac{\pi}{2}-\gd(H^{-1} \tau)\right)^2 \sigma
	\end{align*}
	hold for all $\tau \in [0,\infty)$.
\end{theorem}

\subsection{Prior and related works}
The vacuum de Sitter stability result of Friedrich in $n=4$ spacetime dimensions was extended to all \textit{even} spacetime dimension $n\geq 4$ in \cite{Anderson2005}; see also \cite{kaminski2021} where a gap in this stability proof is addressed. It is worth noting that the de Sitter stability proofs from \cite{Anderson2005,Friedrich1986} and the related Einstein--Yang--Mills stability proof from \cite{Friedrich1991} rely on specific conformal representations of the Einstein field equations. It is because these conformal representations only exist in certain spacetime dimensions that leads to the dimension dependent restrictions in these stability proofs. 

The first de Sitter stability result to extend to all spacetime dimension $n\geq 4$ was established in \cite{Ringstroem2008}. In that article, the future global existence of solutions to the Einstein--scalar field system that are generated from initial data that is sufficiently close to de Sitter intial data was established for a class of potentials that includes the vacuum Einstein equations with a positive cosmological constant. This stability result was later generalized to the Einstein--Maxwell--scalar field equations in \cite{Svedberg2011}.

Related stability results for perfect fluid Friedmann-Lema\^{\i}tre-Robertson-Walker (FLRW) spacetimes  with a postive cosmological constant that are asymptotic to a flat spatial slicing of de Sitter were obtained in the articles \cite{Hadzic2015,Liu2018,Liu2018b,Liu2018a, LeFloch2021,Luebbe2013,Oliynyk2016a, RodnianskiSpeck:2013,Speck2012}. We also note the related stability results from \cite{FOW:2021,FW-CPDE20,Wang2019} for FLRW spacetimes with non-accelerated expansion. 

The temporal gauge for the Yang-Mills equations, which is essential for arguments presented in this and the companion article \cite{LW2021b}, is well-known and has been employed in prior studies of the Yang-Mills equations; for example, 
Eardley-Moncrief \cite{Eardley1982, Eardley1982-global} and Ginibre-Velo \cite{Ginibre1981} used the temporal gauge to deduce first order equations for the Yang--Mills equations in Minkowski spacetime. However, all of the previous first order formulations of the Yang-Mills equations in
the temporal gauge were not symmetric hyperbolic
and the hyperbolicity of these formulations was not evident. A first order symmetric hyperoblic formulation of the Yang--Mills equations in the temporal gauge was first derived in the companion article \cite{LW2021b}. In the present article, we rely on a similar formulation to cast the Einstein--Yang--Mills equations as a symmetric first order system.

\subsection{Overview and proof strategy}
\subsubsection{Main idea}
The main idea behind the proof of Theorem \ref{t:mainthm}, which was first employed in \cite{Oliynyk2016a} for a different matter model, is to formulate the initial value problem for a gauge reduced version of the Einstein--Yang--Mills equations as a Fuchsian initial value problem of the form
\begin{align}
	\mathbf{A}^a(t,\mathbf{U})\nb_a \mathbf{U}&=\frac{1}{t}\mathfrak{A}(t,\mathbf{U})\mathbb{P}\mathbf{U}+ G(t,\mathbf{U}) \quad &&\text{in }\Bigl(0,\frac{\pi}{2H}\Bigr]\times \Sigma, \label{e:modeq0}\\
	\mathbf{U}&=\mathbf{U}_0   \quad &&\text{in }\Bigl\{\frac{\pi}{2H}\Bigr\} \times \Sigma,   \label{e:moddt0}
\end{align}
on the conformal de Sitter spacetime manifold $(\mathcal{M}^n,\ulg_{ab})$, where we recall that $t=0$ corresponds to future timelike infinity and $t=\frac{\pi}{2H}$ corresponds to $\tau=0$. The advantage of this reformulation is once it is shown that the coefficients of the Fuchsian equation \eqref{e:modeq0} satisfy certain structure conditions, see \S\ref{s:verif} for details, then, under a suitable smallness assumption on the initial data, the existence of solutions
to \eqref{e:modeq0}-\eqref{e:moddt0} on the spacetime domain $\bigl(0,\frac{\pi}{2H}\bigr]\times \Sigma \subset \mathcal{M}^n$ follows from an application of Theorem \ref{t:glex} in \S\ref{s:glbex}, which is an adapted version Theorem 3.8 from \cite{Beyer2020}.

For technical reasons, we do not formulate the gauge reduced Einstein--Yang--Mills equations as a Fuchsian system. Instead, we show in \S\ref{s:confflsm}, see, in particular, Theorems \ref{thm-FOSHS-m} and \ref{thm-FOSHS-Maxwell}, that solutions of Einstein--Yang--Mills equations that satisfy a temporal and wave gauge condition yield solutions to a Fuchsian equation of the form \eqref{e:modeq0}; see \eqref{e:FchEYM} for the actual equation.
We then appeal to the local-in-time existence theory for the Einstein--Yang--Mills equations from the companion article \cite{LW2021b} to obtain local-in-time solutions to the Fuchsian \eqref{e:modeq0} 
on a spacetime domain of the form $\bigl(t_*,\frac{\pi}{2H}\bigr]\times \Sigma \subset \mathcal{M}^n$  for some $t_*\in \bigl[0,\frac{\pi}{2H}\bigr)$ 
that satisfies a given initial condition $\mathbf{U}=\mathbf{U}_0$ at $t=\frac{\pi}{2H}$. It is important to note that de Sitter initial data at $\tau=0$ corresponds to the trivial initial data $\mathbf{U}_0=0$, because we are then free to choose $\mathbf{U}_0$ as small as we like  since we are only interested in perturbations of de Sitter initial data.

Now, the time of existence $t_*$ from the local-in-time existence theory may be strictly greater than zero. In order to show that $t_*=0$, which would correspond to solutions to the Einstein--Yang--Mills equations that exist globally to the future, we alternately view the initial value problem \eqref{e:modeq0}--\eqref{e:moddt0} as a stand alone system that, a priori, admits solutions that are not derived from solutions of the reduced Einstein--Yang--Mills equations. The point of doing so is that, under a suitable smallness assumption on the initial data $\mathbf{U}_0$,  we can apply Theorem \ref{t:glex} to conclude the existence of solutions to \eqref{e:modeq0} on $\bigl(0,\frac{\pi}{2H}\bigr]\times \Sigma$ that are uniformly bounded. By uniqueness, this global solution coincides on $\bigl(t_*,\frac{\pi}{2H}\bigr]\times \Sigma$  with the local-in-time solution to \eqref{e:modeq0} obtained from a local-in-time solution to the Einstein--Yang--Mills equations as discussed above. In this way, we obtain uniform bounds on the local solution. Because of these bounds, we can then appeal to the continuation principle from \cite{LW2021b} to extend the local-in-time solution past $t_*$, and in particular, all the way to $t_*=0$, which yields the global existence of solutions to the Einstein--Yang--Mills equations. 

The above arguments constitute the main steps in the proof of Theorem \ref{t:mainthm}, the main result of this article. The complete proof of the theorem is given in \S\ref{mainthm-proof}.

\subsubsection{Additional remarks}
In the companion article \cite{LW2021b}, the first and third authors formulated the Einstein--Yang--Mills equations in a temporal gauge as a first order symmetric hyperbolic system. This was achieved through the introduction of auxiliary variables and the use of symmetrising tensors. The temporal gauge employed in \cite{LW2021b} was adapted to the physical spacetime metric and this choice of gauge was particularly useful for establishing the equivalence between the Einstein--Yang--Mills equations and the symmetric first order system. In this article \cite{LW2021b}, the adapted temporal gauge is defined by \eqref{phys-temp-gauge}.

While the adapted temporal gauge \eqref{phys-temp-gauge} is useful for establishing the local-in-time existence of solutions, it is not suitable for analysing the long time stability of solutions. 
Instead, for the long time stability analysis, we employ the temporal gauge defined by \eqref{temporal}. Fortunately, as we show in this article, the first order symmetric hyperbolic formulation of the Einstein--Yang--Mills equations from  \cite{LW2021b} is robust and can be adapted to the new choice of temporal gauge. However, it should be noted that symmetrizing the Einstein--Yang--Mills equations in the temporal gauge \eqref{temporal} is more difficult due to the fact that in this gauge, unlike the adapted gauge \eqref{phys-temp-gauge}, the propagation and constraint equations of the Yang--Mills equations do not decouple; see Section \ref{sec-Max} for details. Importantly, we can also establish suitable estimates on the gauge transformation that connects the two gauges, see Section \ref{s:maprf} for details, which allows us employ both gauges in our stability proof.

In our view, the key innovation of this article is the use of the temporal gauge \eqref{temporal} and a conformal compactification of spacetime to formulate the Einstein--Yang--Mills equations as a Fuchsian system whose coefficients satisfy structure conditions that make it suitable for establishing the long time stability of solutions. It is worth noting that the use of the temporal gauge seems to be crucial to being able to formulate the Einstein--Yang--Mills equations as a Fuchsian system. Other gauge choices such as the Lorentzian gauges do not seem to lead to a Fuchsian formulation that is suitable for the long time analysis of solutions.

\subsubsection{Fuchsian fields}\label{sec-Fuchsian-field}

The Fuchisan formulation \eqref{e:FchEYM} of the gauge reduced Einstein--Yang--Mills equations is based on a particular choice of fields. To define these fields, we first replace the physical fields with conformally rescaled fields defined by
\begin{align}
	g_{ab}={}&e^{-2\Psi}\tilde{g}_{ab},  \label{CONFG-g} \\
	F_{ab}={}&e^{-\Psi} \tilde{F}_{ab}, \label{CONFG-F} \\
	A_a ={}&e^{- \frac{\Psi}{2} } \Ab_a, \label{CONFG-A}
\end{align}	
where $\Psi$ as previously defined above by \eqref{E:CONFFAC}. These conformal fields should be viewed as living on the conformal de Sitter spacetime $\mathcal{M}^{n}=\bigl(0,\frac{\pi}{H}\bigr)\times \Sigma$. For use below, we note that
\eqref{CONFG-g}--\eqref{CONFG-F} imply the relations 
\al{CONFG2}{
	g^{ab}=e^{2\Psi}\tilde{g}^{ab},\quad F^{ab}=e^{3\Psi} \tilde{F}^{ab} \AND \tensor{F}{^a_b} =e^{\Psi}\tensor{\tilde{F}}{^a_b} .
}

Next, we decompose the conformal metric $g_{ab}$ as 
\begin{align}\label{decom-g}
	\lambda=g^{ab}\nu_a\nu_b, \quad \xi^c=g^{ab}\nu_a \tensor{\ulh}{^c_b} \AND h^{ab}=\tensor{\ulh}{^a_c} \tensor{\ulh}{^b_d} g^{cd},
\end{align}
where we recall that $\nu_a$ is defined above by \eqref{E:NORMT}. Inspired by the Fuchsian formulation of the Einstein--Euler equations from \cite{Oliynyk2016a, Liu2018b},
we introduce similar variables to parameterize the gravitational
fields. As was first observed in \cite{Oliynyk2016a}, this choice of variables in conjunction with the wave gauge \eqref{e:prewg} allows the Einstein equations to be formulated as a Fuchsian system of equations whose coefficients satisfy
specific structural conditions that  make it possible to derive global estimates on solutions. 
To define these new variables, we first set
\begin{gather}
	\ulh^{ab} = \ulh^{a}{}_{c}\ulh^{b}{}_{d}\ulg^{cd}, \quad	S =\frac{\alpha}{\ula}, \quad \alpha=(\det{(h^{ab}+\nu^a\nu^b)})^{\frac{1}{n-1}}, \quad \ula =(\det{(\ulh^{ab}+\nu^a\nu^b)})^{\frac{1}{n-1}}, \label{e:S}\\
	q= {} \lambda+1+(3-n) \ln S \AND
	\mathfrak{h}^{ab} ={}  \frac{1}{S} h^{ab}. \label{E:q}
\end{gather}  
We then define the Fuchsian gravitational field variables via 
\ali{
	m={}&\frac{1}{\tta t}(\lambda+1), \label{E:W} \\
	p^a={}&\frac{1}{\ttb t} \xi^a, \label{E:V} \\
	m_d= {}&\udn{d}\lambda - \frac{1}{\ttj H t}(\lambda+1) \nu_d,   \label{E:WD}\\
	\tp{a}{d}={}&  \udn{d} \xi^a - \frac{1}{\ttk H t}\xi^a \nu_d,  \label{E:VD}\\
	s^{ab}={}&\mathfrak{h}^{ab}-\ulh^{ab},  \label{E:U}\\
	\tss{ab}{d}={}&\udn{d}(\mathfrak{h}^{ab}-\ulh^{ab}), \label{E:UD}\\
	s={}&q, \label{E:Q}\\
	s_d={}&\udn{d} q, \label{E:QD}
} 
where $\tta, \ttb, \ttj, \ttk$ are constants.  The freedom to choose these constants will be used below to eliminate problematic $1/t$ singular terms in the evolution equations. 
We also find it useful in subsequent calculations to let $h_{ab}$ denote the unique symmetric tensor field satisfying 
\begin{equation*}
	h_{ab}\nu^b=0 \AND h^{cb}h_{ab}=\tensor{\ulh}{^c_a}.
\end{equation*}

For the conformal Yang--Mills fields, we employ a \textit{temporal gauge} defined by
\begin{equation}\label{e:temgg}
	A_a \nu^a = 0,
\end{equation}
and we set
\be\label{decom-F}
\bar A_b = A_a \tensor{\ulh}{^a_b} \AND E_b=- \nu^p F_{p a} \tensor{\ulh}{^a_b}.
\ee	
Due to the temporal gauge, the one form 	$\bar A_b$ completely determines the gauge potential while $E_b$ defines the conformal electric field associated to the splitting determined by the vector field $\nu^a$. In order to express the Yang--Mills equations as a first order symmetric hyperbolic system, we find it necessary to introduce the additional fields
\begin{align}
	\tE^a= {}&
	-h^{a b} E_{b}, \label{def-tE-1} \\
	H_{db}= {}&\tensor{\ulh}{^c_d} F_{c a} \tensor{\ulh}{^a_b}, \label{def-MYM}
\end{align}
where we note that $H_{db}$ is the magnetic field associated to the splitting determined by the vector field $\nu^a$.  The field $H_{ab}$ is used to express the Yang--Mills equations as a first order system \eqref{e:maineq1} in the variables $\bar A_a, \, E_b, \, H_{a b}$, which is non-symmetric. The field $\tE^a$ is then extracted from the non-symmetric part of \eqref{e:maineq1}, which allows us to rewrite the Yang--Mills equation as a symmetric first order system \eqref{Maxwell-FOSHS-1} in the variables $\bar A_a, \, E_b, \, \tE^c, \, H_{a b}$. 
For use in calculations below, we record the identity
\begin{equation}\label{e:Hpq}
	H_{p q} = \sqrt{\frac{\sin(Ht)}{H}} \tensor{\ulh}{^a_p} \nb_a \bar A_q - \sqrt{\frac{\sin(Ht)}{H}} \tensor{\ulh}{^a_q} \nb_a \bar A_p + [\bar A_p, \bar A_q],
\end{equation}
which is easily derived from \eqref{def-F},  \eqref{E:CONFFAC}, \eqref{E:PRO}, \eqref{CONFG-F}--\eqref{CONFG-A}, \eqref{decom-F} and \eqref{E:CAL1}.

We then collect the above fields into the single vector
\be\label{def-U}
\mathbf{U}=(m,\, p^a,\, m_d, \, \tp{a}{d}, \, s^{ab}, \, \tss{ab}{d}, \, s, \, s_d, \,\tE^a, \, E_b, \, H_{db}, \, \bar A_d).
\ee
It is this collection of fields that we will use to transform the gauge reduced Einstein--Yang--Mills equations into Fuchsian form.

\subsubsection{Derivation of the Fuchsian equation}
The majority of this article is devoted to the derivation of the Fuchsian formulation of the gauge reduced conformal Einstein--Yang--Mills equations given by
\eqref{e:FchEYM}. The derivation is lengthy but computationally straightforward and it is split into two parts. The first part concerns the derivation of a Fuchsian formulation of the Einstein equations in the wave gauge \eqref{e:prewg} and it is carried out in \S\ref{sec-Einstein} as well as in the appendix \S\ref{lem-FOSHS-p}. The resulting Fuchsian equations are displayed in Theorem \ref{thm-FOSHS-m}. The second part of the derivation is given in \S\ref{sec-Max}. There a Fuchsian formulation of the Yang--Mills equations in the temporal gauge \eqref{e:temgg} is obtained and the equations are displayed in Theorem \ref{thm-FOSHS-Maxwell}. We then combine the equations from Theorems \ref{thm-FOSHS-m} and \ref{thm-FOSHS-Maxwell} to obtain the Fuchsian equation \eqref{e:FchEYM}.

\subsection{Outlines}
The paper is organized as follows. In \S\ref{sec:prelim}, we fix our notation and set out our conventions.
In \S\ref{s:confflsm}, we present a  Fuchsian formulation for the conformal  Einstein--Yang--Mills system for any dimension $n\geq4 $.
In \S\ref{sec-Model} and \S\ref{s:verif}, we analyse the Fuchsian formulation of conformal Einstein--Yang--Mills system, deriving all the necessary estimates and in particular improved estimates for the Yang--Mills field. We show in \S\ref{s:maprf} the uniform bound of the gauge transformation between the different gauges for local existence and for global stability. In the end, we conclude global stability in \S\ref{mainthm-proof}. More supporting materials and calculations are given in the appendices \S\ref{s:App1}-\S\ref{sec-matrix}. An index of notation is presented in Appendix \S\ref{s:index}.

\section{Preliminaries\label{sec:prelim}}
In this section, we set out some conventions and notation that we will employ throughout this article.

\subsection{Abstract indices and tensor conventions}\label{s:AIN}
For tensors, we employ abstract index notation, e.g. see \cite[\S 2.4]{Wald2010}. Physical fields will be distinguished with a tilde, e.g. $\tilde{g}_{ab}$, while their conformal counterpart will be denoted with the same letter but without the tilde, e.g. $g_{ab}$. 
Underlined fields with a tilde, e.g. $\tilde{\underline{g}}_{ab}$, will refer to background fields that are associated with de Sitter spacetime. In line with the above notation, underlined  fields without a tilde, e.g.  $\underline{g}_{ab}$, will refer to background fields associated with the conformal de Sitter spacetime. In particular, we use $\tilde{\nb}$ 
and $\nb$ to denote the covariant derivatives associated to the background metrics $\tilde{\underline{g}}_{ab}$ and $\underline{g}_{ab}$, respectively, and likewise we use $\tensor{\tilde{\udl{R}}}{_{cde}^a}$ and  $\tensor{\udl{R}}{_{cde}^a}$ to denote the associated curvature tensors. The other curvature tensor will be denoted using similar notation, e.g. $\tilde{\udl{R}}_{ab}$ and $\udl{R}_{ab}$ for the Ricci curvature tensors.

As indicated above, we use  $\tilde{\nabla}$ to denote the covariant derivative associated to the physical metric $\tilde{g}_{ab}$, and in line with our conventions,  we use $\nabla$ to denote the covariant derivative associated to the conformal metric $g_{ab}$. We also use $\tensor{\tilde{R}}{_{cde}^a}$ and  $\tensor{R}{_{cde}^a}$ to denote curvature tensors of  $\tilde{g}_{ab}$ and $g_{ab}$, respectively, and a similar notation for the other curvature tensors. We also use $\Box = g^{a b} \nabla_a \nabla_b$ to denote the wave operator associated to the conformal metric $g_{a b}$.

Unless indicated otherwise, we use the physical metric $\tilde{g}_{ab}$ to raise and lower indices on physical tensors that are \textit{not} background tensors,
and the conformal  metric $g_{ab}$ to raise and lower tensor indices on tensors built out of the conformal fields, but are  not identically background tensors. For background tensors associated with the de Sitter spacetime, we use the de Sitter metric $\tilde{\ulg}_{ab}$ to raise and lower indices, and correspondingly, we use the conformal background metric $\ulg_{ab}$ to raise and lower indices on background tensor fields associated with the conformal de Sitter spacetime.

It is worth noting at this point that the one forms $\tilde{\nu}_a$ and $\nu_a$, which are derived from the background slicing of the de Sitter spacetime and its conformal counterpart, see \eqref{E:NORMT} and \eqref{E:NORMT-a}, are not underlined and constitute an exception to our conventions. Any other exceptions to the above conventions will be clearly indicated when they occur.

\subsubsection{Index of notation} A list of frequently used definitions and notation is provided in Appendix \ref{s:index}.

\subsubsection{Index brackets}
Round and square brackets on tensor indices are used to identify the symmetric and anti-symmetric, respectively, components of a tensor. For example,
\begin{align*}
	A_{[a b c]} =   \frac{1}{6}(A_{a b c} + A_{b c a} + A_{c a b} -A_{a c b}-A_{b a c}-A_{c b a}) \AND
	B_{(a b)} =  \frac{1}{2}(B_{a b} + B_{b a}).
\end{align*}

\subsubsection{Connection coefficients}\label{s:conx}
The covariant derivatives $ \nabla_a\tensor{T}{^{b_1\cdots b_k}_{c_1\cdots c_l}}$ and $\nb_a\tensor{T}{^{b_1\cdots b_k}_{c_1\cdots c_l}}$  are related via
\begin{equation*}
	\nabla_a  \tensor{T}{^{b_1\cdots b_k}_{c_1\cdots c_l}}=\udn{a}  \tensor{T}{^{b_1\cdots b_k}_{c_1\cdots c_l}}+\sum_i \tensor{X}{^{b_i}_{ad}}\tensor{T}{^{b_1\cdots d \cdots b_k}_{c_1\cdots c_l}}-\sum_j \tensor{X}{^d_{ac_j}}\tensor{T}{^{b_1\cdots b_k}_{c_1\cdots d \cdots c_l}}
\end{equation*}
where $\tensor{X}{^a_{bc}}$ is defined by
\begin{equation*}
	\tensor{X}{^a_{bc}}
	=  -\frac{1}{2}\bigl(g_{ec}\udn{b}g^{ae}+g_{be}\udn{c}g^{ae}-g^{ae}g_{bd}g_{cf}\udn{e}g^{df}\bigr). 
\end{equation*}
Contracting $\tensor{X}{^a_{bc}}$ on the covariant indices yields the vector field 
\be\label{def-X}
X^a :=  g^{bc} \tensor{X}{^a_{bc}}=  -\udn{e}g^{ae}+\frac{1}{2} g^{ae}g_{df}\udn{e}g^{df},
\ee
which plays an important role in defining the wave gauge.

\subsubsection{Sobolev norms for spacetime tensors}\label{s:norm}
To define the Sobolev norms for spacetime tensors that are employed in this article, we
first consider the special case of a rank 2 covariant tensor field $S_{ab}$ on a spacetime manifold of the
form $I\times \Sigma$, where $I\subset \mathbb{R}
$ is an interval and, as above, $\Sigma = \mathbb{S}^{n-1}$. Assuming that $t$ is a Cartesian coordinate on $I$, we interpret the associated coordinate vector field $\partial_t$ on $I$ as defining a vector field $(\partial_t)^a$ on $I\times \Sigma$ that satisfies $(\partial_t)^a (dt)_a =1$. We then decompose $S_{ab}$ into a scalar $(\partial_t)^a S_{ab}(\partial_t)^b$, a spatial one form $\ts{\ulh}{^b_a} S_{bc}(\partial_t)^c$, and a spatial tensor field 
$\ulh^c{}_a S_{cd}\ulh^d{}_b$, where, as above, $\ulh_{ab}$ is the standard metric on $\mathbb{S}^{n-1}$, which we interpret as a tensor field on $I\times \Sigma$. Now since $(\partial_t)^a S_{ab}(\partial_t)^b$ is a scalar field, and the tensor fields $\ulh^c{}_a S_{cd}\ulh^d{}_b$ and  $\ulh^b{}_a S_{bc}(\partial_t)^c$ are purely spatial, we can, by restricting them to $\Sigma_t = \{t\}\times \Sigma$, naturally interpret these as tensor fields on the Riemannian manifold $(\Sigma_t,\ulh_{ab})$, which are all trivially isometric for different values of $t$. For $1\leq p\leq \infty$ and $s\in \Zbb_{\geq 0}$, the $W^{s,p}$ Sobolev norm of  $S_{ab}$  is then  defined by 
\begin{align*}
	\norm{S_{ab}(t)}_{W^{k,p}} = & \norm{(\partial_t)^a S_{ab}(\partial_t)^b|_{\Sigma_t}}_{W^{k,p}(\Sigma_t)}+\norm{\ulh^b{}_a S_{bc}(\partial_t)^c|_{\Sigma_t}}_{W^{k,p}(\Sigma_t)} \notag \\
	&+ \norm{\ulh^c{}_aS_{cd}\ulh^d{}_b|_{\Sigma_t}}_{W^{k,p}(\Sigma_t)}
\end{align*}
for $t\in I$ where the Sobolev norms $\norm{\cdot}_{W^{k,p}(\Sigma_t)}$ on the Riemannin manifolds $(\Sigma_t,\ulh_{ab})$ are defined in the usual way, e.g. see \cite[Ch.~2]{Aubin1998}. We also define
\begin{equation*}
	\norm{S_{ab}}_{L^\infty(I,W^{s,p})} = \sup_{t\in I}\norm{S_{ab}(t)}_{W^{s,p}}.
\end{equation*}

By identifying $\Sigma_t$ with $\Sigma$ via the isometry $\Sigma \ni x \longmapsto (t,x)\in \Sigma_t$, we can view $(\partial_t)^a S_{ab}(\partial_t)^b$, $\ulh^c{}_a S_{cd}\ulh^d{}_b$ and  $\ulh^c{}_a S_{cd}\ulh^d{}_b$ as time-dependent tensor fields on the fixed Riemmanian manifold $(\Sigma,\ulh_{ab})$. We then have 
\begin{equation*}
	S_{ab}\in C^\ell(I,W^{s,p}(\Sigma))
\end{equation*}
provided each of the maps
$I\ni t \longmapsto (\partial_t)^a S_{ab}(\partial_t)^b|_{\Sigma_t}  \in W^{s,p}(\Sigma)$, 
$I\ni t \longmapsto \ulh^c{}_a S_{cd}\ulh^d{}_b|_{\Sigma_t}  \in W^{s,p}(\Sigma)$
and
$I\ni t \longmapsto \ulh^c{}_a S_{cd}\ulh^d{}_b|_{\Sigma_t} \in W^{s,p}(\Sigma)$
are $\ell$-times continuously differentiable with respect to $t$.

The $W^{k,p}$ norms for general spacetime tensor fields can be defined in a similar fashion.

\section{Conformal Einstein--Yang--Mills equations}\label{s:confflsm}

The aim of this section is to transform the conformal  Einstein--Yang--Mills equations into Fuchsian form in a step by step fashion. The Fuchsian formulation of the Yang--Mills equations derived below is new while the Einstein equations are transformed into Fuchsian form following the method used in the articles\footnote{In this article, the conformal factor, e.g. $\Psi$ (defined by \eqref{E:CONFFAC}, i.e., $\Psi=
	-\ln(\frac{\sin(Ht)}{H} )$), is different but has the same asymptotic properties as $\Psi=
	-\ln t$, i.e. $\sin(Ht)\sim \tan(Ht)\sim Ht$ near $t=0$. It is because of this that the same methods continue to work.}  \cite{Oliynyk2016a,Liu2018,Liu2018b,Liu2018a, LeFloch2021,Oliynyk2021b}.

\subsection{Fuchsian  formalism of the reduced conformal Einstein equations}\label{sec-Einstein}	
We begin the transformation of the Einstein equations \eqref{eq-einstein} into Fuchsian form by contracting them with $\tilde{g}^{ab}$ to get
\begin{equation}\label{e:R}
	\tilde{R}=\frac{2}{n-2}(n\Lambda-\tilde{T}),
\end{equation}
where we are using $\tilde{T}=\tilde{g}^{cd}\tilde{T}_{cd}$ to denote the trace of the stress-energy tensor.
Inserting \eqref{e:R} into \eqref{eq-einstein} yields
\begin{align}\label{e:tRab}
	\tilde{R}_{ab}-\frac{2}{n-2}\Lambda\tilde{g}_{ab}=\tilde{T}_{ab}-\frac{1}{n-2}\tilde{T}\tilde{g}_{ab}.
\end{align}
With the help of \eqref{E:H} and \eqref{e:tRab}, a straightforward calculation
using well-known conformal transformation rules (e.g. \cite[Appendix VI]{Choquet-Bruhat2009}) then shows that the Einstein equations \eqref{eq-einstein} transform under the conformal  change of variable 
\eqref{CONFG-g}--\eqref{CONFG-F}
into
\al{CONFEIN2}{
	& R^{ab} = (n-2)(\nabla^a \nabla^b \Psi-\nabla^a \Psi \nabla^b \Psi)+\left(\Box \Psi+(n-2)\nabla^c\Psi\nabla_c\Psi+\frac{n-1}{H^2} e^{2\Psi}\right)g^{ab}  \notag   \\
	& \hspace{5cm} + T^{ab}-\frac{1}{n-2} T g^{ab},
}
where $T_{a b} = F_{a}{}^{ c} F_{b c} - \frac{1}{4} g_{a b}  F^{c d} F_{c d}$, $T^{a b} =g^{ac} g^{bd} T_{cd}$, $T=g_{cd}T^{cd}$, and we note that $T_{a b} = \tilde T_{a b}$ and $T=e^{2\Psi}\tilde{T}$ by \eqref{CONFG-g}--\eqref{CONFG-F} and \eqref{E:CONFG2}.

\subsubsection{The reduced Einstein equations}\label{s:REE}
The next step in transforming the Einstein equations into Fuchsian form is to select a \textit{wave gauge} that will allow us to formulate the Einstein equations as a symmetric hyperbolic system and eliminate problematic singular terms.
To this end, we set
\al{WAVEGA}{
	Z^a =X^a+Y^a
}
where $X^a$ is as defined above by \eqref{def-X} and\footnote{These choices for $Y^a$, $\eta^a$, and the tensor $\ts{A}{^{ab}_c}$ defined below are made in order to eliminate problematic $1/t$ singular terms in the conformal Einstein equations. This type of gauge fixing was first introduced in \cite{Oliynyk2016a}. }
\al{XYZ}{
	Y^a =   -(n-2)\nabla^a\Psi+\eta^a  \quad\text{with}\quad
	\eta^a = (n-2)  \udn{}^a\Psi
	=-\frac{n-2}{\tan(Ht)}\nu^a.
}
The \textit{wave gauge} that we employ for the conformal Einstein equations is then defined by the vanishing of the vector field \eqref{E:WAVEGA}, that is,
\al{CONSTR1}{Z^a=0.}
This gauge choice is preserved by the evolution and so we only need to choose initial data so that $Z^a$ vanishes on the initial hypersurface to ensure that it vanishes throughout the evolution. We will refer to the conformal Einstein equations in the wave gauge \eqref{E:CONSTR1} as the \textit{reduced conformal Einstein equations}.

\begin{lemma}\label{t:rdein}
	The reduced conformal Einstein equations are given by
	\al{CONFEIN5}{
		&\frac{1}{2}g^{cd}\udn{c}\udn{d} g^{ab} +\udl{R}^{ab} +P^{ab}+ Q^{ab} +\frac{1}{n-2}X^a X^b -(n-2)\nu_c g^{c(a}\nu^{b)}   \nnb\\
		& +\frac{n-2}{2\tan^2(Ht)}\nu^a  (g^{bc}-\ulg^{bc}) \nu_c  +\frac{n-2}{2\tan^2(Ht)}\nu^b  (g^{ac}-\ulg^{ac}) \nu_c    \nnb\\ = {} &\frac{n-2}{2\tan(Ht)}\nu^c\udn{c}g^{ab} + \left( \frac{\lambda+1}{\sin^2(Ht)} + (n-2)\right)g^{ab} +g^{bd}F^{ac}F_{dc}-\frac{1}{2(n-2)}g^{ab}F^{cd}F_{cd}.
	}
	where $P^{ab}$ (linear and quadratic terms in $g^{bc}-\ulg^{bc}$) are $Q^{ab}$ (quadratic terms in $\udn{e} g^{bd}$) are defined below by \eqref{E:RICCI}. 
\end{lemma}
Before proving this lemma, we first observe that corollary below follows from setting $g^{ab}=\ulg^{ab}$, $\lambda=-1$, and $F_{cd}=0$ in \eqref{E:CONFEIN5}.

\begin{corollary}\label{t:bkgdR}
	The Ricci tensor of the conformal de Sitter metric $\udl{g}_{ab}$ is given by
	\begin{equation*}
		\udl{R}^{ab}=(n-2)\ulh^{ab}
	\end{equation*}
	and satisfies the relations
	\begin{equation*}\label{E:RICCIDES2}
		\udl{R}^{ab}\nu_a\nu_b=0, \quad \udl{R}^{ab}\nu_a\ts{\ulh}{^e_b}=0 \AND \udl{R}^{ab}\ts{\ulh}{^e_a}\ts{\ulh}{^f_b}=(n-2)\ulh^{ef}.
	\end{equation*}
\end{corollary}

\begin{proof}[Proof of Lemma \ref{t:rdein}]
	The reduced Einstein equations are obtained from the Einstein equations \eqref{E:CONFEIN2} by adding the term $-\nabla^{(a} Z^{b)}-\frac{1}{n-2}\ts{A}{^{ab}_c} Z^c$ that vanishes when the wave gauge $Z^a=0$ holds. The resulting equations are given by
	\al{CONFEIN2-reduced}{
		R^{ab}-\nabla^{(a} Z^{b)}-\frac{1}{n-2}\ts{A}{^{ab}_c} Z^c= {} & T^{a b} -  \frac{1}{n-2} g^{a b} g^{c d} T_{c d}+ (n-2)(\nabla^a\nabla^b\Psi-\nabla^a\Psi\nabla^b\Psi) \notag \\
		& +\left(\Box \Psi+(n-2)\nabla^c\Psi\nabla_c\Psi+\frac{n-1}{H^2} e^{2\Psi}\right)g^{ab},
	}
	where $\ts{A}{^{ab}_c}$ is defined by
	\als{
		\ts{A}{^{ab}_c}=-X^{(a}\tensor{\delta}{^{b)}_c}+Y^{(a}\tensor{\delta}{^{b)}_c},
	}
	and we note that
	\gas{
		-\frac{1}{n-2}\ts{A}{^{ab}_c} Z^c=\frac{1}{n-2}X^a X^b-\frac{1}{n-2}Y^a Y^b  \AND
		-\nabla^{(a}Z^{b)}=-\nabla^{(a}X^{b)}-\nabla^{(a}Y^{b)}.
	}
	Recalling the decomposition \eqref{E:RICCI} for the conformal Ricci tensor from Lemma \ref{t:Rdop} and noting the identity 
	\begin{equation*}
		T^{a b} -  \frac{1}{n-2} g^{a b}  g^{c d} T_{c d}
		= {}  g^{bd}F^{ac}F_{dc}-\frac{1}{2(n-2)}g^{ab}F^{cd}F_{cd},
	\end{equation*}
	it is then straightforward to verify that 
	reduced conformal Einstein equations \eqref{E:CONFEIN2-reduced} can be expressed as
	\al{CONFEIN1}{
		&\frac{1}{2}g^{cd}\udn{c}\udn{d}g^{ab} + \udl{R}^{ab} + P^{ab}+ Q^{ab} +\frac{1}{n-2}X^a X^b -\nabla^{(a} Y^{b)}-\frac{1}{n-2}Y^a Y^b   \nnb\\ ={} & (n-2)(\nabla^a\nabla^b\Psi-\nabla^a\Psi\nabla^b\Psi)+\left(\Box \Psi+(n-2)\nabla^c\Psi\nabla_c\Psi+\frac{n-1}{H^2}\frac{H^2}{\sin^2(Ht)}\right)g^{ab} \notag  \\
		&+g^{bd}F^{ac}F_{dc}-\frac{1}{2(n-2)}g^{ab}F^{cd}F_{cd},
	}
	where $P^{ab}$ and $Q^{ab}$ are defined below by \eqref{E:RICCI}. 
	But by \eqref{E:XYZ}, we observe with the help of the identity \eqref{E:RELBABLA} from Lemma \ref{t:nbnuh} that
	\als{
		-\nabla^{(a}Y^{b)}=(n-2)\nabla^a \nabla^b \Psi-\frac{n-2}{\sin^2(Ht)}\nu_c g^{c(a}\nu^{b)}-\frac{n-2}{2\tan(Ht)}\nu^c\udn{c}g^{ab}.
	}
	Using this expression and \eqref{E:XYZ} allows us to write \eqref{E:CONFEIN1} as
	\als{
		&\frac{1}{2}g^{cd}\udn{c}\udn{d} g^{ab} +\udl{R}^{ab} +P^{ab}+ Q^{ab} +\frac{1}{n-2}X^a X^b +(n-2)\nabla^a \nabla^b \Psi-\frac{n-2}{\sin^2(Ht)}\nu_c g^{c(a}\nu^{b)} \nnb\\
		&-\frac{n-2}{2\tan(Ht)}\nu^c\udn{c}g^{ab} -\frac{1}{n-2}(-(n-2)\nabla^a \Psi+\eta^a)(-(n-2)\nabla^b\Psi+\eta^b)   \nnb\\ = {}& (n-2)(\nabla^a\nabla^b \Psi-\nabla^a \Psi\nabla^b\Psi)+\left(\Box \Psi+(n-2)\nabla^c \Psi\nabla_c\Psi+\frac{n-1}{H^2}\frac{H^2}{\sin^2(Ht)}\right)g^{ab} \notag  \\
		&+g^{bd}F^{ac}F_{dc}-\frac{1}{2(n-2)}g^{ab}F^{cd}F_{cd}.
	}
	Finally, using the identity \eqref{E:CAL3} from Lemma \ref{t:conf2} and the relation 
	\als{
		X^a\nu_a=-Y^a\nu_a=-\frac{n-2}{\tan(Ht)}(\lambda + 1), 
	}
	the above formulation of the reduced Einstein equations is easily seen to be equivalent to \eqref{E:CONFEIN5}, which completes the proof.
\end{proof}

We proceed by deriving a $(n-1)+1$ decomposition of the reduced Einstein equations that is formulated in terms of the variables defined by \eqref{decom-g} and  \eqref{decom-F}. Here, one should view the gravitational fields 
$\lambda+1$, $\xi^a$ and $h^{a b} - \ulh^{a b}$ as being small, and hence, representing perturbations of the background de Sitter solution.  As we establish in the following corollary, each of these perturbed variables satisfies a wave equation.

\begin{corollary}
	Let $\lambda$, $\xi^c$, $h^{ab}$, $E_b$ and $H_{db}$ be as defined above by \eqref{decom-g}, \eqref{decom-F} and \eqref{def-MYM}. Then $\lambda+1$, $\xi^a$ and $h^{a b} - \ulh^{a b}$ satisfy the wave equations 
	\al{TTEQ}{
		&\frac{1}{2}g^{cd}\udn{c}\udn{d} (\lambda+1)  +P^{ab}\nu_a\nu_b+ Q^{ab}\nu_a\nu_b +\frac{1}{n-2}X^a X^b\nu_a\nu_b     \nnb\\
		={}&\frac{n-2}{2\tan(Ht)}\nu^c\udn{c} (\lambda+1) +   \frac{(\lambda+1)^2}{\sin^2(Ht)}+ \frac{n-3}{\sin^2(Ht)} (\lambda+1)  -(n-2)(\lambda+1) \notag  \\
		&+\Bigl(\nu_a\nu_bg^{bd}g^{a\ha}-\frac{1}{2(n-2)} \lambda  g^{d\ha}\Bigr) g^{c\hc}(H_{\ha\hc}-E_{\ha}\nu_{\hc}+\nu_{\ha}E_{\hc})(H_{dc}-E_{d}\nu_{c}+\nu_{d}E_{c}),
	}
	\al{TSEQ}{
		&\frac{1}{2}g^{cd}\udn{c}\udn{d} \xi^e   +P^{ab}\nu_a\ts{\ulh}{^e_b}+ Q^{ab}\nu_a\ts{\ulh}{^e_b} +\frac{1}{n-2}X^a X^b\nu_a\ts{\ulh}{^e_b}  \nnb\\
		={}&\frac{n-2}{2\tan(Ht)}\nu^c\udn{c}\xi^e  + \frac{n-2}{2\tan^2(Ht)}\xi^e  + (\lambda+1) \xi^e \frac{1}{\sin^2(Ht)} + \frac{1}{2} (n-2)\xi^e  \notag  \\
		&+\Bigl(\nu_a\ts{\ulh}{^e_b}g^{bd}g^{a\ha}-\frac{1}{2(n-2)}\xi^e g^{d\ha}\Bigr) g^{c\hc}(H_{\ha\hc}-E_{\ha}\nu_{\hc}+\nu_{\ha}E_{\hc})(H_{dc}-E_{d}\nu_{c}+\nu_{d}E_{c}),
	}
	and
	\al{SSEQ}{
		&\frac{1}{2}g^{cd}\udn{c}\udn{d} \left( h^{ef} - \ulh^{e f} \right) + \ts{\ulh}{^e_a}\ts{\ulh}{^f_b}P^{ab}+ \ts{\ulh}{^e_a}\ts{\ulh}{^f_b}Q^{ab} +\frac{1}{n-2}\ts{\ulh}{^e_a}\ts{\ulh}{^f_b}X^a X^b \nnb\\
		={}&\frac{n-2}{2\tan(Ht)}\nu^c\udn{c} \left( h^{ef} - \ulh^{e f} \right) + h^{ef} \frac{\lambda+1}{\sin^2(Ht)}+ (n-2) (\ulh^{e f}- h^{e f}) \notag  \\
		&+\Bigl(\ts{\ulh}{^e_a}\ts{\ulh}{^f_b}g^{bd}g^{a\ha}-\frac{1}{2(n-2)} h^{ef} g^{d\ha}\Bigr) g^{c\hc}(H_{\ha\hc}-E_{\ha}\nu_{\hc}+\nu_{\ha}E_{\hc})(H_{dc}-E_{d}\nu_{c}+\nu_{d}E_{c}),
	}
	respectively.
\end{corollary}
\begin{proof}
	Expressing the non-linear source term in \eqref{E:CONFEIN5} involving $F_{ab}$  in terms of $E_b$ and $H_{a b}$, see \eqref{decom-F} and \eqref{def-MYM}, as 
	\begin{align*}
		& g^{bd}F^{ac}F_{dc}-\frac{1}{2(n-2)}g^{ab}F^{cd}F_{cd} \notag  \\
		={} & \Bigl(g^{bd}g^{a\ha}-\frac{1}{2(n-2)}g^{ab} g^{d\ha}\Bigr) g^{c\hc}(H_{\ha\hc}-E_{\ha}\nu_{\hc}+\nu_{\ha}E_{\hc})(H_{dc}-E_{d}\nu_{c}+\nu_{d}E_{c}),
	\end{align*}
	the proof then follows from applying $\nu_a\nu_b$, $\nu_a\ts{\ulh}{^e_b}$, and $\ts{\ulh}{^e_a}\ts{\ulh}{^f_b}$ to \eqref{E:CONFEIN5} and employing Corollary \ref{t:bkgdR} and the relation \eqref{E:NBNU} from Lemma \ref{t:nbnuh}.
\end{proof}

Next, we further decompose the gravitational field $h^{a b}$ in terms of the variables $q$  and  $\mathfrak{h}^{ab}$ defined by \eqref{E:q}. To carry out the decomposition, we introduce the projection operator
\al{LG}{
	\ts{\mathcal{L}}{^{ab}_{cd}} =\ts{\delta}{^a_c} \ts{\delta}{^b_d}-\frac{1}{n-1}h^{ab} h_{cd}.
}
The motivation for the additional decomposition of $h^{a b}$ is that, due to the identity
\al{LGID}{
	\ts{\mathcal{L}}{^{ab}_{cd}} h^{cd}=0,
}
an application of the projection \eqref{E:LG} to the wave equation \eqref{E:SSEQ} for $h^{ab}$ removes the problematic singular term  
$h^{ef}\frac{\lambda+1}{\sin^2(Ht)}$ and results in a wave equation for $\mathfrak{h}^{ab}$. As we shall see in the following corollary, the particular form of the variable $q$ is chosen to remove problematic singular terms its evolution equation, which also happens to be a wave equation. 

\begin{corollary}\label{t:hh}
	Let $q$ and $\mathfrak{h}^{ab} - \ulh^{a b}$ be defined by \eqref{E:q}. Then $q$ and $\mathfrak{h}^{ab} - \ulh^{a b}$ satisfy the wave equations
	\al{SSEQ2}{
		&\frac{1}{2}g^{cd}\udn{c}\udn{d}q-\frac{3-n}{2(n-1)} g^{cd}\udn{c}h_{ef} \udn{d} h^{ef} +\frac{3-n}{n-1} h_{ab} P^{ab}+ \frac{3-n}{n-1} h_{ab} Q^{ab} \nnb\\ &  +\frac{3-n}{(n-1)(n-2)}h_{ab} X^a X^b
		+P^{ab}\nu_a\nu_b+ Q^{ab}\nu_a\nu_b +\frac{1}{n-2}X^a X^b\nu_a\nu_b \nnb \\
		={}&\frac{(n-2)}{2\tan(Ht)} \nu^c\udn{c} q +  (\lambda+1)^2 \frac{1}{\sin^2(Ht)}  -(n-2)(\lambda+1) +\Bigl(\frac{3-n}{n-1} h_{ab}  g^{bd}g^{a\ha}\notag  \\
		& -\frac{3-n+\lambda}{2(n-2)}   g^{d\ha} +\nu_a\nu_bg^{bd}g^{a\ha}\Bigr)  \cdot g^{c\hc}(H_{\ha\hc}-E_{\ha}\nu_{\hc}+\nu_{\ha}E_{\hc})(H_{dc}-E_{d}\nu_{c}+\nu_{d}E_{c}),
	}
	and
	\al{SSEQTRFR3}{
		&\frac{1}{2}g^{cd}\udn{c}\udn{d}(\mathfrak{h}^{ab}- \ulh^{ab})-\frac{1}{2}g^{cd}\udn{c}(S^{-1}\ts{\mathcal{L}}{^{ab}_{ef}})\udn{d} h^{ef} +S^{-1}\ts{\mathcal{L}}{^{ab}_{ef}}\ts{\ulh}{^e_a}\ts{\ulh}{^f_b}P^{ab}   \nnb\\
		&+ S^{-1}\ts{\mathcal{L}}{^{ab}_{ef}}\ts{\ulh}{^e_a}\ts{\ulh}{^f_b}Q^{ab} +\frac{1}{n-2} S^{-1}\ts{\mathcal{L}}{^{ab}_{ef}}\ts{\ulh}{^e_a}\ts{\ulh}{^f_b}X^a X^b   \nnb\\
		={} &\frac{n-2}{2\tan(Ht)}\nu^c \udn{c}(\mathfrak{h}^{ab}-\ulh^{ab}) - (n-2) S^{-1} \ts{\mathcal{L}}{^{ab}_{cd}} (\mathfrak{h}^{c d} - \ulh^{c d} )  \notag  \\
		&+S^{-1}\ts{\mathcal{L}}{^{ab}_{ef}} \ts{\ulh}{^e_a}\ts{\ulh}{^f_b}g^{bd}g^{a\ha}  g^{c\hc}(H_{\ha\hc}-E_{\ha}\nu_{\hc}+\nu_{\ha}E_{\hc})(H_{dc}-E_{d}\nu_{c}+\nu_{d}E_{c}),
	}
	respectively.
\end{corollary}

\begin{proof}
	Rather than deriving the wave \eqref{E:SSEQ2} for $q$ directly, we first derive a wave equation for $\ln S$, where we recall that $S$ is defined by \eqref{e:S}. We begin the derivation by noting from Jacobi's identity and the definition of $\alpha$, see \eqref{e:S}, that 
	\begin{equation*}
		(n-1)\udn{d}(\ln\alpha)=h_{ab}\del{d} h^{ab} =h_{ab}\udn{d} h^{ab}+\ulh_{ab}\del{d}\ulh^{ab} \AND (n-1)\udn{d}(\ln\ula)=\ulh_{ab}\del{d} \ulh^{ab}.
	\end{equation*} 
	Using these expressions, we obtain
	\als{
		h_{ab}\udn{d} h^{ab}=(n-1)\udn{d}(\ln S),
	}
	which, after differentiating, yields
	\al{alpha}{
		h_{ab}g^{cd} \udn{c}\udn{d} h^{ab} ={}&(n-1) g^{cd}\udn{c}\udn{d}(\ln S)-g^{cd}\udn{c}h_{ab} \udn{d} h^{ab}. 
	}
	Contracting \eqref{E:SSEQ} with $h_{ef}/(n-1)$ while making use of \eqref{E:alpha} then shows that $\ln S$ satisfies
	\al{SSEQ1}{
		&\frac{1}{2}g^{cd}\udn{c}\udn{d}(\ln S)-\frac{1}{2(n-1)} g^{cd}\udn{c}h_{ef} \udn{d} h^{ef} +\frac{1}{n-1} h_{ab} P^{ab}+ \frac{1}{n-1} h_{ab} Q^{ab}  \nnb\\  &\hspace{1cm} +\frac{1}{(n-1)(n-2)}h_{ab} X^a X^b  =\frac{n-2}{2\tan(Ht)} \nu^c\udn{c}\ln S + (\lambda+1)   \frac{1}{\sin^2(Ht)}  \notag  \\
		&  +\Bigl(\frac{1}{n-1} h_{ab}  g^{bd}g^{a\ha}-\frac{1}{2(n-2)}   g^{d\ha}\Bigr) g^{c\hc}(H_{\ha\hc}-E_{\ha}\nu_{\hc}+\nu_{\ha}E_{\hc})(H_{dc}-E_{d}\nu_{c}+\nu_{d}E_{c}).
	}
	Noting the singular terms $(\lambda+1) \frac{1}{\sin^2(Ht)}$ in \eqref{E:SSEQ1} and $\frac{n-3}{\sin^2(Ht)} (\lambda+1)$ in \eqref{E:TTEQ}, the combination $(3-n)  \times \eqref{E:SSEQ1} + \eqref{E:TTEQ} $ removes these terms and, with the help of the definition \eqref{E:q} of $q$, leads to the wave equation \eqref{E:SSEQ2} for $q$.

	In order to obtain the wave equation \eqref{E:SSEQTRFR3}, we first note from \eqref{E:q} and  \eqref{E:LG} that
	\al{LGID-1}{
		\frac{1}{S} \ts{\mathcal{L}}{^{ab}_{cd}} \udn{e} h^{cd}=\udn{e}\mathfrak{h}^{ab}.
	}
	Making use of the identities \eqref{E:LGID}--\eqref{E:LGID-1}, it follows from applying  $S^{-1}\ts{\mathcal{L}}{^{ab}_{ef}}$ to \eqref{E:SSEQ} that 
	\begin{align*}
		&\frac{1}{2}g^{cd}\udn{c}\udn{d}\mathfrak{h}^{ab}-\frac{1}{2}g^{cd}\udn{c}(S^{-1}\ts{\mathcal{L}}{^{ab}_{ef}})\udn{d} h^{ef} +S^{-1}\ts{\mathcal{L}}{^{ab}_{ef}}\ts{\ulh}{^e_a}\ts{\ulh}{^f_b}P^{ab}   \nnb\\
		&+ S^{-1}\ts{\mathcal{L}}{^{ab}_{ef}}\ts{\ulh}{^e_a}\ts{\ulh}{^f_b}Q^{ab} +\frac{1}{n-2} S^{-1}\ts{\mathcal{L}}{^{ab}_{ef}}\ts{\ulh}{^e_a}\ts{\ulh}{^f_b}X^a X^b   \nnb\\
		={} &\frac{n-2}{2\tan(Ht)}\nu^c \udn{c}\mathfrak{h}^{ab} + (n-2) S^{-1} \ts{\mathcal{L}}{^{ab}_{cd}} (\ulh^{c d} - h^{c d}) \notag  \\
		&+S^{-1}\ts{\mathcal{L}}{^{ab}_{ef}} \ts{\ulh}{^e_a}\ts{\ulh}{^f_b}g^{bd}g^{a\ha}  g^{c\hc}(H_{\ha\hc}-E_{\ha}\nu_{\hc}+\nu_{\ha}E_{\hc})(H_{dc}-E_{d}\nu_{c}+\nu_{d}E_{c}).
	\end{align*} 
	However, due to \eqref{E:LGID}, we have \[(n-2) S^{-1} \ts{\mathcal{L}}{^{ab}_{cd}} (\ulh^{c d} - h^{c d}) =  (n-2) S^{-1} \ts{\mathcal{L}}{^{ab}_{cd}} (\ulh^{c d} - S^{-1} h^{c d}),\] and so,  $\mathfrak{h}^{ab}$ satisfies the wave equation \eqref{E:SSEQTRFR3}.
\end{proof}

\subsubsection{A Fuchsian formulation of the reduced conformal Einstein equations}\label{s:fucein}
In this section, we transform the second order wave equations \eqref{E:TTEQ}--\eqref{E:TSEQ} and \eqref{E:SSEQ2}--\eqref{E:SSEQTRFR3} for the gravitational fields  $\lambda$, $\xi^c$, $q$ and $\mathfrak{h}^{ab}$ into a first order, symmetric hyperbolic Fuchsian equation that is expressed in terms of the renormalized fields \eqref{E:W}--\eqref{E:QD}. From \eqref{E:W}--\eqref{E:QD}, we observe that
\al{ID-Lam-Xi}{
	\udn{d}\lambda= m_d+\frac{\tta}{\ttj H} m \nu_d \AND
	\udn{d}\xi^a= \tp{a}{d}+\frac{\ttb}{\ttk H} p^a \nu_d.
}
Using these relations and \eqref{E:NORMT}, we obtain the first order equations
\gat{
	\udn{d} m =\left(\frac{1}{\ttj}-1\right)\frac{1}{Ht} m \nu_d +\frac{1}{\tta t} m_d,  \label{E:ID2-0} \\
	\udn{d} p^a =\left(\frac{1}{\ttk}-1\right)\frac{1}{Ht} p^a \nu_d +\frac{1}{\ttb t} \tp{a}{d},  \label{E:ID3-0}   \\
	\intertext{and}
	\udn{d}s^{ab}=\tss{ab}{d}, \qquad  \udn{d}s=s_d,   \label{E:ID4-0}
}
from differentiating \eqref{E:W}, \eqref{E:V}, \eqref{E:U} and \eqref{E:Q}.  In the derivation of our Fuchsian formulation, we will need to consider these first order equations as well as the second order wave equations \eqref{E:TTEQ}--\eqref{E:TSEQ} and \eqref{E:SSEQ2}--\eqref{E:SSEQTRFR3}.

In the following, we use repeatedly the symmetrizing tensor $Q^{edc}$ defined by
\al{BTENSOR}{
	Q^{edc}= \nu^e g^{dc} + \nu^d g^{ec} -\nu^c g^{ed}
}
to cast the reduced conformal Einstein equations in first order form. To illustrate how this tensor is employed, we consider first the wave equation \eqref{E:SSEQTRFR3}. Using the symmetrizing tensor $Q^{edc}$, we observe that the principle part of \eqref{E:SSEQTRFR3}, which by \eqref{E:U} is given by $g^{cd}\udn{c}\tss{ab}{d}$, can be expressed in the symmetric form $Q^{edc}\udn{c}\tss{ab}{d}$ using the relation
\begin{align*}
	\nu^e (g^{cd}\udn{c}\tss{ab}{d})=& Q^{edc}\udn{c}\tss{ab}{d} -( \nu^d g^{ec}\udn{c}\tss{ab}{d} - \nu^c g^{ed} \udn{c} \tss{ab}{d}) \notag  \\
	=  & Q^{edc}\udn{c}\tss{ab}{d} -\nu^d g^{ec}(\udn{c}\udn{d}s^{ab}- \udn{d} \udn{c}s^{ab}) \notag \\
	=  & Q^{edc}\udn{c}\tss{ab}{d} +\nu^d g^{ec} (\ts{\udl{R}}{_{cdf}^a} s^{fb} +\ts{\udl{R}}{_{cdf}^b} s^{fa} ).
\end{align*}
In this way, the symmetrizing tensor $Q^{edc}$ can be used to obtain a symmetric hyperbolic formulation from the wave equation \eqref{E:SSEQTRFR3}.

With the preliminary calculations out of the way, we are now ready to complete the transformation of the reduced conformal Einstein equations into Fuchsian form. This will be done in a sequence of steps starting with the Fuchsian formulation of the wave equation  \eqref{E:TTEQ} for $\lambda+1$. The resulting Fuchsian equation is displayed in the following lemma. 

\begin{lemma}\label{lem-FOSHS-m}
	Let $\mathbf{U}$ be defined by \eqref{def-U}. Then the wave equation \eqref{E:TTEQ} for the gravitational field $\lambda+1$ can be expressed in the first order Fuchsian form
	\als{
		- \bar{\mathbf{A}}_1^0 \nu^c \nb_{c} \p{-\nu^e m_e \\ \ts{\ulh}{^e_{\he}} m_e \\ m } +	\bar{\mathbf{A}}_1^c \ts{\ulh}{^{b}_{c}} \udn{b}\p{-\nu^e m_e \\ \ts{\ulh}{^e_{\he}} m_e \\ m } =\frac{1}{Ht} \bar{\mathcal{B}}_1 \p{-\nu^e m_e \\ \ts{\ulh}{^e_{\he}} m_e \\ m } +\bar{G}_1,
	}
	where
	\gas{
		\bar{\mathbf{A}}_1^0= \p{-\lambda & 0 & 0\\0 & h^{f \he} & 0 \\0 & 0 & \tte (-\lambda)}, \quad
		\bar{\mathbf{A}}_1^c \ts{\ulh}{^{b}_{c}} =\p{-2\xi^b & -\ts{h}{^{\he b}}  & 0 \\
			-\ts{h}{^{f b}}  & 0 & 0 \\ 0 & 0 & 0 },\quad \bar{G}_1=\p{ \nu_e \triangle^e_1(t, \mathbf{U}) \\ \ts{\ulh}{^f_e} \triangle^e_1(t, \mathbf{U}) \\ 0 }, \nnb \\
		\bar{\mathcal{B}}_1 =  \p{\left(n-2-\frac{1}{\ttj} \right)(-\lambda) & 0 & \left(2-\frac{1}{\ttj}\right)\left(\frac{1}{\ttj}-n+3\right)\frac{\tta}{H}(-\lambda) \\
			0 &  \frac{1}{\ttj} h^{f \he} & 0 \\
			\tte \frac{H}{\tta}(-\lambda) & 0 & \tte \left(\frac{1}{\ttj}-1\right)  (-\lambda)},
	}
	$\tte$ is an arbitrary constant constant, and 
	\als{
		&\triangle^e_1(t, \mathbf{U})= \notag \\ 
		&  \left(2-\frac{1}{\ttj}\right) \left(n-3-\frac{1}{\ttj}\right) \frac{\tta^2}{H^2} m^2 \nu^e - \frac{\ttb}{\ttj H }p^d m_d \nu^e - \frac{\ttb}{\ttj H }  p^a m_a \nu^e  \nnb  \\
		&+ \frac{1}{\ttj H }\tta m m_d \nu^e \nu^d - \frac{\tta^2}{ \ttj H^2}  \left(\frac{1}{\ttj}-1\right)  m^2 \nu^e - \tta \left( \frac{1}{ \ttj }-n+2\right)\frac{1}{H }m\nu^e\nu^d  m_d   \nnb  \\
		&- \left(\frac{1}{Ht}
		-\frac{1}{ \tan(Ht)}\right)(n-2)\nu^e \nu^c m_c + \left(\frac{1}{Ht} -\frac{1}{ \tan(Ht)}\right)(n-2)\frac{\tta}{\ttj H} m \nu^e  \nnb \\
		&-  2(n-3 ) \left(\frac{1}{ (Ht)^2}  + \frac{1}{\sin^2(Ht)} \right) \tta tm \nu^e - \frac{2(\tta tm)^2}{\sin^2(Ht)}  \nu^e - 2 (n-2)\tta tm \nu^e \nnb\\
		&-2P^{ab}\nu_a\nu_b \nu^e- 2Q^{ab}\nu_a\nu_b \nu^e - \frac{2}{n-2}X^a X^b\nu_a\nu_b \nu^e \nnb \\
		&+2\nu^e\Bigl(\nu_a\nu_bg^{bd}g^{a\ha}-\frac{1}{2(n-2)} \lambda  g^{d\ha}\Bigr) g^{c\hc}(H_{\ha\hc}-E_{\ha}\nu_{\hc}+\nu_{\ha}E_{\hc})(H_{dc}-E_{d}\nu_{c}+\nu_{d}E_{c}).
	}
	
\end{lemma}
\begin{proof}
	Using the definition \eqref{E:W} and \eqref{E:WD} for $m_d$ and $m$, respectively, we can express the wave equation \eqref{E:TTEQ} for $\lambda+1$  as
	\al{TTEQ2a}{
		&	g^{cd}\udn{c}m_d   = \notag \\
		& \frac{n-2}{ \tan(Ht)}\nu^c m_c - \frac{n-2}{ \tan(Ht)}\frac{\tta}{\ttj H} m + \frac{2\tta t}{\sin^2(Ht)} m(n-2+\lambda)  - 2P^{ab}\nu_a\nu_b - 2Q^{ab}\nu_a\nu_b  \nnb\\
		&-2 (n-2)\tta tm - \frac{\tta}{ \ttj H^2 t}  \left(\frac{1}{\ttj}-1\right)  \lambda m - \frac{\ttb}{\ttj H }  p^a m_a + \frac{1}{ \ttj H t}\lambda \nu^a m_a- \frac{2}{n-2}X^a X^b\nu_a\nu_b \nnb\\
		&+2\Bigl(\nu_a\nu_bg^{bd}g^{a\ha}-\frac{1}{2(n-2)} \lambda  g^{d\ha}\Bigr) g^{c\hc}(H_{\ha\hc}-E_{\ha}\nu_{\hc}+\nu_{\ha}E_{\hc})(H_{dc}-E_{d}\nu_{c}+\nu_{d}E_{c}),
	}
	where in deriving this we have employed the identity \[-\frac{1}{2\ttj H t} g^{cd}\nu_d  m_c = -\frac{\ttb}{2\ttj H}  p^a m_a+\frac{1}{2\ttj H t}\lambda \nu^a m_a. \]	Multiplying \eqref{E:TTEQ2a} by $\nu^e$ yields
	\al{TTEQ2-2}{
		\nu^e g^{cd}\udn{c}m_d   =   & \left(n-2- \frac{1}{ \ttj }\right)\frac{1}{Ht} \nu^e \nu^d m_d - \left(\frac{1}{\ttj}\left(n-1-\frac{1}{\ttj}\right)-2n+6 \right)   \frac{\tta}{H^2t}m \nu^e  + \triangle_{11}^e
	}
	where $\triangle_{11}^e$ is defined by
	\begin{align}
		\label{Delta-m2}
		&\triangle_{11}^e  = \notag \\
		& \frac{1}{\ttj H }\tta m \nu^e \nu^d m_d  - \frac{\tta}{ \ttj H^2}  \left(\frac{1}{\ttj}-1\right) \tta m^2 \nu^e- \left(\frac{1}{Ht}
		-\frac{1}{ \tan(Ht)}\right)(n-2)\nu^e \nu^c m_c   \nnb \\
		& + \left(\frac{1}{Ht} -\frac{1}{ \tan(Ht)}\right)(n-2)\frac{\tta}{\ttj H} m \nu^e -  2(n-3 ) \left(\frac{1}{ (Ht)^2}  - \frac{1}{\sin^2(Ht)} \right) \tta tm \nu^e + \frac{2(\tta tm)^2}{\sin^2(Ht)}  \nu^e \nnb\\
		&- 2 (n-2)\tta tm \nu^e - \frac{\ttb}{\ttj H }  p^a m_a \nu^e  - 2P^{ab}\nu_a\nu_b \nu^e - 2Q^{ab}\nu_a\nu_b \nu^e - \frac{2}{n-2}X^a X^b\nu_a\nu_b \nu^e  \notag  \\
		&+2\nu^e\Bigl(\nu_a\nu_bg^{bd}g^{a\ha}-\frac{1}{2(n-2)} \lambda  g^{d\ha}\Bigr) g^{c\hc}(H_{\ha\hc}-E_{\ha}\nu_{\hc}+\nu_{\ha}E_{\hc})(H_{dc}-E_{d}\nu_{c}+\nu_{d}E_{c}).
	\end{align}
	
	To symmetrize the system, we use the symmetrizing tensor $Q^{edc}$ defined by \eqref{E:BTENSOR} as
	well as the identities
	\begin{equation}
		\frac{1}{t}\nu^d\nu^e =  \frac{1}{t}Q^{fgc}\nu_c(\ts{\delta}{^e_f}\ts{\delta}{^d_g}-\ts{\ulh}{^e_f}\ts{\ulh}{^d_g})+\tta m\nu^e\nu^d  \AND
		\frac{1}{t}\nu^e =
		-\frac{1}{t} Q^{edc}\nu_c\nu_d+ \tta m\nu^e, \label{QCOMMUN2}
	\end{equation}
	which are a consequence of the definitions \eqref{decom-g}, \eqref{E:W} and  \eqref{E:BTENSOR}.
	By  \eqref{E:BTENSOR} and \eqref{E:TTEQ2-2}, we have
	\al{TTEQ2}{
		& Q^{edc}\udn{c}m_d - \left(\nu^d g^{ec} \udn{c}m_d -\nu^c g^{ed} \udn{c}m_d \right) \nnb \notag \\
		=  &\left( n-2-\frac{1}{ \ttj }\right)\frac{1}{Ht}   Q^{fgc}\nu_c(\ts{\delta}{^e_f}\ts{\delta}{^d_g}-\ts{\ulh}{^e_f}\ts{\ulh}{^d_g})m_d  \notag \\  
		&+ \left(2-\frac{1}{\ttj}\right) \left(\frac{1}{\ttj}-n+3\right)   \frac{\tta}{H^2 t}  Q^{edc}\nu_c\nu_d m   +  \triangle_{12}^e,
	}
	where we note that $\frac{1}{\ttj}\left(n-1-\frac{1}{\ttj}\right)-2n+6 = \left(2-\frac{1}{\ttj}\right) \left(\frac{1}{\ttj}-n+3\right)$ and $\triangle_{12}^e$  is defined by
	\als{
		\triangle_{12}^e = {}& \left(n-2 - \frac{1}{ \ttj }\right) \frac{\tta}{H }m\nu^e\nu^d  m_d - \left(2-\frac{1}{\ttj}\right) \left(\frac{1}{\ttj}-n+3\right)  \frac{\tta^2}{H^2} m^2 \nu^e + \triangle_{11}^e,
	}
	with $\triangle_{11}^e$ given by \eqref{Delta-m2}.
	Observing from
	\eqref{E:ID-Lam-Xi} and \eqref{E:ID2-0} that
	\begin{align*}
		& \nu^d g^{ec} \udn{c}m_d -\nu^c g^{ed} \udn{c}m_d = \nu^d g^{ec} (\udn{c}m_d - \udn{d}m_c)  \nnb  \\
		= {}&\nu^d g^{ec} \left[\udn{c}\udn{d} \lambda -  \udn{c} \frac{\tta m}{\ttj H}  \nu_d  - \udn{d} \udn{c} \lambda + \udn{d} \frac{\tta m}{\ttj H}  \nu_c  \right]
		= \frac{\tta}{\ttj H } \nu^d g^{ec} \left(  - \udn{c}m \nu_d + \udn{d} m \nu_c  \right)  \nnb  \\
		= {} & \frac{\tta}{\ttj H } g^{ec} \ts{\ulh}{^d_c}\udn{d} m
		=  \frac{1}{\ttj H t} g^{ec} \ts{\ulh}{^d_c} m_d   = \frac{1}{\ttj H t} (h^{ed}-\xi^d\nu^e) m_d
		=   -\frac{\ttb}{\ttj H }p^d m_d \nu^e +\frac{1}{\ttj H t}m_d  h^{ed}   \nnb \\
		= {} & -\frac{\ttb}{\ttj H }p^d m_d \nu^e +\frac{1}{\ttj H t}Q^{abc}\nu_c \ts{\ulh}{^e_a} \ts{\ulh}{^d_b} m_d,
	\end{align*}
	we find after substituting this expression into \eqref{E:TTEQ2} that
	\al{TTEQ3}{
		Q^{edc}\udn{c}m_d
		={} & \left( n-2-\frac{1}{ \ttj }\right)\frac{1}{Ht}   Q^{fgc}\nu_c(\ts{\delta}{^e_f}\ts{\delta}{^d_g}-\ts{\ulh}{^e_f}\ts{\ulh}{^d_g})m_d \nnb \\
		& + \left(2-\frac{1}{\ttj}\right) \left(\frac{1}{\ttj}-n+3\right)  \frac{\tta}{H^2 t}  Q^{edc}\nu_c\nu_d m + \frac{1}{\ttj H t}Q^{fgc}\nu_c \ts{\ulh}{^e_f} \ts{\ulh}{^d_g} m_d  + \triangle^e_1
	}
	where $\triangle^e_1$ is defined by
	\[\triangle^e_1 = - \frac{\ttb}{\ttj H }p^d m_d \nu^e + \triangle_{12}^e.  \]

	To complete the derivation,  the evolution equation \eqref{E:TTEQ3} of $m_d$ needs to be supplemented with a first order evolution equation for $m$.
	Using the identity $Q^{ebc} \nu_b\nu_c = \lambda \nu^e$ from Lemma \ref{lem-identity}, it follows immediately from \eqref{E:ID2-0} that $m$ evolves according to
	\als{
		\tte \lambda \nu^e \udn{e} m =-\tte \lambda   \left(\frac{1}{\ttj}-1\right)\frac{1}{Ht} m + \tte  \frac{\lambda }{\tta t} m_e \nu^e,  
	}
	where $\tte$ is a constant will be fixed later. Combining this equation with  \eqref{E:TTEQ3} in a matrix form yields the singular symmetric hyperbolic system 
	\als{
		\mathbf{A}_1^c\udn{c}\p{m_d \\ m }=\frac{1}{Ht} \mathcal{B}_1 \p{m_d\\ m}+\p{ \triangle^e_1 \\ 0 }
	}
	where
	\gas{
		\mathbf{A}_1^c=\p{ Q^{edc} & 0   \\
			0 &  \lambda \tte \nu^c } \label{E:SYS1NONL1}
	}
	and
	\als{
		& \small \mathcal{B}_1 =  Q^{abc}\nu_c\p{
			\left(n-2-\frac{1}{\ttj}\right)(\ts{\delta}{^e_a}\ts{\delta}{^d_b}-\ts{\ulh}{^e_a}\ts{\ulh}{^d_b})+\frac{1}{\ttj} \ts{\ulh}{^e_a}\ts{\ulh}{^d_b} & \left(2-\frac{1}{\ttj}\right)\left(\frac{1}{\ttj}-n+3\right) \frac{\tta}{H} \ts{\delta}{^e_a}\nu_b  \\
			\tte \frac{H}{\tta}\ts{\delta}{^d_a}\nu_b  &   \tte \left(\frac{1}{\ttj}-1\right) \nu_b\nu_a
		}  \nnb  \\
		& \small=\p{ Q^{abc} \nu_c & 0 \\
			0 &  -\lambda\tte}  \p{
			\left(n-2-\frac{1}{\ttj}\right)(\ts{\delta}{^e_a}\ts{\delta}{^d_b}-\ts{\ulh}{^e_a}\ts{\ulh}{^d_b})  +\frac{1}{\ttj} \ts{\ulh}{^e_a}\ts{\ulh}{^d_b}   & \left(2-\frac{1}{\ttj}\right)\left(\frac{1}{\ttj}-n+3\right) \frac{\tta}{H} \ts{\delta}{^e_a}\nu_b   \\
			-\frac{H}{\tta}\nu^d  &  \left(\frac{1}{\ttj}-1\right)}.
	}
	\als{
		\bar{\mathbf{A}}_1^c\udn{c}\p{-\nu^e m_e \\ \ts{\ulh}{^e_{\he}} m_e \\ m }=\frac{1}{Ht} \bar{\mathcal{B}}_1 \p{-\nu^e m_e \\ \ts{\ulh}{^e_{\he}} m_e \\ m } +\bar{G}_1
	}
	where
	\gat{
		\bar{\mathbf{A}}_1^c=\p{ \nu_e Q^{edc}\nu_d &  \nu_e Q^{edc}\ts{\ulh}{^{\he}_d} & 0 \\
			\ts{\ulh}{^f_e}Q^{edc}\nu_d  &  \ts{\ulh}{^f_e}Q^{edc}\ts{\ulh}{^{\he}_d} & 0 \\ 0 & 0 & \lambda \tte \nu^c}, \qquad
		\bar{G}_1=\p{\nu_e & 0 \\ \ts{\ulh}{^f_e} & 0 \\ 0 & 1 }  \p{\triangle^e_1 \\ 0 }, \nnb
	}
	and
	\als{
		\bar{\mathcal{B}}_1={}& \p{\nu_e & 0 \\ \ts{\ulh}{^f_e} & 0 \\ 0 & 1 } \mathcal{B}_1 \p{\nu_d & \ts{\ulh}{^{\he}_d} & 0 \\0 & 0 & 1} \notag \\
		= {}& \p{\left(n-2-\frac{1}{\ttj} \right)(-\lambda) & 0 & \left(2-\frac{1}{\ttj}\right)\left(\frac{1}{\ttj}-n+3\right)\frac{\tta}{H}(-\lambda) \\
			0 &  \frac{1}{\ttj} h^{f \he} & 0 \\
			\tte \frac{H}{\tta}(-\lambda) & 0 & \tte \left(\frac{1}{\ttj}-1\right)  (-\lambda)
		}.
	}
	To complete the proof, we observe from  Lemma \ref{lem-identity} that $\bar{\mathbf{A}}_1^c \ts{\ulh}{^{b}_{c}}$ and $  \bar{\mathbf{A}}^0_1:=\bar{\mathbf{A}}_1^c\nu_c$ agree with the formulas provided in the statement of the lemma.
\end{proof}

Corresponding Fuchsian formulations of the wave equations \eqref{E:TSEQ} and \eqref{E:SSEQ2}--\eqref{E:SSEQTRFR3} for the gravitation fields $\xi^e$, $q$ and $\mathfrak{h}^{ab}- \ulh^{ab}$ can be derived using similar arguments.  Detailed derivations are provided in Appendix \ref{s:App1}, see Lemmas \ref{lem-FOSHS-p}--\ref{lem-FOSHS-4}. The Fuchsian formulations for all the gravitational wave equations   \eqref{E:TTEQ}--\eqref{E:TSEQ} and \eqref{E:SSEQ2}--\eqref{E:SSEQTRFR3} are displayed together in the theorem below. The Fuchsian formulation of the Yang--Mills equaitons will be derived separately in \S \ref{sec-Max}. 

\begin{theorem}\label{thm-FOSHS-m}
	The reduced conformal Einstein equations can be expressed in the following first order, symmetric hyperbolic Fuchsian form:
	\als{
		&- \bar{\mathbf{A}}_1^0 \nu^c \nb_{c} \p{-\nu^e m_e \\ \ts{\ulh}{^e_{\he}} m_e \\ m } +	\bar{\mathbf{A}}_1^c \ts{\ulh}{^{b}_{c}} \udn{b}\p{-\nu^e m_e \\ \ts{\ulh}{^e_{\he}} m_e \\ m } =\frac{1}{Ht} \bar{\mathcal{B}}_1 \p{-\nu^e m_e \\ \ts{\ulh}{^e_{\he}} m_e \\ m } +\bar{G}_1(t,\mathbf{U}),  \\
		&-\bar{\mathbf{A}}_2^0 \nu^c \nb_{c} \p{-\nu^e \tp{a}{e} \\ \ts{\ulh}{^e_{\he}} \tp{a}{e} \\ p^a } + \bar{\mathbf{A}}_2^c \ts{\ulh}{^{b}_{c}} \udn{b} \p{-\nu^e \tp{a}{e} \\ \ts{\ulh}{^e_{\he}} \tp{a}{e} \\ p^a } =\frac{1}{Ht}\bar{\mathcal{B}}_2  \p{-\nu^e \tp{a}{e} \\ \ts{\ulh}{^e_{\he}} \tp{a}{e} \\ p^a } +\bar{G}_2(t,\mathbf{U}),
	}
	\als{
		&-\bar{\mathbf{A}}_3^0\nu^c \nb_{c} \p{-\nu^e \tss{\ha\hb}{e} \\ \ts{\ulh}{^e_{\he}} \tss{\ha\hb}{e} \\ s^{\ha \hb} } + \bar{\mathbf{A}}_3^c \ts{\ulh}{^{b}_{c}} \udn{b} \p{-\nu^e \tss{\ha\hb}{e} \\ \ts{\ulh}{^e_{\he}} \tss{\ha\hb}{e} \\ s^{\ha \hb} } = \frac{1}{Ht} \bar{\mathcal{B}}_3\p{-\nu^e \tss{\ha\hb}{e} \\ \ts{\ulh}{^e_{\he}} \tss{\ha\hb}{e} \\ s^{\ha \hb} }+\bar{G}_3(t,\mathbf{U}),\\
		&-\bar{\mathbf{A}}_4^0\nu^c \nb_{c} \p{-\nu^e s_{e} \\ \ts{\ulh}{^e_{\he}} s_{e} \\ s } + \bar{\mathbf{A}}_4^c \ts{\ulh}{^{b}_{c}} \udn{b} \p{-\nu^e s_{e} \\ \ts{\ulh}{^e_{\he}}s_{e} \\ s } = \frac{1}{Ht} \bar{\mathcal{B}}_4 \p{-\nu^e s_{e} \\ \ts{\ulh}{^e_{\he}}s_{e} \\ s } +\bar{G}_4(t,\mathbf{U}),
	}		
	where
	\als{
		&\bar{\mathbf{A}}_1^0= \p{-\lambda & 0 & 0\\0 & \ulh_{f \hc} h^{f \he} & 0 \\0 & 0 & \tte (-\lambda)}, \quad
		\bar{\mathbf{A}}_1^c \ts{\ulh}{^{b}_{c}} =\p{-2\xi^b & -\ts{h}{^{\he b}}  & 0 \\
			-\ulh_{f \hc} \ts{h}{^{f b}}  & 0 & 0 \\ 0 & 0 & 0 },  \\
		&\bar{\mathbf{A}}_2^0 = \p{-\lambda & 0 & 0\\ 0 & \ulh_{f \hd} h^{f \he} & 0 \\ 0 & 0 & \ttf (-\lambda)}, \quad
		\bar{\mathbf{A}}_2^c \ts{\ulh}{^{b}_{c}} =\p{ -2\xi^b & -\ts{h}{^{\he b}} & 0 \\
			-\ulh_{f \hd} \ts{h}{^{f b}} & 0 & 0 \\ 0 & 0 & 0},  \\
		& \bar{\mathcal{B}}_1 = \p{\left(n-2-\frac{1}{\ttj} \right)(-\lambda) & 0 & \left(2-\frac{1}{\ttj}\right)\left(\frac{1}{\ttj}-n+3\right)\frac{\tta}{H}(-\lambda) \\
			0 &  \frac{1}{\ttj} \ulh_{f \hc} h^{f \he} & 0 \\
			\tte \frac{H}{\tta}(-\lambda) & 0 & \tte \left(\frac{1}{\ttj}-1\right)  (-\lambda)},  \\	
		&\bar{\mathcal{B} }_2 = \p{\left(n-2-\frac{1}{\ttk} \right)(-\lambda) & 0 & \left(1-\frac{1}{\ttk}\right)\left(\frac{1}{\ttk}-n+2\right)\frac{\ttb}{H}(-\lambda) \\
			0 &  \frac{1}{\ttk} \ulh_{f \hd} \ts{h}{^{f \he}} & 0 \\
			\ttf \frac{H}{\ttb}(-\lambda) & 0 & \ttf \left(\frac{1}{\ttk}-1\right) (-\lambda)},  \\
		&\bar{G}_1(t,\mathbf{U})=\p{ \nu_e\triangle^e_1 (t, \mathbf{U})  \\ \ts{\ulh}{_{\hc e}} \triangle^e_1(t, \mathbf{U})  \\ 0 }, \qquad   \bar{G}_2(t,\mathbf{U})=\p{\nu_e \triangle^{ea}_2(t, \mathbf{U})  \\  \ts{\ulh}{_{\hd e}}\triangle^{ea}_2(t, \mathbf{U}) \\ 0  },
	}
	\als{
		&\bar{\mathbf{A}}_3^0 = \p{-\lambda & 0 & 0\\ 0 & \ulh_{f a} h^{f \he} & 0 \\ 0 & 0 & -\lambda}, &
		\bar{\mathbf{A}}_3^c \ts{\ulh}{^{b}_{c}} = \p{ -2\xi^b & -\ts{h}{^{\he b}} & 0 \\
			-\ulh_{f a} \ts{h}{^{f b}} & 0 & 0 \\ 0 & 0 & 0},  \\
		& \bar{\mathbf{A}}_4^0 = \p{-\lambda & 0 & 0\\ 0 & \ulh_{f \ha} h^{f \he} & 0 \\ 0 & 0 & -\lambda}, &
		\bar{\mathbf{A}}_4^c \ts{\ulh}{^{b}_{c}} =\p{ -2\xi^b & -\ts{h}{^{\he b}} & 0 \\
			- \ulh_{f \ha} \ts{h}{^{f b}} & 0 & 0 \\ 0 & 0 & 0},  \\
		&\bar{\mathcal{B} }_3 = \p{-\lambda\left(n-2\right) & 0 & 0 \\
			0 & 0 & 0 \\
			0 & 0 & 0}, &
		\bar{\mathcal{B} }_4 =\p{-\lambda\left(n-2\right) & 0 & 0 \\
			0 & 0 & 0 \\
			0 & 0 & 0}, \\
		&\bar{G}_3(t,\mathbf{U}) = \p{ \nu_e \triangle^{e \ha \hb}_3(t, \mathbf{U})  \\ \ts{\ulh}{_{a e}} \triangle^{e \ha \hb}_3(t, \mathbf{U}) \\ \lambda \nu^e \tss{\ha\hb}{e} }, &
		\bar{G}_4(t,\mathbf{U}) =\p{\nu_e \triangle^{e}_4(t, \mathbf{U}) \\ \ts{\ulh}{_{\ha e}}\triangle^{e}_4(t, \mathbf{U}) \\ \lambda \nu^e s_e }.
	}
	Here, the maps $\triangle^{e}_1(t, \mathbf{U})$, $ \triangle^{ea}_2(t, \mathbf{U})$, $\triangle^{e \ha \hb}_3(t, \mathbf{U})$ and $\triangle^{e}_4(t, \mathbf{U})$ are as defined in Lemmas \ref{lem-FOSHS-m} and \ref{lem-FOSHS-p}--\ref{lem-FOSHS-4}, and there exists constants  $\iota>0$ and $R>0$ such that these maps are analytic for 
	$(t,\mathbf{U})\in (-\iota,\frac{\pi}{H})\times B_R(0)$.
\end{theorem}
\begin{proof}
	The proof is a direct consequence of Lemmas \ref{lem-FOSHS-m} and \ref{lem-FOSHS-p}--\ref{lem-FOSHS-4}. 
\end{proof}

\subsection{Fuchsian  formulation of the Yang--Mills equations}\label{sec-Max}
In this section, we turn to deriving a Fuchsian formulation of the Yang--Mills equations. 	
Under the conformal change of variables \eqref{CONFG-F}--\eqref{CONFG-A}, the Yang--Mills system \eqref{eq-maxwell-div-o}--\eqref{eq-maxwell-bianchi-o} transforms into
\begin{align}
	\nabla^a F_{a b} &= - (n-3)  \nabla^{a} \Psi  F_{a b} - e^\frac{\Psi}{2} g^{a c} [A_c, F_{a b} ], \label{eq-conformal-Maxwell-div} \\
	\nabla_{[a} F_{b c]} &= - \nabla_{[a} \Psi \cdot F_{b c]} - e^\frac{\Psi}{2} [A_{[a}, F_{b c]} ], \label{eq-conformal-Maxwell-Bianchi}
\end{align}
which we will refer to as the \textit{conformal Yang--Mills equations}.
We remark that the Yang--Mills Bianchi equation \eqref{eq-conformal-Maxwell-Bianchi} is independent of the choice of connection, and hence, we can replace it by
\begin{equation*}
	\nb_{[a} F_{b c]} = - \nb_{[a} \Psi \cdot F_{b c]} - e^\frac{\Psi}{2} [A_{[a}, F_{b c]} ].
\end{equation*}
Noting that
\begin{align*}
	\nabla_a F_{b c}
	& = \nb_a F_{b c} -  X^d_{a b} F_{d c} -  X^d_{a c} F_{b d},
\end{align*}
the conformal Yang--Mills equation \eqref{eq-conformal-Maxwell-div}--\eqref{eq-conformal-Maxwell-Bianchi} can, with the help of \eqref{E:CONFFAC} and Lemma \ref{t:conf2},  be rewritten as
\begin{align}
	g^{b a} \nb_b F_{a c} & = \frac{ (n-3) }{\tan (Ht)} g^{b a} \nu_a F_{b c} + X^d F_{d c} + g^{b a} X^d_{b c} F_{a d} - \frac{\sqrt{H}}{\sqrt{\sin (Ht)}} g^{a b} [A_b, F_{a c} ], \label{Maxwell-div-F} \\
	\nb_{[b} F_{a c ]} &= \frac{ 1 }{\tan (Ht)} \nu_{[b} F_{a c]}  - \frac{\sqrt{H}}{\sqrt{\sin (Ht)}} [A_{[b}, F_{a c]} ]. \label{Maxwell-bianchi-F}
\end{align}
It is worth noting that among the above equations, the non-vanishing, dynamical equations arise from applying  $\ts{\ulh}{^{c}_{\hc}}$ and $\nu^c$ to \eqref{Maxwell-div-F} and $\nu^b \ts{\ulh}{^{a}_{p}} \ts{\ulh}{^{c}_{q}}$ to \eqref{Maxwell-bianchi-F},
while the non-dynamical constraint equations are obtained from applying $\ts{\ulh}{^{b}_{d}} \ts{\ulh}{^{a}_{p}} \ts{\ulh}{^{c}_{q}}$ and $g^{c d} \nu_d$ to  
\eqref{Maxwell-bianchi-F} and \eqref{Maxwell-div-F}, respectively.

In the following, we assume a \textit{temporal gauge} for the conformal Yang--Mills field defined by 
\be\label{temporal}
A_d \nu^d = 0.
\ee
With the help of this gauge choice, we combine, in the following lemma, the conformal Yang--Mills equations \eqref{Maxwell-div-F}--\eqref{Maxwell-bianchi-F} together with an equation for the spatial component of the gauge potential into a single system that is expressed in terms of the Yang--Mills fields
\[E_b=- \nu^p F_{p a} \tensor{\ulh}{^a_b}, \quad H_{db}=\tensor{\ulh}{^c_d} F_{c a} \tensor{\ulh}{^a_b}, \quad \text{and} \quad \bar A_b = A_a \ts{\ulh}{^{a}_{b}}. \]

\begin{remark} 
	It is worth recalling that $\nu^a$, defined by \eqref{E:NORMT}, is the unit timelike vector with respect to the conformal de Sitter metric $\udl{g}_{ab}$ while $\tilde T^a$, defined by \eqref{def-T}, is the unit timelike vector with respect to the physical metric $\tilde g_{ab}$. This implies, in general, that the gauge choices \eqref{phys-temp-gauge} and \eqref{temporal} are distinct.
\end{remark}

\begin{lemma}\label{lem-maxwell-hyperbolic-0}
	If $(E_{a},A_b)$ solves the conformal Yang--Mills equations \eqref{Maxwell-div-F}--\eqref{Maxwell-bianchi-F} in the temporal gauge \eqref{temporal}, then the triple $(E_{d}, \, H_{p q}, \, \bar A_s)$ defined via \eqref{decom-F} and \eqref{def-MYM} solves 
	\begin{align}\label{e:maineq1}
		- \A^{0}  \nu^c \nb_{c}
		\begin{pmatrix}
			E_{\ha} \\  H_{l \ha} \\ \bar A_{\ha}
		\end{pmatrix}
		+  \A^{c}_{\ha} \ts{\ulh}{^{\hc}_{c}}  \nb_{\hc}
		\begin{pmatrix}
			E_f \\ H_{l f} \\ \bar A_{f}
		\end{pmatrix} = {}&  \frac{ 1 }{Ht} \B
		\begin{pmatrix}
			E_{\ha} \\ H_{l \ha} \\ \bar A_{\ha}
		\end{pmatrix}
		+ \frac{1}{\sqrt{t}} \p{ \Xi_{1\ha} \\ \Xi^h_{2\ha} \\ \Xi_{3\ha}
		}
		+
		\begin{pmatrix}  \widehat{\Delta}_{1\ha}  \\ \widehat{\Delta}_{2\ha}^{h} \\ \widehat{\Delta}_{3\ha} \end{pmatrix},
	\end{align}
	where
	\begin{equation*}
		\A^0   =
		\begin{pmatrix}
			- \lambda & 0 & 0 \\
			0 &  h^{h  l} & 0 \\
			0 & 0 & 1
		\end{pmatrix}, \quad
		\B =
		\begin{pmatrix}
			-(n-3)   \lambda   &
			- (n-4) \xi^l & 0  \\
			0  & h^{h l} & 0 \\
			0 & 0 & \frac{1}{2}
		\end{pmatrix},
	\end{equation*}
	\begin{align*}
		\A_{\ha}^c \ts{\ulh}{^{\hc}_{c}} & =
		\begin{pmatrix}
			- 2 \tsh{f}{\ha} g^{d c} \nu_d \ts{\ulh}{^{\hc}_{c}} + \ts{\ulh}{^{\hc}_{\ha}} g^{d c} \nu_d \ts{\ulh}{^{f}_{c}} &
			- \ts{\ulh}{^{f}_{\ha}} h^{l \hc} & 0 \\
			\ts{\ulh}{^{\hc}_{\ha}} h^{h f} - \ts{\ulh}{^{f}_{\ha}} h^{h \hc} & 0 & 0 \\
			0 & 0 & 0
		\end{pmatrix},
	\end{align*}
	and
	\begin{align*}
		\Xi_{1\ha} ={}& h^{d c} [\bar A_c, H_{d \ha}], \quad \Xi^e_{2 \ha} = h^{e c}  ( [\bar A_{c}, E_{\ha} ] - [\bar A_{\ha}, E_{c} ]   ), \quad \Xi_{3 \ha} = E_{\ha}, \\
		\widehat{\Delta}_{1\ha} = {}& \sqrt{t} \ttb p^c (- [\bar A_{\ha}, E_c]  + 2 [\bar A_{c}, E_{\ha} ] )   - X^d (\nu_d E_{\ha} + H_{d \ha}) \notag \\
		&- g^{b c} \tensor{X}{^d_{b a}} \ts{\ulh}{^{a}_{\ha}} (\nu_c E_d - \nu_d E_c + H_{c d}) \\
		& +  \left( \frac{ 1 }{\tan (Ht)} - \frac{1}{Ht} \right) \left( -(n-3) \lambda E_{\ha}  -(n-4) \ttb t p^{d} H_{d \ha}  \right) \\
		& + \left(\frac{\sqrt{H}}{\sqrt{\sin (Ht)}} - \frac{1}{\sqrt{t}} \right) \left( h^{d c} [\bar A_c, H_{d \ha}] + t \ttb p^c (2 [\bar A_{c}, E_{\ha} ] - [\bar A_{\ha}, E_c]  )  \right), \\
		\widehat{\Delta}_{2\ha}^{h} = {} &   \left( \frac{ 1 }{\tan (Ht)} - \frac{1}{Ht} \right) h^{h d} H_{d \ha} + \left(\frac{\sqrt{H}}{\sqrt{\sin (Ht)}} - \frac{1}{\sqrt{t}} \right) h^{h c}  ( [\bar A_{c}, E_{\ha} ] - [\bar A_{\ha}, E_{c} ]   ), \\
		\widehat{\Delta}_{3 \ha}  ={}&  \left( \frac{1}{2\tan (Ht)} - \frac{1}{2Ht}  \right) \bar A_{\ha} + \left(\frac{\sqrt{H}}{\sqrt{\sin (Ht)}} - \frac{1}{\sqrt{t}} \right) E_{\ha}.
	\end{align*}	
	Moreover, there exists constants $\iota>0$ and $R>0$ such that the maps $\Xi_i$, $\widehat{\Delta}_i$, $i=1,2,3$, are analytic for $(t,\mathbf{U})\in (-\iota,\frac{\pi}{H})\times B_R(0)$.	
\end{lemma}

\begin{remark} 
	An important point regarding the equation \eqref{e:maineq1} is that it is \textit{not} equivalent to the conformal Yang--Mills equations expressed in the temporal gauge. Indeed, the system \eqref{e:maineq1} originating from multiplying \eqref{Maxwell-div-F} on both sides by the spatial projection tensor $\tensor{\ulh}{^{a}_{b}}$ contains only part of the propagation Yang--Mills equations. The missing propagation equation (see \eqref{E:Maxwell-div-nu} and Remark \ref{rk-adding-propogation}) can be obtained from applying $\nu^c$ to \eqref{Maxwell-div-F}. 
	At the same time, it is true that if $(E_{a}, \, \bar A_b)$ solves the conformal Yang--Mills equations
	in the temporal gauge, then the triple $(E_{a}, \, H_{p q}, \, \bar A_b)$, where $H_{p q}$ is given in terms of $\bar A_a$ by \eqref{e:Hpq}, determines a solution of the \eqref{e:maineq1}. 
\end{remark}		

\begin{proof}[Proof of Lemma \ref{lem-maxwell-hyperbolic-0}]
	Multiplying \eqref{Maxwell-div-F} on both sides by the spatial projection tensor $\tensor{\ulh}{^{a}_{b}}$ gives
	\begin{align*}
		&  Q^{edc} \nb_c F_{d a} \tensor{\ulh}{^{a}_{b}}+ (- \nu^d g^{e c} + \nu^c g^{e d} ) \nb_c F_{d a} \tensor{\ulh}{^{a}_{b}}\\
		=&  \frac{ (n-3) }{\tan (Ht)} \nu^e g^{d c} \nu_c F_{d a} \tensor{\ulh}{^{a}_{b}}+  \nu^e X^d F_{d a} \tensor{\ulh}{^{a}_{b}}+ \nu^e g^{\ha \hb} \tensor{X}{^d_{\ha a}} F_{\hb d} \tensor{\ulh}{^{a}_{b}} \notag \\
		& - \frac{\sqrt{H}}{\sqrt{\sin (Ht)}} \nu^e g^{c d} [A_d, F_{c a} ] \ts{\ulh}{^{a}_{b}}.
	\end{align*}
	Noting that $ - \nu^d g^{e a} \nb_c F_{d a} \ts{\ulh}{^{c}_{b}} =  - \nu^d g^{e f} \ts{\ulh}{^{a}_{f}} \nb_c F_{d a} \ts{\ulh}{^{c}_{b}} =  - \nu^d g^{e q} \ts{\ulh}{^{f}_{q}} \ts{\ulh}{^{a}_{f}} \nb_c F_{d a} \ts{\ulh}{^{c}_{b}}$ due to the anti-symmetry of the Yang--Mills field, it follows from the Yang--Mills Bianchi equation \eqref{Maxwell-bianchi-F} that 
	\begin{align*}
		& (- \nu^d g^{e c} + \nu^c g^{e d} ) \nb_c F_{d a}  \tensor{\ulh}{^{a}_{b}}
		=  \nu^d g^{e c} (\nb_c F_{a d} + \nb_d F_{c a} )  \tensor{\ulh}{^{a}_{b}} \\
		= & - \nu^d g^{e c} \nb_a F_{d c}  \tensor{\ulh}{^{a}_{b}} + \frac{ 3 }{\tan (Ht)} \nu^d g^{e c}  \tensor{\ulh}{^{a}_{b}} \nu_{[d} F_{c a]} - \frac{\sqrt{H}}{\sqrt{\sin (Ht)}} \nu^d g^{e c} \ts{\ulh}{^{a}_{b}} [A_{[d}, F_{c a]} ]  \\
		= & - \nu^d g^{e a} \nb_c F_{d a} \ts{\ulh}{^{c}_{b}} - \frac{ 1 }{\tan (Ht)} ( g^{e d} + \nu^d  g^{e c} \nu_c ) F_{d a}   \tensor{\ulh}{^{a}_{b}}  - \frac{\sqrt{H}}{\sqrt{\sin (Ht)}} \nu^d g^{e c} \ts{\ulh}{^{a}_{b}} [A_{[d}, F_{c a]} ]\\
		= & - \nu^d g^{e q} \ts{\ulh}{^{f}_{q}} \ts{\ulh}{^{a}_{f}} \nb_c F_{d a} \ts{\ulh}{^{c}_{b}} - \frac{ 1 }{\tan (Ht)} g^{e c} \ts{\ulh}{^{d}_{c}} F_{d a}   \tensor{\ulh}{^{a}_{b}}  - \frac{\sqrt{H}}{\sqrt{\sin (Ht)}} \nu^d g^{e c} \ts{\ulh}{^{a}_{b}} [A_{[d}, F_{c a]} ].
	\end{align*}
	Combining the above two expressions, we arrive at
	\begin{align}
		&\ts{\ulh}{^{f}_{b}} Q^{edc} \nb_c F_{d a} \ts{\ulh}{^{a}_{f}}  - \ts{\ulh}{^{c}_{b}} \nu^d g^{e q} \ts{\ulh}{^{f}_{q}} \nb_c F_{d a} \ts{\ulh}{^{a}_{f}} \notag \\
		= & \frac{ 1 }{H t} \ts{\ulh}{^{f}_{b}} ((n-3) \nu^e g^{d c} \nu_c + g^{e c} \ts{\ulh}{^{d}_{c}}) F_{d a} \ts{\ulh}{^{a}_{f}} \nnb \\
		&+ \frac{1}{\sqrt{t}} \left( - \ts{\ulh}{^{a}_{b}} \nu^e g^{c d} [A_d, F_{c a} ]  +  \ts{\ulh}{^{a}_{b}} \nu^d g^{e c} [A_{[d}, F_{c a]} ] \right) + \tensor{\widehat{\Delta}}{^{\prime e}_b}, \label{Maxwell-0}
	\end{align}
	where
	\begin{align*}
		\tensor{\widehat{\Delta}}{^{\prime e}_b} ={} &  \nu^e X^d F_{d a}  \tensor{\ulh}{^{a}_{b}}+ \nu^e g^{\ha \hb} \tensor{X}{^d_{\ha a}} F_{\hb d} \ts{\ulh}{^{a}_{b}}   \\
		& +  \left( \frac{ 1 }{\tan (Ht)} - \frac{1}{Ht} \right) \left( (n-3) \nu^e g^{d c} \nu_c +  g^{e c} \ts{\ulh}{^{d}_{c}} \right) F_{d a} \ts{\ulh}{^{a}_{b}} \\
		& + \left(\frac{\sqrt{H}}{\sqrt{\sin (Ht)}} - \frac{1}{\sqrt{t}} \right) \ts{\ulh}{^{a}_{b}} \left( \nu^d g^{e c}  [A_{[d}, F_{c a]} ] -  \nu^e g^{c d}  [A_d, F_{c a} ] \right).
	\end{align*}

	Then noting the decomposition
	\begin{align*}
		F_{d a} \ts{\ulh}{^{a}_{f}} 
		={}&
		\begin{pmatrix}
			\nu_d, \ts{\ulh}{^{l}_{d}}
		\end{pmatrix}
		\begin{pmatrix}
			-\nu^p F_{p a} \ts{\ulh}{^{a}_{f}} \\ \ts{\ulh}{^q_l} F_{q a} \ts{\ulh}{^{a}_{f}}
		\end{pmatrix}
		=
		\begin{pmatrix}
			\nu_d, \ts{\ulh}{^{l}_{d}}
		\end{pmatrix}
		\begin{pmatrix}
			E_f \\ H_{l f}
		\end{pmatrix},
	\end{align*}
	we can act on \eqref{Maxwell-0} with
	$\begin{pmatrix} \nu_e \\ \ts{\ulh}{^{h}_{e}} \end{pmatrix}$ to get
	\begin{align}\label{eq-Max-1}
		&	\begin{pmatrix}
			\ts{\ulh}{^{f}_{b}} \nu_e Q^{edc} \nu_d +  \ts{\ulh}{^{c}_{b}} \nu_e g^{e q} \ts{\ulh}{^{f}_{q}} & \ts{\ulh}{^{f}_{b}} \nu_e Q^{edc} \ts{\ulh}{^{l}_{d}}  \\
			\ts{\ulh}{^{f}_{b}} \ts{\ulh}{^{h}_{e}} Q^{edc} \nu_d +  \ts{\ulh}{^{c}_{b}} \ts{\ulh}{^{h}_{e}} g^{e q} \ts{\ulh}{^{f}_{q}} & \ts{\ulh}{^{f}_{b}} \ts{\ulh}{^{h}_{e}} Q^{edc} \ts{\ulh}{^{l}_{d}}
		\end{pmatrix}
		\nb_c
		\begin{pmatrix}
			E_f \\ H_{l f}
		\end{pmatrix}  \notag \\
		&\hspace{0.5cm} =  \frac{1 }{Ht}
		\begin{pmatrix}
			-(n-3)  \ts{\ulh}{^{f}_{b}} \lambda   &
			- (n-4)  \ts{\ulh}{^{f}_{b}} \xi^l   \\
			0  & \ts{\ulh}{^{f}_{b}} h^{h l}
		\end{pmatrix}
		\begin{pmatrix}
			E_f \\ H_{l f}
		\end{pmatrix} \nnb \\
		&\hspace{1cm}	 + \frac{1}{\sqrt{t}} \p{\nu_e \nu^d g^{e c} [A_{[d}, F_{c a]} ] \ts{\ulh}{^{a}_{b}}  + g^{d c} [A_c, F_{d a} ] \ts{\ulh}{^{a}_{b}} \\  \ts{\ulh}{^{h}_{e}}  \nu^d g^{e c} \ts{\ulh}{^{a}_{b}} [A_{[d}, F_{c a]} ]
		}
		+
		\begin{pmatrix} \nu_e \tensor{\widehat{\Delta}}{^{\prime e}_b} \\  \ts{\ulh}{^{h}_{e}} \tensor{\widehat{\Delta}}{^{\prime e}_b}
		\end{pmatrix}.
	\end{align}
	But, since
	\begin{align*}
		\frac{1}{\sqrt{t}} \nu_e \nu^d g^{e c}  [A_{[d}, F_{c a]} ] \ts{\ulh}{^{a}_{b}}
		=  & \frac{1}{\sqrt{t}} \xi^c \nu^d ( [\bar A_{a}, F_{d c} ]  +  [\bar A_{c}, F_{a d} ] ) \ts{\ulh}{^{a}_{b}} \notag \\
		=  & \sqrt{t} \ttb p^c (- [\bar A_{b}, E_c]  +  [\bar A_{c}, E_b ] ) ,
		\\
		\frac{1}{\sqrt{t}} \ts{\ulh}{^{h}_{e}} \nu^d g^{e c}  [A_{[d}, F_{c a]} ] \ts{\ulh}{^{a}_{b}}
		= & \frac{1}{\sqrt{t}} h^{h c}  ( - [\bar A_{b}, E_{c} ]  +  [\bar A_{c}, E_{b} ] )
	\end{align*}
	and
	\begin{align*}
		\frac{1}{\sqrt{t}}  g^{d c} [A_c, F_{d a} ] \ts{\ulh}{^{a}_{b}}
		= &-\sqrt{t} \ttb p^c  \nu^d  [\bar A_{c}, F_{d a} ] \ts{\ulh}{^{a}_{b}} + \frac{1}{\sqrt{t}} h^{d c} [\bar A_c, F_{d a}] \ts{\ulh}{^{a}_{b}} \notag\\
		= & \sqrt{t} \ttb p^c [\bar A_{c}, E_b ] + \frac{1}{\sqrt{t}} h^{d c} [\bar A_c, H_{d b}]
	\end{align*}
	in the temporal gauge, equation
	\eqref{eq-Max-1} becomes
	\begin{align}\label{eq-YM-1}
		& 	\begin{pmatrix}
			\ts{\ulh}{^{f}_{b}} \nu_e Q^{edc} \nu_d +  \ts{\ulh}{^{c}_{b}} \nu_e g^{e q} \ts{\ulh}{^{f}_{q}} & \ts{\ulh}{^{f}_{b}} \nu_e Q^{edc} \ts{\ulh}{^{l}_{d}}  \\
			\ts{\ulh}{^{f}_{b}} \ts{\ulh}{^{h}_{e}} Q^{edc} \nu_d +  \ts{\ulh}{^{c}_{b}} \ts{\ulh}{^{h}_{e}} g^{e q} \ts{\ulh}{^{f}_{q}} & \ts{\ulh}{^{f}_{b}} \ts{\ulh}{^{h}_{e}} Q^{edc} \ts{\ulh}{^{l}_{d}}
		\end{pmatrix}
		\nb_c
		\begin{pmatrix}
			E_f \\ H_{l f}
		\end{pmatrix}\notag \\
		= 	& \small  \frac{1 }{Ht}
		\begin{pmatrix}
			-(n-3)  \ts{\ulh}{^{f}_{b}} \lambda   &
			- (n-4)  \ts{\ulh}{^{f}_{b}} \xi^l   \\
			0  & \ts{\ulh}{^{f}_{b}} h^{h l}
		\end{pmatrix}
		\begin{pmatrix}
			E_f \\ H_{l f}
		\end{pmatrix} 	+ \frac{1}{\sqrt{t}} \p{ h^{d c} [\bar A_c, H_{d b}] \\  h^{h c}  ( [\bar A_{c}, E_{b} ] - [\bar A_{b}, E_{c} ]   )
		}
		+
		\begin{pmatrix} \widehat{\Delta}_{1 b} \\  \tensor{\widehat{\Delta}}{^{h}_{2 b}}
		\end{pmatrix},
	\end{align}
	where
	\als{
		\widehat{\Delta}_{1 b} =  \nu_e \tensor{\widehat{\Delta}}{^{\prime e}_b} + \sqrt{t} \ttb p^c (- [\bar A_{b}, E_c]  + 2 [\bar A_{c}, E_b ] ) \AND
		\widehat{\Delta}^h_{2 b} = \ts{\ulh}{^{h}_{e}} \tensor{\widehat{\Delta}}{^{\prime e}_b}.	
	}
	
	On the other hand,  
	\begin{equation} \label{F-nu}
		\tilde F_{d a} \nu^{d} = \nu^{d} \nb_d \Ab_a - \nb_a \Ab_{d} \nu^{d} + \nu^{d} [\Ab_d, \Ab_a]
	\end{equation} 
	by the definition \eqref{def-F} of the Yang--Mills curvature, and so, it follows from \eqref{CONFG-F}, \eqref{CONFG-A}, \eqref{E:CAL1} and a  direct calculation that, in the temporal gauge \eqref{temporal}, \eqref{F-nu} reduces to 
	\begin{equation}\label{eq-YM-potential}
		- \ts{\ulh}{^{f}_{b}} \nu^{d} \nb_d (\bar A_f)  = \frac{1}{2Ht} \bar A_b + \frac{1}{\sqrt{t}} E_b + \widehat{\Delta}_{3 b}
	\end{equation}
	where $\widehat{\Delta}_{3b}$ is given by
	\begin{equation*}\label{def-Delta-A}
		\widehat{\Delta}_{3b}  = \left( \frac{1}{2\tan (Ht)} - \frac{1}{2Ht}  \right) \bar A_b + \left(\frac{\sqrt{H}}{\sqrt{\sin (Ht)}} - \frac{1}{\sqrt{t}} \right) E_b.
	\end{equation*}			
	Putting \eqref{eq-YM-1} and \eqref{eq-YM-potential} together completes the proof.
\end{proof}

\begin{remark}\label{rk-adding-propogation}
	Although $\A^{0}$ in  \eqref{e:maineq1} is positive and symmetric, the system  \eqref{e:maineq1} is not symmetric hyperbolic due to the non-symmetry of the operators $\A_{\ha}^c \ts{\ulh}{^{\hc}_c}$. To remedy this defect, we supplement \eqref{e:maineq1} with an additional equation and introduce a new variable in order to symmetrize it. We begin by appending to \eqref{e:maineq1} the
	dynamical equation for $E_a$ given by
	\al{Maxwell-div-nu}{
		g^{b a} \nb_b E_{a} &=  \frac{\ttb (n-3)t }{\tan (Ht)} p^b E_b +  (X^d E_{d} + \nu^c g^{b a} \tensor{X}{^d_{b c}} F_{a d}) - \frac{\sqrt{H}}{\sqrt{\sin (Ht)}} h^{a b} [\bar A_b, E_{a} ],
	}
	which is obtained by contracting \eqref{Maxwell-div-F} with $\nu^c$ to get 
	\begin{align*}
		g^{b a} \nb_b F_{a c} \nu^c  =  &\frac{ (n-3) }{\tan (Ht)} g^{b a} \nu_a F_{b c} \nu^c + \nu^c (X^d F_{d c} + g^{b a} \tensor{X}{^d_{b c}} F_{a d})  \notag \\
		& - \frac{\sqrt{H}}{\sqrt{\sin (Ht)}} g^{a b} [A_b, F_{a c} ] \nu^c,
	\end{align*}
	and noting
	\begin{align*}
		g^{b a} \nu_a F_{b c} \nu^c 
		= g^{\ha a} \tensor{\ulh}{^b_{\ha}}  \nu_a \tensor{\ulh}{^d_b}  F_{d c} \nu^c = \xi^b E_b \overset{\eqref{E:V}}{=} \ttb t p^b E_b.
	\end{align*}

	At this point, \eqref{e:maineq1} with \eqref{E:Maxwell-div-nu} appended to it is still not symmetric hyperbolic. To complete the symmetrization, we define
	\begin{equation*}
		\tE^a=-g^{e\hb}\tensor{\ulh}{^a_e}E_{\hb},
	\end{equation*}
	and note that the $E_b$ and $\tE^a$ are related by
	\be\label{eq-E-tE}
	\tE^a=-h^{a \hb} E_{\hb}  \AND
	E_b=-\tE^a g_{ab}-g^{e\hb}\nu_e E_{\hb}\nu^a g_{ab}.
	\ee
	It can then be verified by a straightforward calculation, which we carry out in the proof of the following lemma, that the system consisting of
	\eqref{e:maineq1} and \eqref{E:Maxwell-div-nu} can be cast in a symmetric hyperbolic form in terms of the variables $\tE^e$, $E_{d}$, $H_{a b}$, $ \bar A_s$. 
\end{remark}

\begin{lemma}\label{lem-FOSHS-Maxwell}
	If $(E_{a},\,  A_b)$ solves the conformal Yang--Mills equations \eqref{Maxwell-div-F}--\eqref{Maxwell-bianchi-F} in the temporal gauge \eqref{temporal}, then the quadruple $(\tE^e, E_{d}, \, H_{p q}, \, \bar A_s)$ defined via \eqref{decom-F}, \eqref{def-tE-1} and \eqref{def-MYM} solves the symmetric hyperbolic equation
	\begin{align}\label{Maxwell-FOSHS-1}
		- \check{\mathbf{A}}^{0} \nu^c \nb_{c} \p{\tE^e  \\E_{\hd}  \\ H_{\ha\hb} \\ \bar A_s} - \check{\mathbf{A}}^{ f} \ts{\ulh}{^{c}_{f}}  \nb_{c} \p{\tE^e  \\E_{\hd}  \\ H_{\ha\hb} \\ \bar A_s} = & \frac{1}{Ht}\check{\mathcal{B}} \p{\tE^e  \\ E_{\hd}  \\ H_{\ha\hb} \\ \bar A_s} + \frac{1}{\sqrt{t}} \p{ - \Xi_{1 \he} \\  \tensor{h}{^{d \ha} } \Xi_{1\ha} \\ - \tensor{h}{^{a \ha}} \Xi^{b}_{2 \ha}  \\ h^{r a} \Xi_{3 a}} +
		\p{\mathfrak{D}^\sharp_{1 \he}  \\  \mathfrak{D}_{2}^{\sharp  d} \\ \mathfrak{D}_{3}^{\sharp a b}  \\ \mathfrak{D}_{4}^{\sharp  r} },
	\end{align}
	where
	\begin{align*}
		\check{\mathbf{A}}^{0} ={}& \p{ -\lambda \ts{\ulh}{^{a}_{\he}} g_{b a}\tensor{\ulh}{^b_e} & - \lambda \nu^r g_{r s} \tensor{\ulh}{^{s}_{\he}} \xi^{\hd}  & 0& 0 \\
			- \lambda \nu^r g_{r s} \tensor{\ulh}{^{s}_{e}} \xi^d & \bigl[ -\lambda  h^{\hd d}  - \lambda \nu^r \nu^s g_{r s}  \xi^{d} \xi^{\hd} + 2 \xi^{d} \xi^{\hd} \bigr] & 0 & 0 \\
			0 & 0 &  h^{\ha a} h^{\hb b} & 0 \\
			0& 0 & 0&  h^{r s}}, \\
		\check{\mathbf{A}}^{f} \ts{\ulh}{^{c}_{f}}  ={}& \footnotesize \p{2\xi^c \tensor{\ulh}{^a_\he} g_{b a}\tensor{\ulh}{^b_e} & \bigl[ 2\xi^c \nu^r g_{r s}\tensor{\ulh}{^{s}_{\he}} +\tensor{\ulh}{^{c}_{ \he}} \bigr]\xi^{\hd}  & - h^{\ha c} \ts{\ulh}{^{\hb}_{\he}} & 0\\
			\bigl[ 2\xi^c \nu^r g_{r s} \tensor{\ulh}{^{s}_e}+  \tensor{\ulh}{^c_e}\bigr] \xi^{d} & \bigl[ 2 \nu^r \nu^s g_{r s} \xi^c  \xi^{\hd} \xi^{d} - 2 \xi^{(d} h^{\hd ) c} + 2\xi^c h^{\hd d} \bigr]  & - h^{\ha d} h^{\hb c}  & 0\\
			-  h^{a c} \ts{\ulh}{^{b}_{e}} & - h^{a \hd} h^{b c} & 0 & 0 \\
			0 & 0 & 0 & 0},\\
		\check{\mathcal{B}} ={}&  \p{-(n-3) \lambda \tensor{\ulh}{^{a}_{\he}} g_{b a} \tensor{\ulh}{^{b}_e}   & 0  & 0 & 0\\
			0 & - (n-3)\lambda h^{\hd d} & 0 & 0 \\
			0 & 0 &  h^{\ha a}  h^{\hb b} &0 \\
			0&0&0& \frac{1}{2} h^{r s} },		
	\end{align*}
	\begin{align*}
		\mathfrak{D}^\sharp_{1 \he}(t,\mathbf{U}) = {} & - \widehat{\Delta}_{1\he} + (2\ttb t p^c -\lambda\nu^c) \ts{\ulh}{^{\ha}_{\he}} g_{\ha a}  E_{b} \nb_c h^{b a}+\frac{n-4}{H} \ttb p^a H_{a \he}-\frac{n-3}{H} \ttb   \lambda p^{\hd} \nu^a g_{a \he} E_{\hd}, \\
		\mathfrak{D}_{2}^{\sharp  d}(t,\mathbf{U}) = {}& \tensor{h}{^{d \ha}} \widehat{\Delta}_{1\ha} + (2 \ttb t p^c -\lambda \nu^c) \ttb t \nu^b g_{b\ha} p^{d} E_{\hb} \nb_c h^{\ha \hb} + \ttb t p^{d} E_{\ha} \nb_c h^{c \ha}   \\
		& + 2 \ttb t p^d X^b E_b + 2 \ttb t p^{d} \nu^c (  g^{b s} \tensor{X}{^{\hc}_{s c}} H_{b \hc} + g^{r s} \tensor{X}{^b_{s c}} ( \nu_r  E_b - \nu_b E_r) ) \\
		& + 2 \ttb t \frac{\ttb (n-3)t }{\tan (Ht)} p^{d}  p^b E_b - 2 \ttb p^d \frac{\sqrt{H} t}{\sqrt{\sin (Ht)}} h^{a b} [\bar A_b, E_{a} ]-\frac{n-4}{H}\ttb h^{\hb d}  p^{\ha} H_{\ha\hb}, \\
		\mathfrak{D}_{3}^{\sharp  a b}(t,\mathbf{U})={}&-\tensor{h}{^{a \ha}} \widehat{\Delta}_{2 \ha}^{b} - h^{a c} E_{\ha} \nb_c h^{b \ha}, \quad \mathfrak{D}_{4}^{\sharp  r}(t,\mathbf{U}) = h^{r \ha} \widehat{\Delta}_{3 \ha},
	\end{align*}	
	and the maps $\Xi_i, \, \widehat{\Delta}_{i}, \, i=1,2,3$, are as defined in Lemma \ref{lem-maxwell-hyperbolic-0}.
	Moreover, there exist constants $\iota>0$ and $R>0$ such that the maps $\mathfrak{D}^\sharp_{1 \he}(t,\mathbf{U})$,  $\mathfrak{D}_2^{\sharp d}(t,\mathbf{U})$, $\mathfrak{D}_{3}^{\sharp a b}(t,\mathbf{U})$ and $\mathfrak{D}_{4}^{\sharp  r}(t,\mathbf{U})$ are analytic  for $(t,\mathbf{U})\in (-\iota,\frac{\pi}{H})\times B_R(0)$ and vanish for $\mathbf{U}=0$.   
\end{lemma}

\begin{remark}
	The symmetric hyperbolic nature of \eqref{Maxwell-FOSHS-1} is a consequence of the symmetry of the matrices $\check{\mathbf{A}}^{0}$ and $\check{\mathbf{A}}^{f} \ts{\ulh}{^{c}_{f}} $  in the pairs of indices $(e, \, \he)$, $(d, \, \hd)$, $(a, \, \ha)$, $(b, \, \hb)$ and $(r, \, s)$,
	and the fact that $\check{\mathbf{A}}^{0}$ is positive as can be easily verified.
\end{remark}

\begin{remark}\label{rk-equivalence-EYM-Fuchsian}
	As was the case for \eqref{e:maineq1}, equation \eqref{Maxwell-FOSHS-1} is not equivalent to the conformal Yang--Mills equations expressed in the temporal gauge \eqref{temporal}. This is because the relation  $\tE^e = -h^{ea} E_a$ cannot be recovered from a solution  $(\tE^e, E_{d},H_{a b},\bar A_s)$ of \eqref{Maxwell-FOSHS-1} even in $\tE^e = -h^{ea} E_a$ holds initially. Consequently, we cannot guarantee that the Yang--Mills equation \eqref{E:Maxwell-div-nu} will hold for a
	solution of  \eqref{Maxwell-FOSHS-1}. However, it is true that if  $(E_{a}, \, \bar A_b)$ solves the conformal Yang--Mills system in the temporal gauge, then the quadruple $(\tE^e, \, E_{d}, \, H_{p q}, \, \bar A_s)$, where $\tE^e = -h^{ea} E_a$ and $H_{p q}$ is given in terms of $\bar A_a$ via \eqref{e:Hpq},  will solve \eqref{Maxwell-FOSHS-1}. It is worth noting, as shown in \cite{LW2021b}, that  this situation does not arise when 
	the adapted temporal gauge and the associated $1+3$ splitting of Yang--Mills field is employed. In this case, the equivalence between the gauged reduced Einstein--Yang--Mills equations and a first order symmetric hyperbolic formulation can be established directly. Although in the present paper, we do not have such equivalence, solutions of our first order symmetric hyperbolic formulation on the Einstein--Yang--Mills equations can still be shown to determine solutions to the Einstein--Yang--Mills via a uniqueness argument; see \S\ref{mainthm-proof} for details.
\end{remark}

\begin{proof}[Proof of Lemma \ref{lem-FOSHS-Maxwell}]
	We start the derivation of \eqref{Maxwell-FOSHS-1} by considering the first equation from \eqref{eq-Max-1}, which reads
	\begin{align}\label{e:eq2}
		&\nu_e Q^{edc}\nu_d \nb_c E_{\ha}+\tensor{\ulh}{^c_\ha}\nu_e g^{e\hd} \nb_c E_{\hd}+ \nu_e Q^{edc} \tensor{\ulh}{^\hb_d} \nb_c H_{\hb \ha}\notag  \\
		&\hspace{2cm}=\frac{1}{Ht}\bigl[-(n-3) \lambda \tensor{\ulh}{^{\hd}_{\ha}} E_{\hd}-(n-4) \tensor{\ulh}{^{\hb}_{\ha}}\xi^{a}H_{a\hb}\bigr]+\frac{1}{\sqrt{t}} \Xi_{1\ha} + \widehat{\Delta}_{1 \ha}.
	\end{align}
	With the help of \eqref{eq-E-tE} and the identity
	\begin{align*}
		& - \nu_e Q^{edc}\nu_d \nb_c g_{a\ha} \tE^a - \nu_e Q^{edc}\nu_d \nb_c g_{a\ha} g^{\hc \hb} \nu_{\hc} E_{\hb} \nu^a
		= \nu_e Q^{edc} \nu_d g^{a\hb} E_{\hb} \nb_c g_{a\ha},
	\end{align*}
	we can write \eqref{e:eq2} as
	\begin{align*}
		&\nu_e Q^{edc}\nu_d \nb_c E_{\ha} =  \nu_e Q^{edc} \nu_d\nb_c(-\tE^a g_{a\ha}-g^{\hc\hb}\nu_{\hc}E_{\hb}\nu^a g_{a\ha})   \\
		= &  -\nu_e Q^{edc} \nu_d g_{a\ha} \nb_c \tE^a-\nu_e Q^{edc}\nu_d\nu_{\hc}\nu^a g^{\hc\hb}g_{a\ha}\nb_c E_{\hb} \notag \\
		& -\nu_e Q^{edc}\nu_d\nu_{\hc} \nu^a E_{\hb}g_{a\ha}\nb_c g^{\hc\hb}+ \nu_e Q^{edc}\nu_d g^{a\hb}E_{\hb} \nb_c g_{a\ha}.
	\end{align*}
	Noting that $\nu_e Q^{edc} \nu_d = -2\xi^c+\lambda\nu^c$, it then follows that
	\begin{align}\label{e:tE-1}
		&(2\xi^c-\lambda \nu^c) g_{a \ha}\nb_c  \tE^{a}+\bigl[(2\xi^c-\lambda \nu^c)\nu^a g_{a\ha} +\tensor{\ulh}{^c_\ha}\bigr]\nu_e g^{e\hd} \nb_c E_{\hd} -  h^{\hb c} \nb_c H_{\hb \ha}\notag  \\
		= {} & \frac{1}{Ht}\bigl[-(n-3)\tensor{\ulh}{^{\hd}_{\ha}}\lambda E_{\hd}-(n-4) \tensor{\ulh}{^{\hb}_{\ha}}\ttb t p^a H_{a\hb}\bigr]+\frac{1}{\sqrt{t}} \Xi_{1\ha} + \widehat{\Delta}_{1 \ha} \notag  \\
		& + (-2\xi^c+\lambda\nu^c) \left(\nu_{\hc} \nu^a g_{a\ha} \nb_c g^{\hc\hb} E_{\hb}  - g^{a\hb} \nb_c g_{a\ha} E_{\hb} \right).
	\end{align}
	Applying the projection operator  $\tensor{\ulh}{^\ha_e}$ to  \eqref{e:tE-1}, we find, after using \eqref{eq-E-tE} to replace the $E_{\hd}$ in the first singular term of the right hand side of \eqref{e:tE-1} and noting that $\nb_c g_{ab}=-g_{a\ha} g_{b\hb}\nb_c g^{\ha\hb}$,  $\xi^a=\ttb t p^a$ (recall \eqref{E:V}) and $(n-3) \lambda g^{\he\hd}\nu_{\he} \nu^a g_{a e} E_{\hd}=(n-3) \lambda \xi^{\hd} \nu^a g_{a e} E_{\hd}=(n-3) \ttb t \lambda p^{\hd} \nu^a g_{a e} E_{\hd}$, that 
	\begin{align}\label{e:eq3a}
		&(2\xi^c-\lambda \nu^c) \tensor{\ulh}{^\ha_e} g_{b \ha}\tensor{\ulh}{^b_a} \nb_c  \tE^{a}+\bigl[(2\xi^c-\lambda \nu^c)\nu^a g_{a\ha}\tensor{\ulh}{^\ha_e} +\tensor{\ulh}{^c_e} \bigr]\xi^{\hd}\nb_c  E_{\hd}- h^{\ha c} \nb_c H_{\ha e} \notag  \\
		= {} & \frac{ n-3 }{Ht} \lambda g_{a e} \tE^a  +\frac{1}{\sqrt{t}} \Xi_{1 e} - \mathfrak{D}^\sharp_{1 e}(t,\mathbf{U})
	\end{align}
	where
	\begin{align*}
		\mathfrak{D}^\sharp_{1 e}(t,\mathbf{U})  =& - \widehat{\Delta}_{1 e}
		- \tensor{\ulh}{^\ha_e} (-2\xi^c+\lambda\nu^c) \left( \nu_{\hc} \nu^a g_{a\ha} \nb_c g^{\hc\hd} E_{\hd} + E_{\hd} g_{\ha b}\nb_c g^{\hd b} \right)\notag  \\
		&+\frac{n-4}{H} \ttb p^a H_{a e}  -\frac{(n-3)}{H} \ttb   \lambda p^{\hd} \nu^a g_{a e} E_{\hd}.
	\end{align*}
	Applying the projection $\tensor{\ulh}{^e_\hb}$ to \eqref{e:eq3a}, we arrive at
	\begin{align}\label{Main-maxwell-tE-1}
		& (2\xi^c-\lambda \nu^c) \tensor{\ulh}{^\ha_\hb} g_{b \ha}\tensor{\ulh}{^b_a} \nb_c  \tE^{a}+\bigl[(2\xi^c-\lambda \nu^c)\nu^a g_{a\ha}\tensor{\ulh}{^\ha_\hb}  +\tensor{\ulh}{^c_\hb}  \bigr]\xi^{\hd}\nb_c  E_{\hd}- h^{\ha c} \nb_c H_{\ha \hb}  \notag  \\
		= {} & \frac{ n-3 }{Ht}  \lambda \tensor{\ulh}{^{\hd}_{\hb}} g_{\ha \hd} \tensor{\ulh}{^{\ha}_a} \tE^a   +\frac{1}{\sqrt{t}} \Xi_{1 \hb} -\mathfrak{D}^\sharp_{1 \hb}(t,\mathbf{U}),
	\end{align}
	which determines the first component of \eqref{Maxwell-FOSHS-1} and can be viewed as an evolution equation for the Yang--Mills field $\tE^{a}$.

	Next, we derive an evolution equation for $E_{b}$ that will comprise the second component of \eqref{Maxwell-FOSHS-1}. Although $E_b$ and $\tE^a$ are not independent since they are related by \eqref{eq-E-tE}, we treat them as independent fields in order to obtain a symmetric hyperbolic equation.  The evolution equation for $E_{b}$ is derived from 
	$\tensor{\ulh}{^c_a}\nu_e g^{e b} \nb_c \tE^a-\tensor{\ulh}{^c_a} \nu_e g^{e b} \nb_c \tE^a - g^{b\ha}\times \eqref{e:eq2}$, which, after using \eqref{eq-E-tE}, reads
	\begin{align*}
		& \tensor{\ulh}{^c_a}\nu_e g^{e b} \nb_c\tE^a-\tensor{\ulh}{^c_a}\nu_e g^{e b} \nb_c \left(-g^{\he \hb}\tensor{\ulh}{^a_{\he}}E_{\hb} \right) - g^{b\ha} \nu_e Q^{edc}\nu_d \nb_c E_{\ha} - g^{b\ha} \tensor{\ulh}{^c_\ha}\nu_e g^{e\hd} \nb_c E_{\hd} \\
		&\hspace{0.5cm}- g^{b\ha}  \nu_e Q^{edc} \tensor{\ulh}{^\hb_d} \nb_c H_{\hb \ha}
		=   \frac{1}{Ht}\bigl[ (n-3)\lambda g^{b\ha}\tensor{\ulh}{^{\hd}_{\ha}} E_{\hd}+(n-4)g^{b\ha} \tensor{\ulh}{^{\hb}_{\ha}}\xi^{a}H_{a\hb}\bigr]  \notag\\
		&\hspace{5cm}- \frac{1}{\sqrt{t}} g^{b \ha} \Xi_{1 \ha} - g^{b\ha} \widehat{\Delta}_{1 \ha}.
	\end{align*}
	A straightforward calculation then shows that the above equation is equivalent to 
	\begin{align}\label{e:eq4}
		& \tensor{\ulh}{^c_a}\nu_e g^{e b} \nb_c\tE^a+(\tensor{\ulh}{^c_d}\nu_e g^{e b}g^{d\ha} - g^{b\ha} \nu_e Q^{edc}\nu_d - g^{b d} \tensor{\ulh}{^c_d}\nu_e g^{e\ha}) \nb_c E_{\ha}   + g^{b\ha} h^{\hb c} \nb_c H_{\hb \ha}\notag  \\
		={}&\frac{1}{Ht}\bigl[ (n-3)\lambda g^{b\ha}\tensor{\ulh}{^{\hd}_{\ha}} E_{\hd}+(n-4)g^{b\ha} \tensor{\ulh}{^{\hb}_{\ha}}\xi^{a}H_{a\hb}\bigr] - \frac{1}{\sqrt{t}} g^{b \ha} \Xi_{1 \ha} - g^{b\ha} \widehat{\Delta}_{1 \ha}  \notag \\ 
		& -\tensor{\ulh}{^c_d}\nu_e g^{eb}E_{a}\nb_c g^{d a}.
	\end{align}
	
	From Lemma \ref{lem-identity} (i.e.,  \eqref{E:g-h}) 	and the identity
	\begin{align*}
		\tensor{\ulh}{^c_d}\nu_e g^{e b}g^{d\ha} = & \nu_e g^{e \hb}(\tensor{\ulh}{^b_{\hb}}-\nu^b\nu_{\hb})g^{d a}\tensor{\ulh}{^c_d}(\tensor{\ulh}{^{\ha}_a}-\nu_a\nu^{\ha})= (\xi^b-\lambda \nu^b ) (h^{c\ha}-\xi^c\nu^{\ha}),
	\end{align*}
	we observe that
	\begin{align*}
		& \tensor{\ulh}{^c_d}\nu_e g^{e b}g^{d\ha}   - g^{b\ha} \nu_e Q^{edc}\nu_d - g^{b d} \tensor{h}{^c_d}\nu_e g^{e\ha}
		=  \xi^b h^{c\ha}-\xi^{\ha} h^{cb}+2\xi^c h^{b\ha}-3\nu^{\ha}\xi^b\xi^c   \\
		& -\lambda \nu^b h^{c\ha} +\lambda \nu^{\ha }h^{cb}-\nu^b\xi^{\ha}\xi^c- \lambda \nu^c h^{b\ha}+\lambda \nu^c\nu^{\ha}\xi^b+\lambda \nu^c\nu^b\xi^{\ha} -\lambda^2\nu^c\nu^b\nu^{\ha}+2\lambda\xi^c\nu^{\ha}\nu^b.
	\end{align*}
	Substituting this into \eqref{e:eq4} yields
	\begin{align*}
		&	\tensor{\ulh}{^c_a}\nu_e g^{e b} \nb_c\tE^a+(\xi^b h^{c\hd}-\xi^{\hd} h^{cb}+2\xi^c h^{b\hd} -\lambda \nu^b h^{c\hd}  \\
		& -\nu^b\xi^{\hd}\xi^c-\lambda \nu^c h^{b\hd} +\lambda \nu^c\nu^b\xi^{\hd}  ) \nb_c E_{\hd} + g^{b\ha} h^{\hb c} \nb_c H_{\hb \ha} 	\notag \\
		= &  \frac{1}{Ht}\bigl[ (n-3)\lambda g^{b\ha}\tensor{\ulh}{^{\hd}_{\ha}} E_{\hd}+(n-4)g^{b\ha} \tensor{\ulh}{^{\hb}_{\ha}}\xi^{a}H_{a\hb}\bigr]  \notag \\
		& - \frac{1}{\sqrt{t}} g^{b \ha} \Xi_{1 \ha} - g^{b\ha} \widehat{\Delta}_{1 \ha} - \tensor{\ulh}{^c_d}\nu_e g^{eb}E_{\ha}\nb_c g^{d \ha},
	\end{align*}	
	where in deriving this we also used  $\nu^{\ha}E_{\ha}=0$, which holds by \eqref{decom-F}.
	Applying the projection $\tensor{\ulh}{^{\hc}_b} $ to this equation gives
	\begin{align*} 
		& \xi^{\hc}  \tensor{\ulh}{^c_a} \nb_c\tE^a + (\xi^{\hc} h^{c\hd}-\xi^{\hd} h^{c\hc}+2\xi^c h^{\hc\hd}  -\lambda \nu^c h^{\hc\hd}) \nb_c E_{\hd} +  h^{\hc\ha} h^{\hb c} \nb_c H_{\hb \ha}\notag  \\
		={}&\frac{1}{Ht}\bigl[ (n-3)\lambda h^{\hc\hd} E_{\hd}+(n-4)h^{\hc\hb} \xi^{a}H_{a\hb}\bigr] - \frac{1}{\sqrt{t}} \tensor{h}{^{\hc \ha}} \Xi_{1\ha}  - \tensor{h}{^{\hc \ha}} \widehat{\Delta}_{1\ha} -\tensor{\ulh}{^c_d} \xi^{\hc} E_{\ha}\nb_c g^{d \ha}.
	\end{align*}	 	
	Reformulating this equation as 
	\begin{align*}
		& \bigl[(2\xi^c-\lambda \nu^c)\nu^b g_{b\ha} \tensor{\ulh}{^{\ha}_a} +  \tensor{\ulh}{^c_a}\bigr] \xi^{\hc}\nb_c\tE^a - (2\xi^c-\lambda \nu^c)\nu^b g_{b\ha} \tensor{\ulh}{^{\ha}_a}  \xi^{\hc}\nb_c\tE^a \notag   \\
		&\qquad +(\xi^{\hc} h^{c\hd}-\xi^{\hd} h^{c\hc}+2\xi^c h^{\hc\hd}  -\lambda \nu^c h^{\hc\hd}) \nb_c E_{\hd} -  h^{\hc\ha} h^{\hb c} \nb_c H_{ \ha \hb} \notag  \\
		= {} & \frac{1}{Ht}\bigl[ (n-3)\lambda h^{\hc\hd} E_{\hd}+(n-4)h^{\hc\hb} \xi^{a}H_{a\hb}\bigr] - \frac{1}{\sqrt{t}} \tensor{h}{^{\hc \ha}} \Xi_{1\ha} - \tensor{h}{^{\hc \ha}} \widehat{\Delta}_{1\ha} -\tensor{\ulh}{^c_d} \xi^{\hc} E_{\ha}\nb_c g^{d \ha},
	\end{align*}	
	where we note that $h^{\hc\ha} h^{\hb c} \nb_c H_{\hb \ha}$ has been rewritten as $- h^{\hc\ha} h^{\hb c} \nb_c H_{ \ha\hb}$,  we then use \eqref{eq-E-tE} to replace $\tE^a$ with $E_b$ in the term $-(2\xi^c-\lambda \nu^c)\nu^b g_{b\ha} \tensor{\ulh}{^{\ha}_a}  \xi^{\hc}\nb_c\tE^a$ from the above equation to get  
	\begin{align*}
		& \bigl[(2\xi^c-\lambda \nu^c)\nu^b g_{b\ha} \tensor{\ulh}{^{\ha}_e}+  \tensor{\ulh}{^c_e}\bigr] \xi^{d}\nb_c\tE^e -  h^{d\ha} h^{\hb c} \nb_c H_{ \ha\hb} \notag   \\
		& + \bigl[(2\xi^c-\lambda \nu^c) \xi^{d} \nu^b g_{b s} h^{s\hd}   +(\xi^{d} h^{c\hd}-\xi^{\hd} h^{c d}+2\xi^c h^{d\hd}  -\lambda \nu^c h^{d\hd})\bigr] \nb_c E_{\hd} \notag  \\
		={} &\frac{1}{Ht}\bigl[ (n-3)\lambda h^{d\hd} E_{\hd}+(n-4)h^{d\hb} \xi^{\ha}H_{\ha\hb}\bigr] - \frac{1}{\sqrt{t}} \tensor{h}{^{d \ha}}  \Xi_{1 \ha}  - \mathfrak{D}_2^{\dagger d}(t,\mathbf{U})
	\end{align*}
	where
	\begin{align*}
		\mathfrak{D}_2^{\dagger d}(t,\mathbf{U}) = \tensor{h}{^{d \ha}} \widehat{\Delta}_{1\ha} + (2\xi^c-\lambda \nu^c)\nu^b g_{b\ha} \xi^{d}E_{\hb} \tensor{\ulh}{^{\ha}_e} \nb_c g^{e\hb} + \xi^{d} E_{\ha} \tensor{\ulh}{^c_{\hd}} \nb_c g^{\hd \ha}.
	\end{align*}
	Using
	\als{
		\nu^b g_{b s} h^{s \hd} =  &\nu^b g_{b s} (g^{s \hd} - \lambda \nu^s \nu^{\hd} + \xi^{s} \nu^{\hd} + \xi^{\hd} \nu^s) \notag \\
		= & \nu^{\hd} (1-\lambda \nu^b \nu^s g_{b s} + \nu^b \xi^s g_{b s} ) + \nu^b \nu^s g_{b s} \xi^{\hd}
	}
	and $\nu^{\hd} \nb_c E_{\hd} = 0$ then allows us to express the above equation as
	\begin{align}
		& \bigl[(2\xi^c-\lambda \nu^c)\nu^b g_{b\ha} \tensor{\ulh}{^{\ha}_e}+  \tensor{\ulh}{^c_e}\bigr] \xi^{d}\nb_c\tE^e -  h^{d\ha} h^{\hb c} \nb_c H_{ \ha\hb} \notag   \\
		& + \bigl[(2\xi^c-\lambda \nu^c) \xi^{d} \xi^{\hd} \nu^b \nu^s g_{b s}  +\xi^{d} h^{c\hd}-\xi^{\hd} h^{c d}+2\xi^c h^{d\hd}  -\lambda \nu^c h^{d\hd} \bigr] \nb_c E_{\hd} \notag  \\
		={} &\frac{1}{Ht}\bigl[ (n-3)\lambda h^{d\hd} E_{\hd}+(n-4)h^{d\hb} \xi^{\ha}H_{\ha\hb}\bigr] - \frac{1}{\sqrt{t}} \tensor{h}{^{d \ha}} \Xi_{1\ha} - \mathfrak{D}_2^{\dagger d}(t,\mathbf{U}). \label{e:eq4a-A}
	\end{align}

	Next, we observe that the coefficients of $\nb_c E_{\hd}$ in \eqref{e:eq4a-A} can be expressed as
	\begin{align*}
		& (2\xi^c-\lambda \nu^c) \xi^{d} \xi^{\hd} \nu^b \nu^s g_{b s} + \xi^{d} h^{c\hd}-\xi^{\hd} h^{c d}+2\xi^c h^{d\hd}  -\lambda \nu^c h^{d\hd}  \\
		={}&\bigl[ (2\xi^c-\lambda \nu^c)\xi^{\hd} \xi^{d} \nu^b \nu^s g_{b s} - 2 \xi^{(d} h^{\hd ) c} + 2\xi^c h^{\hd d} - \lambda \nu^c h^{\hd d} \bigr] + 2\xi^{d} h^{c \hd},
	\end{align*}
	where the first term in the bracket is symmetric in the indices $d$ and $\hd$ while the remaining term $2\xi^{d} h^{c \hd}$ is non-symmetric in $d$ and $\hd$.
	The non-symmetric term need to be addressed in order to obtain a symmetric hyperbolic equation. To handle this term,   we use \eqref{E:g-h} to write $\xi^{d} h^{c \hd} \nb_c E_{\hd}$ as
	\begin{align*}
		2 \xi^{d} h^{c \hd} \nb_c E_{\hd}   = & 2 \xi^{d} (g^{c \hd} - \lambda \nu^c \nu^{\hd} + \xi^c \nu^{\hd} + \xi^{\hd} \nu^c)  \nb_c E_{\hd}   \notag \\
		=&  2 \xi^{d} g^{c \hd} \nb_c E_{\hd} +  2 \xi^{d} \xi^{\hd} \nu^c  \nb_c E_{\hd},
	\end{align*}
	where we note that the coefficient $2 \xi^{d} \xi^{\hd} \nu^c$ appearing on the  right hand side is symmetric in $d$ and $\hd$. 
	
	This leaves us to consider the term $2 \xi^{d} g^{c \hd} \nb_c E_{\hd}$. Making use of \eqref{E:Maxwell-div-nu} to express $g^{c \hd} \nb_c E_{\hd}$ as
	\begin{equation*}
		g^{c \hd} \nb_c E_{\hd} =  \frac{\ttb (n-3)t }{\tan (Ht)} p^b E_b - \frac{\sqrt{H}}{\sqrt{\sin (Ht)}} h^{a b} [\bar A_b, E_{a} ] + (X^d E_{d} +  \nu^c g^{b a} \tensor{X}{^d_{b c}} F_{a d}),
	\end{equation*}
	we see that $2 \xi^{d} g^{c \hd} \nb_c E_{\hd}$ is given by
	\als{
		2 \xi^{d} h^{c \hd} \nb_c E_{\hd}
		={}&  2 \xi^{d} \frac{\ttb (n-3)t }{\tan (Ht)} p^b E_b - 2 \xi^d \frac{\sqrt{H}}{\sqrt{\sin (Ht)}} h^{a b} [\bar A_b, E_{a} ] \nnb \\
		&+ 2 \xi^{d} (X^b E_{b} + \nu^b g^{r a} \tensor{X}{^s_{r b}} F_{a s}) +  2 \xi^{d} \xi^{\hd} \nu^c \nb_c E_{\hd},
	}
	where we note now that the principle term $2\xi^{d} \xi^{\hd} \nu^c \nb_c E_{\hd}$ is symmetric in $d$ and $\hd$.  With the help of this expression and the above arguments, it is clear that we can write \eqref{e:eq4a-A} as  
	\begin{align}\label{Main-maxwell-E-1}
		& \bigl[(2\xi^c-\lambda \nu^c)\nu^b g_{b \ha} \tensor{\ulh}{^{\ha}_e}+  \tensor{\ulh}{^c_e}\bigr] \xi^{d}\nb_c\tE^e  \notag   \\
		& + \bigl[(2\xi^c-\lambda \nu^c)\xi^{\hd} \xi^{d} \nu^b \nu^s g_{b s} +(2 \xi^{d} \xi^{\hd} \nu^c - \xi^{d} h^{c \hd} - \xi^{\hd} h^{c d} + 2\xi^c h^{\hd d} - \lambda \nu^c h^{\hd d} ) \bigr] \nb_c E_{\hd} \notag  \\
		& - h^{\ha d} h^{\hb c} \nb_c H_{ \ha\hb} =  \frac{ n-3 }{Ht} \lambda h^{\hd d} E_{\hd}   - \frac{1}{\sqrt{t}} \tensor{h}{^{d \ha}}  \Xi_{1\ha} - \mathfrak{D}_{2}^{\sharp  d}(t,\mathbf{U}),
	\end{align}
	where
	\begin{align*}
		\mathfrak{D}^{\sharp d}_{2}(t,\mathbf{U}) = {}& \mathfrak{D}^{\dagger  d}_{2}(t,\mathbf{U})  + 2 \ttb t p^{d} (X^b E_{b} + \nu^b g^{r a} \tensor{X}{^s_{r b}} F_{a s})+ 2 \ttb t p^{d} \frac{\ttb (n-3)t }{\tan (Ht)} p^b E_b  \\
		& - 2 \ttb p^d \frac{\sqrt{H} t}{\sqrt{\sin (Ht)}} h^{a b} [\bar A_b, E_{a} ]-\frac{n-4}{H}h^{\hb d} \ttb p^{\ha} H_{\ha\hb},
	\end{align*}
	which is the second component from \eqref{Maxwell-FOSHS-1}.

	Turning now to the derivation of the third equation from \eqref{Maxwell-FOSHS-1}, we 
	have from the second line of \eqref{eq-Max-1} that
	\begin{align}\label{eq-Max-3-1}
		& \tensor{\ulh}{^h_e} Q^{edc} \nu_d\nb_c E_{\ha} +  \tensor{\ulh}{^{c}_{\ha}} \tensor{h}{^{h a}}  \nb_c E_{a} + \tensor{\ulh}{^h_e} Q^{edc} \ts{\ulh}{^{\hb}_d} \nb_c H_{\hb\ha} \notag \\
		=&\frac{1}{Ht}\tensor{\ulh}{^{\hb}_{\ha}} h^{hb} H_{b\hb} + \frac{1}{\sqrt{t}} \Xi^h_{2 \ha}  + \widehat{\Delta}^h_{2 \ha}.
	\end{align} 
	Multiplying  \eqref{eq-Max-3-1} by $\tensor{h}{^{\hc \ha}}$, we find, with the help of the identities
	$\tensor{\ulh}{^{\hd}_e} Q^{edc} \nu_d= - h^{\hd c}$
	$\tensor{\ulh}{^\hd_e}Q^{edc} \tensor{\ulh}{^\hb_d}=  - \nu^c h^{\hb\hd}$
	from Lemma \ref{lem-identity},
	that
	\begin{align*}
		& -\tensor{h}{^{\hc \ha} } h^{\hd c} \nb_c E_{\ha} + h^{\hc c} h^{\hd \ha}  \nb_c E_{\ha} -\tensor{h}{^{\hc \ha}} h^{\hb\hd} \nu^c \nb_c H_{\hb\ha}  \notag \\
		= &  \frac{1}{Ht} h^{\hc\hb} h^{\hd d} H_{d\hb} + \frac{1}{\sqrt{t}} \tensor{h}{^{\hc \ha}} \Xi^{\hd}_{2 \ha} + \tensor{h}{^{\hc \ha}} \widehat{\Delta}^{\hd}_{2 \ha}.
	\end{align*}
	But by \eqref{eq-E-tE}, we have \[ h^{\hc c} h^{\hd \ha}  \nb_c E_{\ha}=-h^{\hc c}\nb_c \tE^{\hd}-h^{\hc c} E_{\ha} \nb_c h^{\hd \ha},\]
	and so, we conclude that
	\begin{align*}
		- h^{\hc\ha} h^{\hd c}\nb_c E_{\ha} - h^{\hc c}  \nb_c \tE^{\hd} - & \tensor{h}{^{\hc \ha}} h^{\hb\hd} \nu^c \nb_c H_{\hb\ha}   =\frac{1}{Ht} h^{\hc\hb} h^{\hd \ha} H_{\ha \hb} + \frac{1}{\sqrt{t}} \tensor{h}{^{\hc \ha}} \Xi^{\hd}_{2 \ha}   - \mathfrak{D}_3^{\sharp \hc \hd}(t,\mathbf{U}),
	\end{align*}
	where
	\[ \mathfrak{D}_3^{\sharp  \hc \hd}(t,\mathbf{U}) =- \tensor{h}{^{\hc \ha}} \widehat{\Delta}_{2\ha}^{\hd} - h^{\hc c} E_{\ha} \nb_c h^{\hd \ha}.\]
	Noting that $h^{a \hb} h^{b \ha} H_{\ha\hb} =-h^{a \ha} h^{b \hb} H_{\ha \hb}$, we see that the above equation is equivalent to the third component from  \eqref{Maxwell-FOSHS-1}, which is given by 
	\begin{align}\label{Main-maxwell-H-1}
		&- h^{a c}  \nb_c \tE^{b} - h^{a \hd} h^{b c}\nb_c E_{\hd} + h^{a \ha}h^{\hb b}\nu^c \nb_c H_{\ha \hb} \notag \\
		=&  -\frac{1}{Ht} h^{a \ha} h^{\hb b} H_{\ha \hb} + \frac{1}{\sqrt{t}}\tensor{h}{^{a \ha}} \Xi^{b}_{2 \ha}   - \mathfrak{D}_3^{\sharp a b}(t,\mathbf{U}).
	\end{align} 
	The proof then follows from collecting \eqref{Main-maxwell-tE-1},  \eqref{Main-maxwell-E-1} and \eqref{Main-maxwell-H-1} together with the equation \eqref{eq-YM-potential}  for the gauge potential. 
\end{proof}

The last step needed to bring the Yang--Mills equations into a form that will be favourable for our analysis involves multiplying each line of \eqref{Maxwell-FOSHS-1} by $\ulh^{\hb \he}$,  $\ulh_{d f}$, $\ulh_{a \bar a} \ulh_{b \bar b}$ and $\ulh_{o r}$, respectively. The resulting first order, symmetric hyperbolic Fuchsian equation is displayed in the theorem below. 

\begin{theorem}\label{thm-FOSHS-Maxwell}
	If $(E_{a},A_b)$ solves the conformal Yang--Mills equations \eqref{Maxwell-div-F}--\eqref{Maxwell-bianchi-F} in the temporal gauge \eqref{temporal}, then the quadruple $(\tE^e, E_{d}, \, H_{p q}, \, \bar A_s)$ defined via \eqref{decom-F}, \eqref{def-tE-1} and \eqref{def-MYM}
	solves the first order, symmetric hyperbolic Fuchsian equation
	\begin{align}\label{Maxwell-FOSHS}
		- \acute{\mathbf{A}}^0 \nu^c \nb_{c} \p{\tE^e  \\E_{\hd}  \\ H_{\ha\hb} \\  \bar A_s} + \acute{\mathbf{A}}^f \ts{\ulh}{^{c}_{f}}  \nb_{c} \p{\tE^e  \\E_{\hd}  \\ H_{\ha\hb} \\ \bar A_s} =\frac{1}{Ht}\acute{\mathcal{B}} \p{\tE^e  \\E_{\hd}  \\ H_{\ha\hb} \\ \bar A_s}+\acute{G}(t,\mathbf{U}),
	\end{align}
	where
	\begin{align*}
		\acute{\mathbf{A}}^0 ={}& \scriptsize \p{ -\lambda \ulh^{\hb \he} \tensor{\ulh}{^a_{\hb}} g_{a b}\tensor{\ulh}{^b_e} & - \lambda \nu^r g_{r s} \tensor{\ulh}{^{s \he}} \xi^{\hd}  & 0 & 0 \\
			- \lambda \nu^r g_{r s} \tensor{\ulh}{^{s}_{e}} \ulh_{df}\xi^d & \bigl[ -\lambda  h^{\hd d}  - \lambda \nu^r \nu^s g_{r s}  \xi^{d} \xi^{\hd} + 2 \xi^{d} \xi^{\hd} \bigr] \ulh_{d f} & 0 & 0 \\
			0 & 0 & \ulh_{a \bar a} \ulh_{b \bar b} h^{\ha a} h^{\hb b} & 0 \\
			0&0&0& \ulh_{o r} h^{rs} }, \\
		\acute{\mathbf{A}}^f \ts{\ulh}{^{c}_{f}} ={}& \notag  \\
		&\hspace{-1cm} \scriptsize \p{ -2\xi^c \tensor{\ulh}{^a_{s}} g_{a b}\tensor{\ulh}{^b_e} \ulh^{s \he} & -\bigl[ 2\xi^c \nu^r g_{r s}\tensor{\ulh}{^{s \he}} +\tensor{\ulh}{^{c \he}} \bigr]\xi^{\hd}  &  \ulh^{\he \hb}h^{\ha c} &0 \\
			- \bigl[ 2\xi^c \nu^r g_{r s} \tensor{\ulh}{^{s}_e}+  \tensor{\ulh}{^c_e}\bigr] \ulh_{d f} \xi^{d} & - \bigl[ 2\nu^r \nu^s g_{r s} \xi^c \xi^{\hd} \xi^{d} - 2 \xi^{(d} h^{\hd) c} + 2\xi^c h^{\hd d} \bigr] \ulh_{d f} &  \ulh_{d f} h^{\ha d} h^{\hb c} &0 \\
			\ulh_{e \bar b} \ulh_{a \bar a} h^{a c} & \ulh_{a \bar a}  \ulh_{b \bar b} h^{a \hd} h^{b c} & 0 &0 \\
			0&0&0&0}, \\
		\acute{\mathcal{B}}={}& \p{- (n-3) \lambda \ulh^{\hb \he} \tensor{\ulh}{^a_{\hb}} g_{a b}\tensor{\ulh}{^b_e}   & 0  & 0 & 0\\
			0 &  -(n-3)\lambda h^{\hd d} \ulh_{d f} & 0 & 0\\
			0 & 0 & \ulh_{a \bar a}  \ulh_{b \bar b} h^{\ha a}  h^{\hb b}&0 \\
			0&0&0&\frac{1}{2} \ulh_{or}h^{rs} },		
	\end{align*}
	\begin{align*}
		\acute{G}(t,\mathbf{U}) &= \acute{G}_0(t,\mathbf{U}) + \frac{1}{\sqrt{t}} \acute{G}_1(t,\mathbf{U}),\\ \acute{G}_0(t,\mathbf{U}) &= \left(\mathfrak{D}^{  \he}_1(t,\mathbf{U}),  \mathfrak{D}_{2 f}(t,\mathbf{U}), \mathfrak{D}_{3 \bar a\bar b}(t,\mathbf{U}),   \mathfrak{D}_{4 o} (t,\mathbf{U}), \right)^{\tr}, \\
		\acute{G}_1(t,\mathbf{U}) &= \left(- \ulh^{\he e} \Xi_{1 e}, \ulh_{df} h^{d \ha} \Xi_{1\ha}, - \ulh_{a \bar a} \ulh_{b \bar a} h^{a \ha} \Xi_{2\ha}^b, \ulh_{o r} h^{r a} \Xi_{3a} \right)^{\tr},
	\end{align*}
	the maps $\Xi_{1 e}$, $\Xi^b_{2\hat{a}}$, $\Xi_{3a}$ are as defined in Lemma \ref{lem-maxwell-hyperbolic-0}, and the
	maps $\mathfrak{D}_1^{ \he}$, $\mathfrak{D}_{2 f}$, $\mathfrak{D}_{3 \bar a\bar b}$ and $\mathfrak{D}_{4o}$ are defined by
	\begin{align*}
		\mathfrak{D}_1^{ \he}(t,\mathbf{U}) = {}& \ulh^{\he \hb} \mathfrak{D}^{\sharp}_{1 \hb} (t,\mathbf{U}), \quad
		\mathfrak{D}_{2 f}(t,\mathbf{U}) =  \ulh_{d f} \mathfrak{D}_{2}^{\sharp  d}(t,\mathbf{U}), \\
		\mathfrak{D}_{3 \bar a\bar b}(t,\mathbf{U})={}&\ulh_{a \bar a} \ulh_{b \bar b} \mathfrak{D}_{3}^{\sharp a b}(t,\mathbf{U}), \AND \mathfrak{D}_{4o} (t,\mathbf{U}) = \ulh_{o r} \mathfrak{D}_{4}^{\sharp  r}(t,\mathbf{U}),
	\end{align*}	
	respectively. Moreover, there 
	exist constants $\iota>0$ and $R>0$ such that the
	maps $\mathfrak{D}_1^{ \he}$, $\mathfrak{D}_{2 f}$, $\mathfrak{D}_{3 \bar a\bar b}$ and  $\mathfrak{D}_{4o}$ are analytic for $(t,\mathbf{U})\in (-\iota,\frac{\pi}{H})\times B_R(0)$ and vanish for $\mathbf{U}=0$.
\end{theorem}

\section{Symmetric hyperbolic Fuchsian  equations}\label{sec-Model}

In the previous section, we derived, in Theorems \ref{thm-FOSHS-m} and \ref{thm-FOSHS-Maxwell}, a Fuchsian formulation of the gauge reduced conformal Einstein--Yang--Mills equations. In this section, we quickly review the global existence theory for symmetric hyperbolic Fuchsian equations developed in \cite{Beyer2020}, which is an extension of the Fuchsian existence theory from \cite{Oliynyk2016a}; see also \cite{BeyerOliynyk:2020} for related results. This theory relies on the Fuchsian system satisfying a number of structural conditions, which we recall for the convenience of the reader. A tailored version of the main Fuchsian existence theorem from \cite{Beyer2020} is stated below in Theorem \ref{t:glex}. In the next section, we apply this theorem to our Fuchsian formulation of the gauge reduced conformal Einstein--Yang--Mills equations. The purpose of doing so is to obtain uniform bounds and decay estimates as $t\searrow 0$ for solutions to the conformal Einstein--Yang--Mills equations, where we recall that, in the conformal picture, $t=0$ corresponds to future timelike infinity.

Rather than considering the general class of Fuchsian systems analyzed in \cite{Beyer2020}, we instead consider a restricted class that encompasses our Fuchsian formulation of the Einstein--Yang--Mills system. This will allow us to simplify the presentation and application of the existence theory from \cite{Beyer2020} in our setting. 
The class of Fuchsian equations and the corresponding initial value problem that we consider in this section are of the form
\begin{align}
	-\mathbf{A}^0(u)\nu^c\nb_c u+\mathbf{A}^c(u)\tensor{\ulh}{^b_c} \nb_b u={}&\frac{1}{t}\mathfrak{A}(u)\mathbb{P} u+ G(u), \quad && \text{in }[T_0,0)\times \Sigma, \label{e:modeq}\\
	u={}&u_0,   \quad &&\text{in }\{T_0\} \times \Sigma,   \label{e:moddt}
\end{align}
where $T_0< 0$ and $\mathbf{A}^0=\mathbf{A}^c\nu_c$. Here\footnote{In this section only, we do not assume that $\Sigma$ is $\mathbb{S}^{n-1}$, although, in light of our application to the Einstein--Yang--Mills equations, it is fine to assume this.}, $\Sigma$ is a $(n-1)$-dimensional closed Riemannian manifold with time-independent metric $\ulh_{ab}$, $t$ is a Cartesian coordinate on the interval $[T_0,0)\subset \Rbb$, $\nb$ is the Levi-Civita connection of the Lorentzian metric\footnote{Note that we are now interpreting $\ulh_{ab}$ as spacetime tensor field on $\mathcal{M}$ in the obvious manner.}
\begin{equation*}
	\ulg_{ab}=-(dt)_a(dt)_b+\ulh_{ab}
\end{equation*}
on the spacetime manifold
\begin{equation*}
	\mathcal{M}=[T_0,0)\times \Sigma,
\end{equation*} $\nu_c=(dt)_c$ is the unit co-normal to $\Sigma$ and $\nu^c=\ulg^{cd}\nu_d=-(\partial/\partial t)^c$, $\ts{\ulh}{^{b}_{c}} =\ts{\delta}{^b_c}+\nu^b\nu_c= \ts{\delta}{^b_c}-(\partial/\partial t)^b (dt)_c$ is the projection onto $\ulg$-orthogonal subspace to $\nu^c$, and
\begin{equation*}
	u=(u_{(1)},\ldots,u_{(\ell)}) 
\end{equation*}
is a section of the vector bundle
\[\textbf{V}=\bigoplus^{\ell}_{k=1}h(T^{m_k}_{n_k}\mathcal{M})
\]
over $\mathcal{M}$, 
where we are using $h(T^{r}_{s}\mathcal{M})$ to denote the projection of the tensor bundle $T^{r}_{s}\mathcal{M}$ by $\ts{\ulh}{^{b}_{c}}$, i.e.
$\tensor{S}{^{a_1\cdots a_{r}}_{b_1\ldots b_s}} \in h(T^{r}_{s}\mathcal{M})$ if and only if 
$\ts{\ulh}{^{a_1}_{c_1}}\cdots \ts{\ulh}{^{a_r}_{c_r}} \ts{\ulh}{^{d_1}_{b_1}} \cdots \ts{\ulh}{^{d_s}_{b_s}} \tensor{S}{^{c_1\cdots c_{r}}_{d_1\ldots d_s}}=\tensor{S}{^{a_1\ldots a_{r}}_{b_1\ldots b_s}}$.
The coefficients $\mathbf{A}^c(u)$, $\mathfrak{A}(u)$, $\mathbb{P}$ and $G(u)$ will be assumed to satisfy conditions $(1)$--$(5)$ from \S\ref{s:mdlasp1} below.

\begin{remark} \label{V-bundle-rem}
	\begin{enumerate}[(i)]
		$\;$
		
		\item The coefficients $\mathbf{A}^c$ and $G$  in \eqref{e:modeq} implicitly depend on the spacetime points $(t,x)\in \mathcal{M}$ via the section $u$ of $\mathbf{V}$. The dependence can be made explicit by locally representing $u$ in a vector bundle chart as $(t,x,\tilde{u}(t,x))$.
		\item We can naturally view $u$ as a time-dependent section of the vector bundle 
		\begin{equation*}
			V= \bigoplus^{\ell}_{k=1} T^{m_k}_{n_k}\Sigma
		\end{equation*}
		over $\Sigma$. Since, in this viewpoint, $u$ no longer includes $t$ in its base point, i.e. $u$ is locally of the form $(x,\tilde{u}(t,x))$ with $(t,x)\in \mathcal{M}$, we will interpret the coefficients $\mathbf{A}^c$ and $G$ as depending on $t$ as well as $u$, and we will write
		$\mathbf{A}^c(t,u)$ and $G(t,u)$. This interpretation also allows us to write the Fuchsian equation \eqref{e:modeq} as
		\begin{equation*}
			\mathbf{A}^0(t,u)\partial_t u+\mathbf{A}^c(t,u) \ts{\ulh}{^{b}_{c}} \underline{D}_b u =\frac{1}{t}\mathfrak{A}(t,u)\mathbb{P} u+ G(t,u)
		\end{equation*}
		where $\mathbf{A}^0 = -\mathbf{A}^c\nu_c$ and  $\underline{D}$ denotes the Levi-Civita connection of the Riemannian metric $\ulh_{ab}$ on $\Sigma$. 
	\end{enumerate}
\end{remark}

For the remainder of this section, we will favour the interpretation of $u$ taking values in the vector bundle $V$ in line with Remark \ref{V-bundle-rem}.(ii).

\subsection{Symmetric linear operators and inner products}\label{s:SLO}
For fixed $t$, $u$ and $\eta_b$ , $\mathbf{A}^c(t,u)\ulh^b{}_c\eta_b$ and $\mathbf{A}^0(t,u)$ define linear operators on $V$. 
One of the assumptions needed for the 
Fuchsian existence theory from \cite{Beyer2020} is that these operators are symmetric with respect to a given inner product on $V$. In the following, we will employ the inner product defined on elements
\begin{equation*}
	v=(v_{(1)},\ldots,v_{(\ell)}),\;  u=(u_{(1)},\ldots,u_{(\ell)}) \in V
\end{equation*}
by
\begin{align}\label{e:inprod}
	\la v,u \ra_{\ulh}= & \sum^{\ell}_{k=1}\Bigl(\prod_{i=1}^{m_k}  \ulh_{c_{i}b_{i}}\Bigr) \Bigl(\prod_{j=1}^{n_k}  \ulh^{d_{j}a_{j}}\Bigr) \tensor{(v_{(k)})}{^{c_1\cdots c_{m_k}}_{d_1\ldots d_{n_k}}} \tensor{(u_{(k)})}{^{b_1\cdots b_{m_k}}_{a_1\ldots a_{n_k}}}.
\end{align}
For use below, we define projections $\Pbb_{(j)}:V \rightarrow V$ and maps $\phi_{(j)}: \Ima \Pbb_{(j)} \rightarrow T^{m_j}_{n_j}\Sigma$ by
\begin{align*}
	&\Pbb_{(j)}u
	= (0,\cdots,u_{(j)}, \cdots 0) \quad \text{and} \quad  
	\phi_{(j)} (0,\cdots,u_{(j)}, \cdots 0)= u_{(j)}, 
\end{align*}
respectively, 
and we set $\widetilde{\Pbb}_{(j)}=\phi_{(j)}\circ \Pbb_{(j)}$. We also note using these maps that $u
=\sum^\ell_{k=1}\phi^{-1}_{(k)}\widetilde{\Pbb}_{(k)}u$.

\begin{definition}\label{t:trps}
	The \textit{transpose} of\footnote{Here, $L(V)$ denotes the set of linear operators on $V$, which is isomorphic to $V^*\otimes V$.} $\mathbf{A}\in L(V)$, denoted $\mathbf{A}^{\text{tr}}$, is the unique element of $L(V)$ satisfying 
	\begin{align*}
		\la v, \mathbf{A} u\ra_{\ulh}=\la \mathbf{A}^{\text{tr}} v, u \ra_{\ulh}
	\end{align*}
	for all $u,v\in V$. Moreover, we say that $\mathbf{A}\in L(V)$ is \textit{symmetric} if $\mathbf{A}^{\text{tr}} = \mathbf{A}$.
\end{definition}

Given $\mathbf{A}\in L(V)$, we can represent it in block form as $\mathbf{A}=(\mathbf{A}_{(kl)})$
where the blocks $\mathbf{A}_{(kl)}$
are defined by
\begin{equation*}
	\mathbf{A}_{(kl)}=  \widetilde{\Pbb}_{(k)} \mathbf{A} \phi_{(l)}^{-1}.
\end{equation*} 
Using this notation, it can then be verified by a straight forward calculation that the transpose of $\mathbf{A}=(\mathbf{A}_{(kl)})$ is given by
\begin{align}\label{e:Asym}
	& \tensor{ \bigl( (\mathbf{A}^{\text{tr}})_{(lk)}\bigr)}{^{\he_1\cdots\he_{m_l} d_1\cdots d_{n_k}  }_{e_1\cdots e_{n_l} c_1\cdots c_{m_k}} }= \notag  \\
	& \hspace{1cm} \tensor{ \bigl(\mathbf{A}_{(kl)} \bigr)}{^{b_1\cdots b_{m_k}  \ha_1\cdots \ha_{n_l} }_{a_1\cdots a_{n_k}  \hb_1\cdots \hb_{m_l}} } \Bigl(\prod_{i=1}^{m_k}  \ulh_{c_{i}b_{i}}\Bigr)  \Bigl(\prod_{i=1}^{n_l}  \ulh_{\ha_{i}e_{i}}\Bigr) \Bigl(\prod_{j=1}^{n_k}  \ulh^{d_{j}a_{j}}\Bigr)   \Bigl(\prod_{j=1}^{m_l}  \ulh^{\hb_{j}\he_{j}}\Bigr).
\end{align}

\subsection{Coefficient assumptions}\label{s:mdlasp1}
Here, we state a tailored version of the coefficient assumptions from  \cite[\S $3.1$]{Beyer2020}. These assumptions need to be satisfied in order to apply the Fuchsian existence results \cite{Beyer2020}. In the following, we employ  the order notation, i.e. $\mathrm{O}$ and $\mathcal{O}$, from  \cite[\S 2.4]{Beyer2020}.

\medskip

\begin{enumerate}[(1)]
	
	\item The map $\Pb$ is a time-independent section of the vector bundle $L(V)$ over $\Sigma$ that is covariantly constant and defines a symmetric projection operator on $V$, i.e.
	$$ \Pb^2 = \Pb, \quad \Pb^{\text{tr}} = \Pb, \quad \partial_t \Pb=0 \quad \text{and} \quad \underline{D} \Pb = 0.$$
	
	\medskip
	
	\item There exist constants $R,\kappa, \gamma_1, \gamma_2>0$ such that the maps
	\begin{equation*}
		\mathbf{A}^0 \in C^1([T_0, 0], C^\infty(B_R(V), L(V) )) \AND \mathfrak{A} \in C^0([T_0, 0], C^\infty(B_R(V), L(V) ))
	\end{equation*}
	satisfy\footnote{Note here in the following $\pi$ is used to denote a vector bundle projection map. The particular vector bundle will be clear from context and so no confusion should arise from the use of the same symbol for all the vector bundle projections.} $\pi(\mathbf{A}^0(t, v)) = \pi(\mathfrak{A}(t, v)) = \pi(v)$ and 
	\begin{align}\label{e:coefcp}
		\frac{1}{\gamma_1} \la u,u \ra_{\ulh} \leq \la u, \mathbf{A}^0(t, v)u \ra_{\ulh} \leq \frac{1}{\kappa}  \la u,\mathfrak{A} (t ,v)u\ra_{\ulh} \leq \gamma_2 \la u,u \ra_{\ulh}
	\end{align}
	for all
	$(t,u,v) \in [T_0, 0) \times V \times B_R(V)$.
	Moreover, $\mathbf{A}^0$ satisfies the relations
	\begin{gather*}
		\mathbf{A}^{0}(t, v)^{\text{tr}}= \mathbf{A}^0(t, v),\\
		[\Pbb(\pi(v)), \mathfrak{A}(t, v)] = 0,   \\
		\Pbb(\pi(v))\mathbf{A}^0(t,v)\Pbb^\perp(\pi(v))=\mathrm{O}\bigl( \Pbb(\pi(v))v\bigr)
		\intertext{and}
		\Pbb^\perp(\pi(v))\mathbf{A}^0(t,v)\Pbb(\pi(v))=\mathrm{O}\bigl( \Pbb(\pi(v))v\bigr), \label{mA-2}
	\end{gather*}
	for all $(t,v)\in[T_0,0)\times B_R(V)$, and there exist maps\footnote{Here, $\Gamma(L(V))$ denotes the sections of the vector bundle $L(V)$ over $\Sigma$.} $\mathring{\mathbf{A}}^0, \, \mathring{\mathfrak{A}}\in C^0([T_0,0],\Gamma(L(V)))$ satisfying
	$[\Pbb,\mathring{\mathfrak{A}}]=0$,
	and
	\[\mathbf{A}^0(t,v)-\mathring{\mathbf{A}}^0(t,\pi(v))=\mathrm{O}(v) \AND
	\mathfrak{A}(t,v)-\mathring{\mathfrak{A}}(t,\pi(v))=\mathrm{O}(v)\]
	for all $(t,v)\in [T_0,0)\times B_R(V)$.
	
	\medskip

	\item
	The map $G(t, v) \in C^0([T_0, 0), C^\infty(B_R(V), V))$ admits an expansion of the form
	\begin{equation*}
		G(t, v ) = \mathring{G}(t, \pi(v)) + G_0(t, v) + |t|^{-\frac{1}{2}} G_1(t, v) + |t|^{-1} G_2(t, v)
	\end{equation*}
	where $\mathring{G}\in C^0([T_0, 0], \Gamma(V))$ and the maps $G_\ell(t, v) \in C^0([T_0, 0], C^\infty(B_R(V), V))$, $\ell=0,1,2$, satisfy $\pi(G_\ell(t, v)) = \pi(v)$ and \[ \Pb(\pi(v)) G_2 (t, v) =0 \] for all $(t,v)\in[T_0,0]\times B_R(V)$. Moreover, there exist constants $\lambda_\ell \geq 0$, $\ell=1,2,3$, such that 
	\begin{gather*}
		G_0 (t, v)  = \mathrm{O}(v), \quad
		\Pb(\pi(v)) G_1 (t, v)  = \mathcal{O}(\lambda_1 v),\quad
		\Pbp(\pi(v)) G_1 (t, v)  = \mathcal{O}(\lambda_2 \Pb(\pi(v)) v), \label{F1-restrict-2}
		\intertext{and}
		\Pbp(\pi(v)) G_2 (t, v)  = \mathcal{O}\Bigl(\lambda_3 R^{-1} \Pb(\pi(v)) v \otimes \Pb v\Bigr) \label{F2-restrict-2}
	\end{gather*}
	for all $(t,v)\in[T_0,0)\times B_R(V)$.
	
	\medskip
	
	\item The map $\mathbf{A}^c \ulh^{b}{}_{c} \in C^0([T_0, 0], C^\infty(B_R(V), L(V) \otimes T \Sigma ))$ satisfies $\pi(\mathbf{A}^c(t,v)\ulh^{b}{}_{c} ) = \pi(v)$
	and 
	\begin{equation*}
		[\sigma_c(\pi(v))\mathbf{A}^c(t,v)]^{\text{tr}}= \sigma_c(\pi(v)) \mathbf{A}^c(t,v)
	\end{equation*}  for all $(t, v) \in [T_0, 0) \times B_R(V)$ and spatial one forms $\sigma_a\in \mathfrak{X}^*(\Sigma)$.
	
	\medskip
	
	\item For each $(t,v)\in [T_0,0)\times B_R(V)$, there exists a $s_0>0$ such that
	\begin{equation*}
		\Theta(s)=\mathbf{A}^0\Bigl(t,v + s[\mathbf{A}^0(t,v)]^{-1}\Bigl(\frac{1}{t}\mathfrak{A}(t,v)\Pbb v+G(t,v)\Bigr)\Bigr), \quad |s|<s_0,
	\end{equation*}
	defines smooth curve in $L(V)$. There exist   
	constants $\theta$ and $\beta_\ell\geq 0$, $\ell=0,1,\ldots, 7$, such that the derivative $\Theta^\prime(0)$ satisfies\footnote{This condition is a reformulation of \cite[\S 3.1.\textrm{v}]{Beyer2020}. It is straightforward to check that it implies the condition \cite[\S 3.1.\textrm{v}]{Beyer2020} for the Fuchsian equation \eqref{e:modeq} that we are considering here}  
	\begin{align*}
		&\la v, \Pbb(\pi(v))\Theta^\prime(0)\Pbb(\pi(v)) v\ra_{\ulh} \notag \\
		=&\mathcal{O}\Bigl(\theta v\otimes v+|t|^{-\frac{1}{2}}\beta_0v\otimes \Pbb(\pi(v))v+|t|^{-1}\beta_1\Pbb(\pi(v))v\otimes \Pbb(\pi(v))v\Bigr),  \\
		&\la v, \Pbb(\pi(v))\Theta^\prime(0)\Pbp(\pi(v))v\ra_{\ulh}\notag \\
		=&\mathcal{O}\Biggl(\theta v\otimes v+|t|^{-\frac{1}{2}}\beta_2 v\otimes \Pbb(\pi(v))v+\frac{|t|^{-1}\beta_3}{R}\Pbb(\pi(v))v\otimes \Pbb(\pi(v))v\Biggr),  \\
		&	\la v, \Pbp(\pi(v))\Theta^\prime(0)\Pbb(\pi(v))v\ra_{\ulh}\notag \\
		=&\mathcal{O}\Biggl(\theta v\otimes v+|t|^{-\frac{1}{2}}\beta_4 v\otimes \Pbb(\pi(v))v +\frac{|t|^{-1}\beta_5}{R}\Pbb(\pi(v))v\otimes \Pbb(\pi(v))v\Biggr)
		\intertext{and}
		&		\la v, \Pbp(\pi(v))\Theta^\prime(0)\Pbp(\pi(v)) v\ra_{\ulh}\notag \\
		=&\mathcal{O}\Biggl(\theta v\otimes v +\frac{|t|^{-\frac{1}{2}}\beta_6}{R} v \otimes \Pbb(\pi(v))v +\frac{|t|^{-1}\beta_7}{R^2}\Pbb(\pi(v))v\otimes \Pbb(\pi(v))v\Biggr)
	\end{align*}	
	for all $(t,v)\in [T_0,0)\times B_R(V)$.
\end{enumerate}

\subsection{Global existence theorem}\label{s:glbex} Theorem 3.8 from \cite{Beyer2020} implies the global existence result for the Fuchsian initial value problem \eqref{e:modeq}--\eqref{e:moddt} that is stated in the following theorem. It is worth noting that the regularity for the initial data in the theorem below is less than what is required for \cite[Theorem~3.8]{Beyer2020}. This is because the spatial derivative term $\mathbf{A}^c \ts{\ulh}{^{b}_{c}} \nb_c$ appearing in \eqref{e:modeq} does not contain any $1/t$ singular terms by assumption. As a consequence, the use of \cite[Lemma~3.5]{Beyer2020} can be avoided in the proof of \cite[Theorem~3.8]{Beyer2020} under this assumption, which leads to the reduction in the required regularity of the initial data.

\begin{theorem}[Global existence theorem]\label{t:glex}
	Suppose $s\in\Zbb_{> \frac{n+1}{2}}$, $T_0<0$, $u_0\in H^s(\Sigma; V)$, the coefficient assumptions $(1)$--$(5)$ from  \S\ref{s:mdlasp1} are satisfied, and the constants $\kappa,\gamma_1,\lambda_3,\beta_0, \beta_{2k+1}$, $k=0,\ldots,3$ satisfy
	\begin{equation*}
		\kappa>\frac{1}{2}\gamma_1\Bigl(\sum_{k=0}^3\beta_{2k+1}+2\lambda_3\Bigr).
	\end{equation*}
	Then, there exists a constant $\delta>0$, such that if
	\begin{equation*}\label{e:data3}
		\max\Bigl\{\lVert u_0\rVert_{H^s},\sup_{T_0\leq \tau<0}\lVert  \mathring{G}(\tau)\rVert_{H^s}\Bigr\} \leq \delta,
	\end{equation*}
	then there is a unique solution
	\begin{equation*}\label{e:urgn}
		u\in C^0([T_0,0),H^s(\Sigma;V)) \cap C^1([T_0,0),H^{s-1}(\Sigma;V))\cap \Li ([T_0,0),H^s(\Sigma;V))
	\end{equation*}
	of the initial value problem \eqref{e:modeq}--\eqref{e:moddt} such that $\Pbp u(0):=\lim_{t\nearrow 0}\Pbp u(t)$ exists in $H^{s-1}(\Sigma;V)$.
	
	\medskip
	
	\noindent Moreover, the solution $u$ satisfies the energy estimate
	\begin{equation}\label{e:ineq1v}
		\lVert u(t)\rVert_{H^s}^2+\sup_{T_0\leq \tau<0}\lVert \mathring{G}(\tau)\rVert^2_{H^s} - \int^t_{T_0}\frac{1}{\tau}\lVert \Pbb u(\tau)\rVert^2_{H^s}d\tau \leq C(\delta,\delta^{-1})\bigl(\lVert u_0\rVert^2_{H^s}+\sup_{T_0\leq \tau<0}\lVert \mathring{G}(\tau)\rVert^2_{H^s} \bigr)
	\end{equation}
	and the decay estimates
	\begin{align}\label{e:Puest}
		\lVert \Pbb u(t)\rVert_{H^{s-1}}\lesssim \begin{cases}
			|t|+ \lambda_1 |t|^{\frac{1}{2}}, \quad &\text{if }\zeta>1 \\
			|t|^{\zeta-\sigma}+ \lambda_1 |t|^{\frac{1}{2}}, \quad &\text{if }\frac{1}{2}<\zeta\leq1 \\
			|t|^{\zeta-\sigma}, \quad &\text{if }0<\zeta\leq \frac{1}{2}
		\end{cases}
	\end{align}	
	and
	\begin{align}\label{e:Ppuest}
		\lVert \Pbp u(t)-\Pbp u(0)\rVert_{H^s}\lesssim \begin{cases}
			|t|^{\frac{1}{2}}+|t|^{\zeta-\sigma}, \quad &\text{if }\zeta>\frac{1}{2}\\
			|t|^{\zeta-\sigma}, \quad &\text{if }\zeta\leq \frac{1}{2}
		\end{cases}
	\end{align}
	for all $t\in [T_0,0)$ where $\zeta=\kappa-\frac{1}{2}\gamma_1\beta_1$.
\end{theorem}

\section{Global existence for the Fuchsian formulation of the EYM equations}\label{s:verif}

In this section, we carry out one of the main steps in the proof of the Theorem \ref{t:mainthm} by establishing the global existence of solutions to the Fuchsian equation obtained from combining the Fuchsian equations from Theorems \ref{thm-FOSHS-m} and \ref{thm-FOSHS-Maxwell}, which is given by
\begin{align}\label{e:FchEYM}
	- \widehat{\mathbf{A}}^0 \nu^c \nb_{c} \widehat{\mathbf{U}} + \widehat{\mathbf{A}}^b \ts{\ulh}{^{c}_{b}}   \nb_{c} \widehat{\mathbf{U}} =\frac{1}{Ht}\widehat{\mathcal{B}} \widehat{\mathbf{U}}  +\widehat{G}(t,\widehat{\mathbf{U}}),
\end{align}
where
\begin{gather}
	\widehat{\mathbf{A}}^0=\p{\bar{\mathbf{A}}_1^0 & 0 & 0 & 0 & 0\\ 0 & \bar{\mathbf{A}}_2^0 & 0 & 0 & 0\\0 & 0 & \bar{\mathbf{A}}_3^0 & 0 & 0\\0 & 0 & 0 & \bar{\mathbf{A}}_4^0 & 0\\
		0 & 0 & 0 & 0 & \acute{\mathbf{A}}^0 }, \quad \widehat{\mathbf{A}}^b \ts{\ulh}{^{c}_{b}} = \p{\bar{\mathbf{A}}_1^b \ts{\ulh}{^{c}_{b}}  &0&0&0&0 \\
		0 &\bar{\mathbf{A}}_2^b \ts{\ulh}{^{c}_{b}} &0&0&0 \\
		0 &0&\bar{\mathbf{A}}_3^b \ts{\ulh}{^{c}_{b}} &0&0 \\
		0 &0&0&\bar{\mathbf{A}}_4^b \ts{\ulh}{^{c}_{b}} &0 \\
		0 &0&0&0&\acute{\mathbf{A}}^b \ts{\ulh}{^{c}_{b}} }, \label{e:coef1}\\
	\widehat{\mathcal{B}}=\p{\bar{\mathcal{B} }_1 &0&0&0&0\\0&\bar{\mathcal{B} }_2 &0&0&0\\0&0&\bar{\mathcal{B} }_3 &0&0\\0&0&0&\bar{\mathcal{B} }_4 &0\\0&0&0&0&\acute{\mathcal{B}}}, \quad \widehat{G}(t,\widehat{\mathbf{U}})=\p{\bar{G}_1(t,\widehat{\mathbf{U}}) \\ \bar{G }_2(t,\widehat{\mathbf{U}}) \\ \bar{G}_3(t,\widehat{\mathbf{U}}) \\ \bar{G}_4(t,\widehat{\mathbf{U}}) \\
		\acute{G}(t,\widehat{\mathbf{U}})}, \label{e:coef2}
\end{gather}
and
\begin{align}\label{e:hatu}
	\widehat{\mathbf{U}}=&(-\nu^e m_e,\ts{\ulh}{^e_{\he}}m_e, m , -\nu^e \tp{a}{e}, \ts{\ulh}{^e_{\he}} \tp{a}{e}, p^a ,  -\nu^e \tss{\ha\hb}{e},  \notag\\
	&\hspace{2cm}\ts{\ulh}{^e_{\he}} \tss{\ha\hb}{e}, s^{\ha \hb}, -\nu^e s_{e}, \ts{\ulh}{^e_{\he}} s_{e}, s , \tE^e, E_{d}, H_{\ha\hb}, \bar A_s)^{\tr}.
\end{align}

While, by construction, solutions $(g_{ab},A_b)$  of the reduced conformal Einstein--Yang--Mills equations in the temporal gauge determine solutions of the Fuchsian system \eqref{e:FchEYM} via \eqref{decom-g}, \eqref{E:W}--\eqref{E:QD} and \eqref{def-tE-1}--\eqref{def-MYM}, we will, in this section, analyze general solutions to \eqref{e:FchEYM} and not, a priori, assume that they are derived from solutions of the conformal Einstein--Yang--Mills equations. The main purpose of establishing the existence of solutions to the Fuchsian system \eqref{e:FchEYM} will be to obtain global bounds on solutions of the reduced conformal Einstein--Yang--Mills equations. 
The details of this argument can be found in \S\ref{mainthm-proof}.

The main global existence theorem for the Fuchsian equation \eqref{e:FchEYM} is stated below in Theorem \ref{t:Uglobal}. In order to be able to apply this theorem, we need to establish that \eqref{e:FchEYM} verifies the coefficient assumptions $(1)$--$(5)$ from \S\ref{s:mdlasp1}. As we show below, this requires us to make certain dimension dependent choices for the free parameters $\tta, \ttb, \ttj, \ttk, \tte, \ttf$ that appear in the coefficients of \eqref{e:FchEYM}.

\subsection{Parameter selection }\label{sec-decom-Einstein}
The parameters  $\tta, \ttb, \ttj, \ttk, \tte, \ttf$ are determined by the spacetime dimension $n$ and the following decomposition for the matrix  $\widehat{\mathcal{B}}$ defined above by \eqref{e:coef2}:
\begin{align*}
	\widehat{\mathcal{B}} = \mathfrak{\widehat{A}}  \widehat{\Pbb}=\p{\mathfrak{\bar{A}}_1 & 0 & 0 &0&0\\
		0 & \mathfrak{\bar{A}}_2 &  0 &0&0\\
		0&0 & \mathfrak{\bar{A}}_3 &  0 &0\\
		0 &0&0 & \mathfrak{\bar{A}}_4 &  0\\
		0 &  0 &0 & 0 & \mathfrak{\acute{A}}   }\p{\bar{\Pbb}_1 & 0 & 0 &0&0\\
		0 & \bar{\Pbb}_2 &  0 &0&0\\
		0&0 & \bar{\Pbb}_3 &  0 &0\\
		0 &0&0 & \bar{\Pbb}_4 &  0\\
		0 &  0 &0 & 0 & \acute{\Pbb}  },
\end{align*}
which implies, in particular that 
$\bar{\mathcal{B}}_\ell = \mathfrak{\bar{A}}_\ell \bar{\Pbb}_\ell$, $\ell=1,\ldots, 4$, and $\acute{\mathcal{B}}=\acute{\mathfrak{A}}\acute{\Pbb}$. This decomposition is not unique and we have some freedom in choosing the particular form of the operators   $\mathfrak{\bar{A}}_\ell, \acute{\mathfrak{A}}$ and the projection operators $\bar{\Pbb}_\ell,\acute{\Pbb}$. We fix $\mathfrak{\bar{A}}_3$, $\mathfrak{\bar{A}}_4$, $\acute{\mathfrak{A}}$,  $\bar{\Pbb}_3$, $\bar{\Pbb}_4$ and $\acute{\Pbb}$ by setting  
\als{
	\mathfrak{\bar{A}}_3=\mathfrak{\bar{A}}_4=\p{-\lambda (n-2) & 0 & 0 \\
		0 & 1 & 0 \\
		0 & 0 & 1}, \quad \bar{\Pbb}_3=\bar{\Pbb}_4 =\p{1 & 0 & 0 \\
		0 & 0 & 0 \\
		0& 0 & 0}, \quad \acute{\mathfrak{A}} = \acute{\mathcal{B}} \AND  \acute\Pbb = \mathds{1},
}
respectively. The choice of the remaining operators $\mathfrak{\bar{A}}_1$, $\mathfrak{\bar{A}}_2$,  $\bar{\Pbb}_1$ and $\bar{\Pbb}_2$ will depend on the spacetime dimensions $n$, which we separate into the following three cases: $n=4$, $n\geq 6$ and $n=5$.

\noindent\underline{$(1)$ $n=4$:} For $n=4$, we fix the parameters
$\tta, \ttb, \ttj, \ttk, \tte, \ttf$
by setting
\als{
	\ttk=\ttj=\frac{2}{3}, \qquad \tta=\ttb=2H \AND \tte=\ttf=1.
}
Inserting these parameters into $\bar{\mathbf{A}}_\ell^0$ and $\bar{\mathcal{B}}_\ell$ (see Theorem \ref{thm-FOSHS-m}), we can ensure that $\bar{\mathcal{B}}_\ell=\mathfrak{\bar{A}}_\ell \Pbb_\ell$ for $\ell=1,2$ by setting
\al{Decom-n=4}{
	\mathfrak{\bar{A}}_1=\mathfrak{\bar{A}}_2=\p{-\lambda & 0 & 0 \\
		0 & \frac{3}{2} \ulh_{f d} h^{f \he} & 0 \\
		0 & 0 & -\lambda}
	\AND
	\bar{\Pbb}_1 =\bar{\Pbb}_2 =  \p{\frac{1}{2} & 0 & \frac{1}{2} \\
		0 & \ts{\delta}{^{d}_{\hc}} & 0 \\
		\frac{1}{2} & 0 & \frac{1}{2}}.
}

\begin{lemma}\label{e:ACOND1}
	There exist constants $\iota>0$, $R>0$, $\gamma_1>1$, and $\gamma_2>\frac{3}{2}$ such that
	\[\frac{1}{\gamma_1}\mathds{1} \leq \bar{\mathbf{A}}^0_\ell \leq   \mathfrak{\bar{A}}_\ell \leq \gamma_2 \mathds{1}\AND [\mathfrak{\bar{A}}_\ell, \, \bar{\Pbb}_\ell] =0   \]
\end{lemma}
for all $(t,\widehat{\mathbf{U}}) \in (-\iota,\frac{\pi}{H})\times B_R(0)$ and  $\ell=1,2$.
\begin{proof} 
	It is straightforward to verify from the \eqref{E:Decom-n=4} and the definitions of the matrices $\bar{\mathbf{A}}^0_\ell$ from Theorem \ref{thm-FOSHS-m} that  $\bar{\mathbf{A}}^0_\ell \leq   \mathfrak{\bar{A}}_\ell$ and $[\mathfrak{\bar{A}}_\ell, \, \bar{\Pbb}_\ell]=0$. Moreover, by taking $R>0$ sufficiently small, it is also not difficulty to verify that there exists constants $\gamma_1>1$ and $\gamma_2>\frac{3}{2}$ such that $\mathds{1}\leq \gamma_1  \bar{\mathbf{A}}^0_\ell $ and $\mathfrak{\bar{A}}_\ell \leq \gamma_2\mathds{1}$, which completes the proof.
\end{proof}	

\bigskip

\noindent \underline{$(2)$ $n \geq 6$:} For $n\geq 6$, we fix the parameters
$\tta, \ttb, \ttj, \ttk, \tte, \ttf$
by setting
\al{PARASET6}{
	\frac{1}{\ttj} =n-\frac{7}{2}, \qquad \tte=1,\quad \frac{H}{\tta}=\sqrt{\frac{1}{2}\left(n-\frac{11}{2}\right)},
}
\al{PARASETk6}{
	\frac{1}{\ttk}=\frac{n-1}{2}, \qquad \ttf=1 \AND \frac{H}{\ttb}=\frac{n-3}{2}.
}
Substituting these choices into $\bar{\mathbf{A}}_\ell^0$ and $\bar{\mathcal{B}}_\ell$, we find that  $\bar{\mathcal{B}}_\ell=\mathfrak{\bar{A}}_\ell \Pbb_\ell$,
$\ell=1,2$, where we have set
\al{AANDP06}{
	\mathfrak{\bar{A}}_1 =\p{
		\frac{3}{2} (-\lambda) & 0 & \sqrt{\frac{1}{2}\left(n-\frac{11}{2}\right)}(-\lambda)  \\
		0 & \left(n-\frac{7}{2}\right) \ulh_{f d} h^{f \he} & 0\\		
		\sqrt{\frac{1}{2}\left(n-\frac{11}{2}\right)} (-\lambda) & 0 & (n-\frac{9}{2})(-\lambda)
	}, \quad  \bar{\Pbb}_1=\mathds{1},
}
\al{AANDP06b}{
	\mathfrak{\bar{A}}_2=\p{(n-3)(-\lambda) & 0 & 0\\
		0 & \frac{1}{2}(n-1) \ulh_{f d}h^{f \he} & 0 \\
		0 & 0 & (n-3)(-\lambda)
	} \AND \bar{\Pbb}_2=\p{\frac{1}{2} & 0 & \frac{1}{2}\\
		0 & \ts{\delta}{^d_{\hd}} & 0\\
		\frac{1}{2} & 0 & \frac{1}{2} }.
}

\begin{lemma}\label{e:ACOND0}
	There exist constants $\iota>0$, $R>0$, $\gamma_1>1$,  $\gamma_2>2n-10$, $\tilde{\gamma}_1>1$ and  $\tilde{\gamma}_2>\frac{2(n-3)}{n-1}$ such that
	\begin{gather*}
		\frac{1}{\gamma_1} \mathds{1} \leq \bar{\mathbf{A}}^0_1\leq   \mathfrak{\bar{A}}_1 \leq \gamma_2 \mathds{1},\quad \frac{1}{\tilde{\gamma}_1}\mathds{1} \leq \bar{\mathbf{A}}^0_2 \leq  \frac{2}{n-1} \mathfrak{\bar{A}}_2 \leq \tilde{\gamma}_2 \mathds{1}
		\intertext{and}
		[\mathfrak{\bar{A}}_\ell, \bar{\Pbb}_\ell] =0
	\end{gather*}
	for  all  $(t,\widehat{\mathbf{U}}) \in (-\iota,\frac{\pi}{H})\times B_R(0)$ and $\ell=1,2$.
\end{lemma}
\begin{proof}	
	By \eqref{E:AANDP06}, we note, for any  $\omega=(\omega_1, \omega_2, \omega_3)^{\tr}$, that
	\gas{
		\la\mathfrak{\bar{A}}_1 \omega, \omega\ra_{\ulh}= -\lambda \frac{3}{2} \omega_1^2+ \left(n-\frac{7}{2}\right) h^{a b} \omega_{2 a} \omega_{2 b} +(-\lambda) \left(n-\frac{9}{2}\right) \omega_3^2-2\lambda \sqrt{\frac{1}{2}\left( n-\frac{11}{2}\right)}\omega_1\omega_3,
	}
	where in deriving this we have employed the relation $\omega_{2a} = \ulh_{a \ha} \omega_2^{\ha}$. But, by Young's inequality (see Lemma \ref{t:young}) 
	\begin{equation*}
		2\sqrt{\frac{1}{2}\left( n-\frac{11}{2}\right)} |\omega_1\omega_3| \leq  \frac{1}{2}\omega_1^2+\left( n-\frac{11}{2}\right) \omega_3^2,
	\end{equation*}
	and so we deduce that
	\als{
		- \lambda \omega_1^2 + \left(n-\frac{7}{2}\right) h^{a b} \omega_{2a} \omega_{2 b} + & (-\lambda) \omega_3^2
		\leq \la\bar{\mathfrak{A}}_1 \omega, \omega\ra_{\ulh} \\ \leq & - 2\lambda \omega_1^2 + \left(n-\frac{7}{2}\right) h^{a b} \omega_{2a} \omega_{2 b} + (-\lambda) (2n-10) \omega_3^2,
	}
	where in deriving this we have used the fact that $\lambda < 0$, see \eqref{decom-g}.
	
	On the other hand, from the definition of $\bar{\mathbf{A}}^0_1$ in Theorem \ref{thm-FOSHS-m}, we observe that $\la\bar{\mathbf{A}}^0_1 \omega, \omega\ra_{\ulh}= -\lambda \omega_1^2+ h^{a b}\omega_{2a} \omega_{2 b} +(-\lambda)\omega_3^2$. With the help of this identity, we conclude,  by taking small enough $R>0$, that there exist constants $\gamma_1>1$ and $\gamma_2>2n-10$ such that
	\begin{equation*}
		\frac{1}{\gamma_1} \mathds{1} \leq \bar{\mathbf{A}}^0_1\leq   \mathfrak{\bar{A}}_1 \leq \gamma_2 \mathds{1}. 
	\end{equation*}
	To complete the proof, we note that the other stated inequality follows from similar arguments while the relation $[\mathfrak{\bar{A}}_\ell, \bar{\Pbb}_\ell] =0$, $\ell=1,2$, is a direct consequence of the definitions \eqref{E:AANDP06}--\eqref{E:AANDP06b}.
\end{proof}

\bigskip

\noindent \underline{$(3)$ $n=5$:} For $n=5$, we fix the parameters $\ttk, \, \ttf, \, \ttb$ in the same as \eqref{E:PARASETk6}, and we correspondingly define the operators $\mathfrak{\bar{A}}_2$ and $\bar{\Pbb}_2$ by \eqref{E:AANDP06b}.
In addition, we set
\als{
	\tte=3, \qquad \frac{1}{\ttj}=2 \AND \frac{H}{\tta}=1,
}
and in insert these parameters into $\bar{\mathbf{A}}_1^0$ and $\bar{\mathcal{B}}_1$. We then observe that $\bar{\mathcal{B}}_1=\mathfrak{\bar{A}}_1 \bar{\Pbb}_1$ for
$\bar{\mathfrak{A}}_1$ and $\bar{\Pbb}_1$ defined by
\al{A1}{
	\bar{\mathfrak{A}}_1 =\p{ -\lambda  & 0 & 0 \\
		0 & 2 \ulh_{f d} h^{f \he}  & 0 \\
		-3\lambda & 0 & -3\lambda
	} \AND \bar{\Pbb}_1=\mathds{1},
}
respectively.

\begin{lemma}\label{e:ACOND} 	
	There exist constants $\iota>0$, $R>0$,  $\gamma_2>32$, $\gamma_1>1$, $\tilde{\gamma}_1>1$,  and $\tilde{\gamma}_2>\frac{2(n-3)}{n-1}$ such that
	\begin{gather*}
		\frac{1}{\gamma_1} \mathds{1} \leq \bar{\mathbf{A}}_1^0 \leq 8 \bar{\mathfrak{A}}_1  \leq \gamma_2 \mathds{1}, \quad \frac{1}{\tilde{\gamma}_1}\mathds{1} \leq \bar{\mathbf{A}}^0_2 \leq  \frac{2}{n-1} \mathfrak{\bar{A}}_2 \leq \tilde{\gamma}_2 \mathds{1}, 
		\intertext{and} [\mathfrak{\bar{A}}_\ell, \bar{\Pbb}_\ell] =0
	\end{gather*}
	for  $(t,\widehat{\mathbf{U}}) \in (-\iota,\frac{\pi}{H})\times B_R(0)$ and $\ell=1,2$.
\end{lemma}
\begin{proof}
	From the definition of $\bar{\mathbf{A}}^0_1$, see Theorem \ref{thm-FOSHS-m}, and \eqref{E:A1}, we have
	\als{
		\la \bar{\mathfrak{A}}_1 \omega, \omega \ra_{\ulh} =& -\lambda \omega_1^2+2h^{a b}\omega_{2a} \omega_{2 b}+3(-\lambda)\omega_3^2+3(-\lambda)\omega_1\omega_3
		\intertext{and}
		\la \bar{\mathbf{A}}^0_1 \omega, \omega \ra_{\ulh} =& -\lambda \omega_1^2+ h^{a b}\omega_{2 a} \omega_{2 b}+3(-\lambda)\omega_3^2
	}
	for any $\omega=(\omega_1, \omega_2, \omega_3)^{\tr}$. In order to bound $\la \bar{\mathfrak{A}}_1 \omega, \omega \ra_{\ulh} $, we first bound $3(-\lambda)\omega_1\omega_3$ above and below by using Young's inequality (see Lemma \ref{t:young}) twice with $\epsilon = \frac{\sqrt{3}}{3}$ and $\epsilon=\frac{2+\sqrt{13}}{3}$ to obtain
	\als{
		-\frac{3}{2} \frac{\sqrt{3}}{3} (-\lambda) \omega_1^2-\frac{3}{2}  \sqrt{3} (-\lambda)\omega^2_3 \leq 3(-\lambda) \omega_1 \omega_3 \leq \frac{3}{2}\frac{2+\sqrt{13}}{3} (-\lambda) \omega_1^2+\frac{3}{2} \frac{3}{2+\sqrt{13}}(-\lambda)\omega^2_3.
	}
	Using this inequality, it follows that 
	\als{
		\Bigl(1-\frac{ \sqrt{3} }{2} \Bigr) (-\lambda) \omega_1^2+2h^{a b}\omega_{2a} \omega_{2b}&+3\Bigl(1-\frac{\sqrt{3}}{2 } \Bigr)(-\lambda)\omega_3^2
		\leq \la \bar{\mathfrak{A}}_1 \omega, \omega \ra_{\ulh} \\ \leq & \Bigl(2+\frac{\sqrt{13}}{2 }\Bigr) (-\lambda) \omega_1^2+2h^{a b}\omega_{2a} \omega_{2 b}+\Bigl(2+\frac{\sqrt{13}}{2 } \Bigr)(-\lambda)\omega_3^2,
	}
	from which we deduce that $\frac{1}{8}\bar{\mathbf{A}}_1^0<\bigl(1-\frac{ \sqrt{3} }{2} \bigr) \bar{\mathbf{A}}_1^0 \leq \bar{\mathfrak{A}}_1$. By taking $R>0$ small enough, it is also clear that there exists constants $\gamma_1>1$ and $\gamma_2>32$ such that $8 \bar{\mathfrak{A}}_1 \leq \gamma_2 \mathds{1}$ and $\mathds{1}\leq \gamma_1 \bar{\mathbf{A}}_1^0$.
	
	To complete the proof, we note that the second inequality follows from similar arguments while the relations $[\mathfrak{\bar{A}}_\ell, \bar{\Pbb}_\ell] =0$, $\ell=1,2$, are a direct consequence of the definitions of $\mathfrak{\bar{A}}_\ell$ and $\bar{\Pbb}_\ell$.
\end{proof}

\subsection{Verification of the coefficient assumptions}
With the results of the previous section in hand, we now turn to verifying that the Fuchsian equation \eqref{e:FchEYM} satisfies all of the coefficient assumptions from \S\ref{s:mdlasp1}.
Due to the block structure of the coefficients \eqref{e:coef1}--\eqref{e:coef2}, the verification of the coefficient assumptions for the conformal Einstein (i.e.~the firsts 4 blocks) and the conformal Yang--Mills (i.e.~the last block) components of the equation can be carried out separately.  

\begin{lemma}\label{sym-A-B}
	Suppose
	$\bar{\mathbf{A}}_i^0, \, \bar{\mathbf{A}}_i^b \ts{\ulh}{^{c}_{b}}$, $\bar{\mathcal{B}}_i$, $\bar{G}_i$, $i=1, \cdots, 4$,  and $\acute{\mathbf{A}}^0, \,\acute{\mathbf{A}}^b \ts{\ulh}{^{c}_{b}}$,  $\acute{\mathcal{B}}$, $\acute{G}$ are defined as in Theorems \ref{thm-FOSHS-m} and \ref{thm-FOSHS-Maxwell}, and $\mathfrak{\bar{A}}_i$, $\bar{\Pbb}_i$,  $\acute{\mathfrak{A}}$ and $\acute\Pbb$ are as defined\footnote{Recall that $\mathfrak{\bar{A}}_\mathcal{j}$ and $\bar{\Pbb}_\mathcal{j}$,  for $\mathcal{j}=1, 2$,  vary according to the dimension $n$. } in \S\ref{sec-decom-Einstein}. Then there exist positive constants $R,\kappa,\gamma_1,\gamma_2>0$ and non-negative constants $\lambda_\mathcal{l},\theta,\beta_\mathcal{k}\geq 0$, $\mathcal{l}=1,2,3$ and $\mathcal{k}=0,\ldots,7$, such that coefficients of the Fuchsian system \eqref{e:FchEYM} satisfy the assumptions $(1)$--$(5)$ from \S\ref{s:mdlasp1} and the inequality 
	\begin{equation}\label{e:kineq}
		\kappa>\frac{1}{2}\gamma_1 \biggl(\sum_{k=0}^3\beta_{2k+1}+2\lambda_3\biggr)
	\end{equation}	  
	holds. 
\end{lemma}
\begin{proof}
	$\;$
	
	\noindent \underline{$1.$ The conformal Einstein component.}  Noting that the gravitational variables used in the Fuchsian formulation of the reduced conformal Einstein equation in the articles
	\cite{Oliynyk2016a,Liu2018,Liu2018b,Liu2018a} are essentially the same ones used in this article, the same arguments from \cite{Oliynyk2016a,Liu2018,Liu2018b,Liu2018a} that are employed to verify the  coefficient assumptions can be adapted, with the help of Lemmas \ref{e:ACOND1}--\ref{e:ACOND}, in a straightforward manner to show that the Fuchsian formulation of the reduced Einstein equations in this article satisfies the coefficient assumptions. We omit the details.
	
	\bigskip
	
	\noindent \underline{$2.$ The conformal Yang--Mills component.}
	
	\bigskip
	
	\noindent \underline{Assumptions $(2)$ \& $(4)$:} We begin the verification of assumption $(2)$ and $(4)$ by establishing the symmetry of the operators  $\acute{\mathbf{A}}^0$ and $\acute{\mathbf{A}}^b \ts{\ulh}{^{c}_{b}} \eta_c$. To this end, we observe that the blocks of $\acute{\mathbf{A}}^0$ satisfy
	\begin{align*}
		\tensor{\bigl( \acute{\mathbf{A}}^0_{(11)} \bigr) }{^{\he}_{e}} \ulh^{e f}  \ulh_{\he d} & = -\lambda \hat{h}_{\hb e} \ulh^{\hb \he} \ulh^{e f}  \ulh_{\he d}  = -\lambda \ulh^{f \he} \hat{h}_{\he d} = \tensor{\bigl( \acute{\mathbf{A}}^0_{(1 1)}\bigr) }{^{f}_d } \quad (\hat{h}_{\hb e} = \tensor{\ulh}{^a_{\hb}} g_{a b}\tensor{\ulh}{^b_e}), \end{align*}		
	\begin{align*}
		\tensor{\bigl( \acute{\mathbf{A}}^0_{(22)}\bigr) }{^{\hd}_f }  \ulh_{\hd \hc} \ulh^{f \he} & = \left(  -\lambda  h^{\hd d} - \lambda \nu^r \nu^s g_{r s}  \xi^{d} \xi^{\hd} + 2 \xi^{d} \xi^{\hd} \right)  \ulh_{d f} \ulh_{\hd \hc} \ulh^{f \he} \\
		&= \left(  -\lambda  h^{\hd \he} - \lambda \nu^r \nu^s g_{r s} \xi^{\he} \xi^{\hd} + 2  \xi^{\he} \xi^{\hd}  \right) \ulh_{\hd \hc}  = \tensor{\bigl( \acute{\mathbf{A}}^0_{(2 2)}\bigr) }{^{ \he}_{\hc} }, \end{align*}		
	\begin{align*}\tensor{\bigl( \acute{\mathbf{A}}^0_{(33)}\bigr) }{^{\hc \hd}_{c d} }  \ulh_{\bar a \hc} \ulh_{\bar b \hd} \ulh^{c \ha} \ulh^{d \hb}  =& \ulh_{c a^\prime} \ulh_{d b^\prime} h^{a^\prime \hc} h^{b^\prime \hd} \ulh_{\bar a \hc} \ulh_{\bar b \hd} \ulh^{c \ha} \ulh^{d \hb}  \notag \\
		=& h^{\ha \hc} h^{\hb \hd} \ulh_{\hc \bar a} \ulh_{\hd \bar b} = \tensor{\bigl( \acute{\mathbf{A}}^0_{(33)}\bigr) }{^{\ha \hb}_{\bar a \bar b} },  
	\end{align*}		
	\begin{align*}		\tensor{\bigl( \acute{\mathbf{A}}^0_{(21)}\bigr) }{_{f e} }  \ulh^{f \hd} \ulh^{e \he} & = - \lambda \nu^r g_{r s} \tensor{\ulh}{^{s}_{e}} \ulh_{d f} \xi^d  \ulh^{f \hd} \ulh^{e \he} = - \lambda \nu^r g_{r s} \tensor{\ulh}{^{s \he}} \xi^{\hd} = \tensor{\bigl( \acute{\mathbf{A}}^0_{(1 2)}\bigr) }{^{\he \hd} }.
	\end{align*}
	We further observe that the blocks 
	$\acute{\mathbf{A}}^0_{(1 3)}$, $\acute{\mathbf{A}}^0_{(31)}$, $\acute{\mathbf{A}}^0_{(23)}$, and $\acute{\mathbf{A}}^0_{(32)}$ all vanish and that the block $\acute{\mathbf{A}}^0_{(44)}$ is symmetric. We therefore conclude from \eqref{e:Asym} that  $\acute{\mathbf{A}}^0$ is symmetric.
	The symmetry of the operator $\acute{\mathbf{A}}^b \ts{\ulh}{^{c}_{b}} \eta_c$ can be established in a similar fashion from the following identities 
	\begin{align*}
		\tensor{\bigl( (\acute{\mathbf{A}}^b \ulh^c{}_{ b} )_{(21)} \bigr) }{_{f e} }  \ulh^{f \hd} \ulh^{e \he} & = -\bigl[ 2\xi^c \nu^r g_{r s} \tensor{\ulh}{^{s}_e}+  \tensor{\ulh}{^c_e}\bigr] \ulh_{d f} \xi^{d}  \ulh^{f \hd} \ulh^{e \he} \\
		&=  -\bigl[ 2\xi^c \nu^r g_{r s}\tensor{\ulh}{^{s \he}} +\tensor{\ulh}{^{c \he}} \bigr]\xi^{\hd}  = \tensor{\bigl( (\acute{\mathbf{A}}^b \ulh^c{}_{ b})_{(1 2)}\bigr) }{^{\he \hd} },  \\
		\tensor{\bigl( (\acute{\mathbf{A}}^b \ulh^c{}_{ b} )_{(31)}\bigr) }{_{\bar a \bar b e} }  \ulh^{\bar a \ha} \ulh^{\bar b \hb} \ulh^{\he e} & =  \ulh_{e \bar b} \ulh_{a \bar a} h^{a c} \ulh^{\bar a \ha} \ulh^{\bar b \hb} \ulh^{\he e}  = h^{\ha c} \ulh^{\hb \he} = \tensor{\bigl( (\acute{\mathbf{A}}^b \ulh^c{}_{b} )_{(1 3)}\bigr) }{^{\he \ha \hb} }, \\
		\tensor{\bigl( (\acute{\mathbf{A}}^b \ulh^c{}_{ b} )_{(32)}\bigr) }{^{\hd}_{\bar a \bar b}}  \ulh^{\bar a \ha} \ulh^{\bar b \hb} \ulh_{\hd f} & = \ulh_{b \bar b} \ulh_{a \bar a} h^{\hd a} h^{b c} \ulh^{\bar a \ha} \ulh^{\bar b \hb} \ulh_{\hd f} \\
		&= \ulh^{\hd \ha} \ulh^{\hb c} \ulh_{\hd f} =  \ulh^{\ha d} \ulh^{\hb c} \ulh_{d f}  = \tensor{\bigl( (\acute{\mathbf{A}}^b \ulh^c{}_{ b})_{(2 3)}\bigr) }{^{\ha \hb}_f }, \\
		\tensor{\bigl( (\acute{\mathbf{A}}^b \ulh^c{}_{ b})_{(11)}\bigr) }{^{\he}_{e}}  \ulh^{e f}  \ulh_{\he d} & = -2\xi^c \hat{h}_{e s} \ulh^{s \he}  \ulh^{e f}  \ulh_{\he d}   =  -2\xi^c \hat{h}_{\he d}  \ulh^{f \he} = \tensor{\bigl( (\acute{\mathbf{A}}^b \ulh^c{}_{ b})_{(1 1)}\bigr) }{^{f}_d }, \\
		\tensor{\bigl( (\acute{\mathbf{A}}^b \ulh^c{}_{ b})_{(22)}\bigr) }{^{\hd}_f }  \ulh_{\hd \hc} \ulh^{f \he} & = -\bigl( 2\nu^r \nu^s g_{r s} \xi^c  \xi^{\hd} \xi^{d} - 2 \xi^{(d} h^{\hd) c}  + 2\xi^c h^{\hd d} \bigr)  \ulh_{d f} \ulh_{\hd \hc} \ulh^{f \he} \\
		&= -\bigl( 2\nu^r \nu^s g_{r s} \xi^c \xi^{d} \xi^{\he} - 2 \xi^{(\he} h^{d) c} + 2\xi^c h^{\he d} \bigr)  \ulh_{d \hc} \notag\\
		&= \tensor{\bigl( ( \acute{\mathbf{A}}^b \ulh^c{}_{ b} )_{(2 2)} \bigr) }{^{ \he}_{\hc} },
	\end{align*}
	and the vanishing of the remaining blocks of $\acute{\mathbf{A}}^b \ts{\ulh}{^{c}_{b}} \eta_c$.		
	
	Next, we claim that $\acute{\mathbf{A}}^0$ and $\acute{\mathfrak{A}}$ satisfy an inequality of the form \eqref{e:coefcp}. To see why this is the case, we consider an element
	\begin{equation*}
		\vartheta = (Y^e,Z_{\hd},X_{\ha\hb},W_s)\in  T \Sigma \oplus T^\ast \Sigma \oplus T^{0}_2 \Sigma \oplus T^\ast \Sigma.
	\end{equation*}
	Then by \eqref{e:inprod} and the definition of   $\acute{\mathbf{A}}^0$ from Theorem \ref{thm-FOSHS-Maxwell}, we have
	\begin{align*}
		\la \vartheta,\acute{\mathbf{A}}^0 \vartheta \ra_{\ulh} 
		= & -\lambda   \tensor{\ulh}{^a_{d}} g_{a b}\tensor{\ulh}{^b_e}Y^eY^d -   \lambda \nu^r g_{r s} \tensor{\ulh}{^{s}_{e}}  \xi^{\hd} Z_{\hd} Y^e- \lambda   \nu^r g_{r s} \tensor{\ulh}{^{s}_{e}}  \xi^d Y^e Z_{d} \notag  \\
		& + \bigl[ -\lambda  h^{\hd d}  - \lambda \nu^r \nu^s g_{r s}  \xi^{d} \xi^{\hd} + 2 \xi^{d} \xi^{\hd} \bigr] Z_{\hd} Z_d +    h^{\ha a} h^{\hb b} X_{\ha\hb} X_{a b} +  h^{b s} W_s W_b
	\end{align*}
	from which we find, with the help of Lemma \ref{t:young2}, that
	\begin{align}\label{e:tat1}
		\la \vartheta,\acute{\mathbf{A}}^0 \vartheta \ra_{\ulh} 
		\leq & -\lambda   \tensor{\ulh}{^a_{d}} g_{a b}\tensor{\ulh}{^b_e}Y^eY^d +  \epsilon  \ulh^{\hd d}Z_{\hd}Z_{d} +\frac{1}{ \epsilon} |\xi|^2_{\ulh}  \lambda^2 \nu^r g_{r s} \tensor{\ulh}{^{s}_{e}}    \nu^c g_{c \hb} \tensor{\ulh}{^{\hb}_{b}}     Y^eY^b   \notag  \\
		& + \bigl[ -\lambda  h^{\hd d}  - \lambda \nu^r \nu^s g_{r s}  \xi^{d} \xi^{\hd} + 2 \xi^{d} \xi^{\hd} \bigr] Z_{\hd} Z_d +    h^{\ha a} h^{\hb b} X_{\ha\hb} X_{a b} +  h^{b s} W_s W_b
	\end{align}
	and
	\begin{align}\label{e:tat2}
		\la \vartheta,\acute{\mathbf{A}}^0 \vartheta \ra_{\ulh} 
		\geq & -\lambda   \tensor{\ulh}{^a_{d}} g_{a b}\tensor{\ulh}{^b_e}Y^eY^d - \epsilon  \ulh^{\hd d}Z_{\hd}Z_{d} -\frac{1}{ \epsilon} |\xi|^2_{\ulh} \lambda^2 \nu^r g_{r s} \tensor{\ulh}{^{s}_{e}}   \nu^c g_{c \hb} \tensor{\ulh}{^{\hb}_{b}}    Y^eY^b   \notag  \\
		& + \bigl[ -\lambda  h^{\hd d}  - \lambda \nu^r \nu^s g_{r s}  \xi^{d} \xi^{\hd} + 2 \xi^{d} \xi^{\hd} \bigr] Z_{\hd} Z_d +    h^{\ha a} h^{\hb b} X_{\ha\hb} X_{a b} +  h^{b s} W_s W_b
	\end{align}	
	for any $\epsilon > 0$ where we have set $|\xi|^2_{\ulh}=\ulh_{ad}\xi^a\xi^d$.
	
	On the other hand,
	\begin{align}\label{e:tAt0}
		& \la \vartheta,\acute{\mathfrak{A}} \vartheta \ra_{\ulh} \notag\\
		=& - (n-3) \lambda  \ts{\ulh}{^{\ha}_{a}} g_{s \ha} \tensor{\ulh}{^{s}_e} Y^e Y^a  -(n-3)\lambda  h^{\hd d}  Z_{\hd} Z_{d}+  h^{\ha a}  h^{\hb b} X_{\ha\hb}X_{a b} +\frac{1}{2} h^{b s}W_s W_b. 
	\end{align}	 
	By setting $\epsilon=|\xi|_{\ulh}$ and
	taking $R>0$ small enough, it then follows from \eqref{e:tat1}--\eqref{e:tAt0} that there exist constants $\kappa, \gamma_1, \gamma_2>0$ such that the inequality \eqref{e:coefcp} holds.

	Finally, we note that it is
	straightforward to check that the operators $\acute{\mathbf{A}}^0$,  $\acute{\mathbf{A}}^b \ts{\ulh}{^{c}_{b}} \eta_c$ and $\acute{\mathfrak{A}}$ verify the remaining stated properties from the coefficient assumptions $(2)$ and $(4)$. This completes the verification of the coefficient assumptions $(2)$ and $(4)$.
	
	\bigskip
	
	\noindent \underline{Assumptions $(1)$ \& $(5)$:} 	
	With our choice of $\acute\Pbb=\mathds{1}$ and hence $\acute{\Pbb}^\perp =0$, we immediately have that $[\acute{\Pbb}, \, \acute{\mathfrak{A}}]=0$, and the only non-vanishing relation involving $\Theta^\prime(0)$ is the first one, which reads
	\begin{equation*}
		\la v, \Theta^\prime(0)  v\ra_{\ulh}=\mathcal{O}(\theta v\otimes v+|t|^{-\frac{1}{2}}\beta_0v\otimes  v+|t|^{-1}\beta_1 v\otimes  v).
	\end{equation*}
	It is straightforward to verify that this condition on $\Theta^\prime(0)$ is satisfied for an appropriate choice of constants $\theta$, $\beta_0$ and $\beta_1$. Since it is also straightforward verify the other stated properties from the coefficient assumptions $(1)$ and $(5)$, we will omit the details of the verification of these properties. This completes the verification of the coefficient assumptions $(1)$ and $(5)$.  
	
	\bigskip
	
	\noindent \underline{Assumption $(3)$:}  Noticing that the maps $\acute{G}_0$ and $ \acute{G}_1$ from the expansion of 
	$\acute{G}$ given in Theorem \ref{thm-FOSHS-Maxwell} are both analytic for $(t,\mathbf{U})\in \bigl(-\iota,\frac{\pi}{H}\bigr)$ for $\iota,R>0$ small enough and vanish for $\mathbf{U}=0$, it follows immediately that there exists constants $\lambda_\ell\geq 0$, $\ell=1,2$, such the all the conditions on the map $\acute{G}$ from the coefficient assumption $(3)$ are satisfied. This completes the verification of the coefficient assumption $(3)$.
	
	\bigskip
	
	\noindent \underline{$3.$ The combined system.} The above arguments yield positive constants  $R, \kappa,\gamma_1,\gamma_2>0$, and non-negative  $\lambda_\mathcal{l},\theta,\beta_\mathcal{k}\geq 0$, $\mathcal{l}=1,2$ and $\mathcal{k}=0,\cdots,7$, such that the coefficients of the Fuchsian system \eqref{e:FchEYM} verify the assumptions $(1)$--$(5)$ from \S\ref{s:mdlasp1}. 
	To complete the poof, we make a number of observations with regard to the selection of some of the parameters beginning with the observation that the term from the Fuchsian system \eqref{e:FchEYM} that corresponds to $G_2$ vanishes, which allows us to take $\lambda_3=0$. Next, we make the following observations with regard to the parameters $\beta_{2k+1}$, $k=0,1,2,3$:
	\begin{enumerate}
		\item The $1/t$ singular term from $ \la v, \Pbb(\pi(v))\Theta^\prime(0)\Pbb(\pi(v)) v\ra_{\ulh}$ is of the form $\mathcal{O}(C_1 |t|^{-1} \Pbb(\pi(v))v\otimes \Pbb(\pi(v))v\otimes \Pbb(\pi(v))v)$ for some constant $C_1>0$. Thus for any $\beta_1>0$, we can, by choosing $R>0$ suffuciently small, ensure that $C_1R\leq \beta_1 $. 
		\item The $1/t$ singular term of from $ \la v, \Pbb(\pi(v))\Theta^\prime(0)\Pbb^\perp (\pi(v)) v\ra_{\ulh}$ is of the form $\mathcal{O}(C_3 |t|^{-1} \Pbb(\pi(v))v\otimes \Pbb(\pi(v))v\otimes \Pbb^\perp (\pi(v))v)$ for some constant $C_3>0$. As a consequence, for any $\beta_3>0$, we can arrange that $C_3R^2\leq \beta_3 $ by choosing $R>0$ small. By similar considerations, for any choose of $\beta_5>0$, we can guarantee that $C_5R^2\leq \beta_5 $ holds by choosing $R>0$ sufficiently small. 
		\item Noting that the term $ \la v, \Pbb^\perp(\pi(v))\Theta^\prime(0)\Pbb^\perp(\pi(v)) v\ra_{\ulh}$ only involves the Einstein components, similar arguments as employed in \cite[\S7.1]{Liu2018} (see also in  \cite{Liu2018b}) imply that this term contains no $1/t$ singular term, and as a consequence, we can take $\beta_7=0$. 
	\end{enumerate}
	From these considerations and the fact that the parameters $\beta_k$ can be chosen independently of $\kappa$ and $\gamma_1$, which are both positive, it is then clear that we can ensure that the inequality \eqref{e:kineq} holds by choosing $R>0$ sufficiently small. This complete the proof.
\end{proof}

\begin{remark}\label{constants-rem}
	Although we have established that Fuchsian system \eqref{e:FchEYM} satisfies all of the coefficient assumptions from \S\ref{s:mdlasp1} for some choice of constants, we have not calculated explicit values for these constants. This is because explicit values are only required for determining the decay rates, which we will not make use of in any subsequent arguments.
\end{remark}

Now that we have verified that the Fuchsian system \eqref{e:FchEYM} satisfies all of the coefficient assumptions from \S\ref{s:mdlasp1}, we can apply Theorem \ref{t:glex}, after making the time transformation $t\longmapsto -t$,  to it to obtain the global existence result contained in the following theorem.  

\begin{theorem}\label{t:Uglobal}
	Suppose $s\in\Zbb_{> \frac{n+1}{2}}$, then there exist constants $\delta_0>0$ and $C>0$ such that if the initial data $\widehat{\mathbf{U}}_0$
	satisfies $\lVert \widehat{\mathbf{U}}_0\rVert_{H^s}\leq \delta$ for any $\delta \in (0,\delta_0]$, then there exists a unique solution\footnote{Recall that $\Sigma=\mathbb{S}^{n-1}$ and the initial conformal time $t=\frac{\pi}{2H}$ corresponds to the physical time $\tau=0$.}
	\begin{align*} 
		\widehat{\mathbf{U}} \in C^0\Bigl(\Bigl(0, \frac{\pi}{2H}\Bigr],H^s(\Sigma;V)\Bigr) \cap C^1\Bigl(\Bigl(0, \frac{\pi}{2H}\Bigr],H^{s-1}(\Sigma;V)\Bigr)\cap \Li \Bigl(\Bigl(0, \frac{\pi}{2H}\Bigr],H^s(\Sigma;V)\Bigr)
	\end{align*}
	to the Fuchsian equation \eqref{e:FchEYM} on $\bigl(0, \frac{\pi}{2H}\bigr]\times \Sigma$ that satisfies $\widehat{\mathbf{U}}|_{t=\frac{\pi}{2H}}=\widehat{\mathbf{U}}_0$. 
	Moreover, the solution $\widehat{\mathbf{U}}$ satisfies the energy estimate 
	\begin{equation}\label{e:ineq1}
		\lVert \widehat{\mathbf{U}}(t)\rVert_{H^s}^2 + \int_t^{\frac{\pi}{2H}}\frac{1}{\bar t }\lVert \widehat{\Pbb} \widehat{\mathbf{U}}(\bar t )\rVert^2_{H^s}d \bar t  \leq C  \delta^2
	\end{equation}
	for all $t\in \bigl(0, \frac{\pi}{2H}\bigr]$. 
\end{theorem}

\subsection{Improved estimates for the conformal Yang--Mills component of $\widehat{\mathbf{U}}$}\label{S:YMipv}	
The conformal Yang--Mills fields correspond to the last four entries $(\tE^e,E_{d},H_{ab}, \bar A_s)$ of the vector $\widehat{\mathbf{U}}$ defined by \eqref{e:hatu}. Given a solution $\widehat{\mathbf{U}}$ from Theorem \ref{t:Uglobal} of the Fuchsian system \eqref{e:FchEYM}, we establish improved bounds on the fields $(\tE^e,E_{d},H_{ab}, \bar A_s)$ by showing that the renormalized conformal Yang--Mills fields
\begin{align}\label{e:newYM}
	\p{\mathring{\tE}^e, \mathring{E}_{d},\mathring{H}_{a b},\mathring{A}_s}^{\tr}=\diag\{t^{-1},t^{-1},t^{-1},t^{-\frac{1}{2}}\}	\p{\tE^e,E_{d},H_{ab}, \bar A_s}^{\tr},
\end{align}
remain bounded. The precise statement of this result is given in the following corollary. 

\begin{corollary}\label{t:impv} Suppose $\widehat{\mathbf{U}}$ is a solution of the Fuchsian system \eqref{e:FchEYM} from Theorem \ref{t:Uglobal}, the initial data and the constants $s$, $\delta$ and $\delta_0$ are as given in that theorem.
	Then there exists a constant $C>0$, independent of $\delta \in (0,\delta_0]$, such that the renormalized conformal Yang--Mills fields \eqref{e:newYM} are uniformly bounded by \[\lVert \mathring A_a (t)\rVert_{H^s} + \lVert \mathring E_a (t)\rVert_{H^s} + \lVert \mathring H_{a b} (t)\lVert _{H^s}  \leq C\delta \]  
	for all $t\in \bigl(0, \frac{\pi}{2H}\bigr]$. 
\end{corollary}

\begin{remark}
	It is not difficult to verify that the renormalized conformal Yang--Mills fields $\mathring{E}_{d}$, $\mathring{H}_{ab}$, and $\mathring{A}_s$ are quantitatively equivalent to the physical Yang--Mills variables $\tilde{F}_{ab}$ and $\tilde{A}_s$. Because of this, Corollary \ref{t:impv} yields uniform bounds for the physical Yang--Mills variables.
\end{remark}
\begin{proof}[Proof of Corollary \ref{t:impv}] 
	We begin by noting from \eqref{def-U} and \eqref{e:hatu} that $\mathbf{U}$ and $\widehat{\mathbf{U}}$ are uniquely determine each other and have comparable norms. As a consequence, in the following arguments, we do not distinguish between the two. Next, by substituting \eqref{e:newYM} into \eqref{Maxwell-FOSHS}, we get 
	\begin{align}\label{e:newymeq0}
		&- \acute{\mathbf{A}}^0\diag\{t,t,t,\sqrt{t}\} \nu^c  \nb_{c} \mathring{\mathbf{V}} + \acute{\mathbf{A}}^b\ts{\ulh}{^{c}_{b}} \diag\{t,t,t,\sqrt{t}\}   \nb_{c} \mathring{\mathbf{V}} \notag   \\
		=&\frac{1}{Ht}\acute{\mathcal{B}}  \diag\{t,t,t,\sqrt{t}\}\mathring{\mathbf{V}}+\acute{G}(t,\mathbf{U})-\frac{1}{H} \acute{\mathbf{A}}^0 \diag\Bigl\{1,1,1,\frac{1}{2\sqrt{t}}\Bigr\}\mathring{\mathbf{V}}
	\end{align}
	where we have set
	\begin{equation*}
		\mathring{\mathbf{V}}:=\p{\mathring{\tE}^e, \mathring{E}_{\hd},\mathring{H}_{\ha\hb},\mathring{A}_s}^{\tr}.
	\end{equation*}  
	Then multiplying \eqref{e:newymeq0} on the left by $(\diag\{t,t,t,\sqrt{t}\})^{-1}$ leads to
	\begin{align}\label{e:newymeq1}
		&- \acute{\mathbf{A}}^0 \nu^c  \nb_{c}   \mathring{\mathbf{V}}+  \acute{\mathbf{A}}^b  \ts{\ulh}{^{c}_{b}}  \nb_{c}   \mathring{\mathbf{V}} 
		= \frac{1}{Ht} \mathring{\mathbf{B}} \mathring{\mathbf{V}} + \mathfrak{S}(t,\mathring{\mathbf{U}}) 
	\end{align}
	where 
	\begin{align}
		\mathbf{\mathring{U}}&=(m,\, p^a,\, m_d, \, \tp{a}{d}, \, s^{ab}, \, \tss{ab}{d}, \, s, \, s_d, \,\mathring{\mathbf{V}})^{\tr}, \label{e:rgu}\\
		\mathring{\mathbf{B}}&=\p{- (n-4) \lambda \ulh^{\hb \he} \tensor{\ulh}{^a_{\hb}} g_{a b}\tensor{\ulh}{^b_e}   & 0  & 0 & 0\\
			0 &  -(n-4)\lambda h^{\hd d} \ulh_{d f} & 0 & 0\\
			0 & 0 & 0 & 0 \\
			0 & 0 & 0 & 0}, \\
		\mathfrak{S} (t,\mathbf{\mathring{U}})&=\p{\mathfrak{S}^{\he}_1(t,\mathbf{\mathring{U}}),
			\mathfrak{S}_{2 f}(t,\mathbf{\mathring{U}}),
			\mathfrak{S}_{3 \bar a\bar b}(t,\mathbf{\mathring{U}}),
			\mathfrak{S}_{4 o}(t,\mathbf{\mathring{U}})  }^{\tr}
	\end{align}
	and the components of $\mathfrak{S}$ are given by 
	\begin{align*}
		\mathfrak{S}^{\he}_1(t,\mathbf{\mathring{U}})= &  H^{-1} \ttb \lambda \nu^r g_{r s} \tensor{\ulh}{^{s \he}} p^{\hd} \mathring{E}_{\hd} + t^{-1} \mathfrak{D}^{  \he}_1(t,\mathbf{U}) -t^{-\frac{3}{2}} \ulh^{\he e} \Xi_{1 e}(E_d,H_{ab},\bar{A}_s),   \\
		\mathfrak{S}_{2f}(t,\mathbf{\mathring{U}})=&	H^{-1} \lambda \nu^r g_{r s} \tensor{\ulh}{^{s}_{e}} \ulh_{df} \ttb p^d \mathring{\tE}^e - H^{-1} \ttb^2 t \bigl[ -   \lambda \nu^r \nu^s g_{r s}  p^{d} p^{\hd} + 2   p^{d} p^{\hd} \bigr] \ulh_{d f} \hat{E}_{\hd} \notag \\
		& +t^{-1}\mathfrak{D}_{2 f}(t,\mathbf{U}) +t^{-\frac{3}{2}} \ulh_{df} h^{d \ha} \Xi_{1\ha}(E_d,H_{ab},\bar{A}_s),   \\
		\mathfrak{S}_{3 \bar a\bar b}(t,\mathbf{\mathring{U}})=& t^{-1} \mathfrak{D}_{3 \bar a\bar b}(t,\mathbf{U})- t^{-\frac{3}{2}} \ulh_{a \bar a} \ulh_{b \bar a} h^{a \ha} \Xi_{2\ha}^b(E_d,H_{ab},\bar{A}_s),   \\
		\mathfrak{S}_{4 o}(t,\mathbf{\mathring{U}})=& t^{-\frac{1}{2}} \mathfrak{D}_{4 o} (t,\mathbf{U})+t^{-1} \ulh_{o r} h^{r a} \Xi_{3a}(E_d,H_{ab},\bar{A}_s).   
	\end{align*}
	In the above definition, we note that the maps $\mathfrak{S}^{\he}_1(t,\mathbf{\mathring{U}})$, $\mathfrak{S}_{2f}(t,\mathbf{\mathring{U}})$, $\mathfrak{S}_{3 \bar a\bar b}(t,\mathbf{\mathring{U}})$ and $\mathfrak{S}_{4 o}(t,\mathbf{\mathring{U}})$, the maps $\mathfrak{D}_i$, $i=1,\cdots, 4$, and $\Xi_j$, $j=1,2,3$, are expressed in terms of the variables $(t,\mathbf{U})$ and $(E_d,H_{ab},\bar{A}_s)$. This makes sense since $\mathring{\mathbf{U}}$ is determined by $(t,\mathbf{U})$ as can be seen from \eqref{def-U}, \eqref{e:newYM} and the definition of $\mathbf{\mathring{U}}$ given above. 
	
	We now claim that the maps $\mathfrak{S}^{\he}_1(t,\mathbf{\mathring{U}})$, $\mathfrak{S}_{2f}(t,\mathbf{\mathring{U}})$, $\mathfrak{S}_{3 \bar a\bar b}(t,\mathbf{\mathring{U}})$ and $\mathfrak{S}_{4 o}(t,\mathbf{\mathring{U}})$ are analytic for $(t,\mathbf{\mathring{U}}) \in (-\iota,\frac{\pi}{H})\times B_R(0)$ for $\iota,R>0$  sufficiently small and vanish for $\mathring{\mathbf{U}}=0$. 
	To see why this is the case, 
	we observe from Lemmas \ref{lem-maxwell-hyperbolic-0}--\ref{lem-FOSHS-Maxwell} and Theorem \ref{thm-FOSHS-Maxwell} that each term in $\mathfrak{D}_{i} (t, \mathbf{U})$, $i=1,2,3$, involves one of the factors $E_a, \, H_{bc}$, $[\bar A_a, E_b]$,  and $[\bar A_a, H_{bc}]$, while each term in $\mathfrak{D}_{4} (t, \mathbf{U})$ involves either $\bar A_a$ or $E_b$ as a factor. From this observation, it can then be verified via a direct calculation that $t^{-1} \mathfrak{D}_{i} (t, \mathbf{U})$, $i=1,2,3$, and  $t^{-\frac{1}{2}} \mathfrak{D}_{4} (t, \mathbf{U})$ are analytic in $(t,\, \mathbf{\mathring{U}})$. Analogously, each term in $\Xi_i$, $i =1,2$, contains one of the factors $[\bar A_a, E_b]$, $[\bar A_a, H_{b c}]$ and $\Xi_{3a} = E_a$, and consequently,  $t^{-\frac{3}{2}}\Xi_i$, $i =1,2$ and $t^{-1} \Xi_3$ are  analytic in $(t,\, \mathbf{\mathring{U}})$. For example, a calculation shows that
	\als{
		t^{-\frac{3}{2}} \ulh_{df} h^{d \ha} \Xi_{1\ha} = t^{-\frac{3}{2}} \ulh_{df} h^{d \ha} h^{\hd c} [\bar A_c, H_{\hd \ha}]  
		= & \ulh_{df} h^{d \ha} h^{\hd c} [\mathring{A}_c, \mathring{H}_{\hd \ha}].
	}
	Using these types of arguments, the analytic dependence of the maps $\mathfrak{S}^{\he}_1(t,\mathbf{\mathring{U}})$, $\mathfrak{S}_{2f}(t,\mathbf{\mathring{U}})$, $\mathfrak{S}_{3 \bar a\bar b}(t,\mathbf{\mathring{U}})$ and $\mathfrak{S}_{4 o}(t,\mathbf{\mathring{U}})$ of the variables $(t,\mathbf{\mathring{U}})$ can be verified by direct calculations. 
	
	It then follows from the above observations, the arguments from the proof of Lemma \ref{sym-A-B}, and the fact that $\widehat{\mathbf{U}}$ is a solution of the Fuchsian system \eqref{e:FchEYM} from Theorem \ref{t:Uglobal}  that, for $R>0$ chosen small enough, there exist positive constants $\kappa,\gamma_1,\gamma_2>0$ and non-negative constants $\lambda_\mathcal{l},\theta,\beta_\mathcal{k}\geq 0$, $\mathcal{l}=1,2,3$ and $\mathcal{k}=0,\ldots,7$, such that the Fuchsian system \eqref{e:newymeq1} satisfies the assumptions $(1)$--$(5)$ from \S\ref{s:mdlasp1} and the inequality \eqref{e:kineq}. Moreover, we observe from \eqref{e:newYM} that
	\begin{align*} 
		\mathring{\mathbf{V}}|_{t=\frac{\pi}{2H}}=\diag\Bigl\{\frac{2H}{\pi},\frac{2H}{\pi},\frac{2H}{\pi},\sqrt{\frac{2H}{\pi}}\Bigr\}	\p{\tE^e,E_{d},H_{ab}, \bar A_s}^{\tr}|_{t=\frac{\pi}{2H}}. 
	\end{align*}
	Since $H^s$ norm of the right hand side of the above expression is bounded by $\norm{\widehat{\mathbf U}_0}$, which by assumption satisfies
	$\norm{\widehat{\mathbf U}_0}\leq \delta$, we have that $\lVert \mathring{\mathbf{V}}(\frac{\pi}{2H})\rVert_{H^s}\leq C\delta$. We therefore conclude via an application of Theorem \ref{t:glex} that
	$\mathring{\mathbf{V}}$ is bounded by 
	$\lVert \mathring{\mathbf{V}}(t)\rVert_{H^s} \leq C  \delta$ for all $t\in \bigl(0,\frac{\pi}{2H}\bigr]$,
	which completes the proof.
\end{proof}

\section{Gauge transformations\label{s:maprf}}
In the proof of the main result, Theorem \ref{t:mainthm}, we need to perform a gauge transformation in order to change from a formulation of the Einstein--Yang--Mills equations that is useful for establishing the local-in-time existence of solutions to a formulation that is suited to establishing global bounds. In particular, the local-in-time existence
of solutions to the Einstein--Yang--Mills equations is obtained from applying the local existence theory developed in \cite{LW2021b}, which requires the use of the adapted temporal gauge \eqref{phys-temp-gauge} and the wave gauge \eqref{e:prewg}.  On the other hand, we obtain global estimates from our Fuchsian formulation of the conformal Einstein--Yang--Mills equations, which was derived using the temporal and wave  gauges \eqref{e:temgg} and \eqref{E:CONSTR1}, respectively. These two gauge choices are not equivalent and the transformation between the two has to be taken into account.  In this section, we collect together the technical results that are needed to establish that this gauge transformation is well defined.

\begin{lemma}\label{t:slrl1}
	Suppose $(\tilde{g}^{ab}, \, \tilde{A}^\star_a)$ is a solution of the Einstein--Yang--Mills equations \eqref{eq-einstein}--\eqref{eq-maxwell-bianchi-o} in the temporal gauge \eqref{phys-temp-gauge}, 
	$\tilde{\mfu} : \widetilde{\mathcal{M}} \rightarrow G$ solves the differential equation 
	\begin{align} \label{gauge-ode}
		\del{\tau}\tilde{\mathfrak{u}} = -\tilde{A}^\star_0\tilde{\mathfrak{u}},
	\end{align}
	and let
	\begin{align}\label{e:ggtsf1}
		\tilde{A}_a=\tilde{\mathfrak{u}}^{-1}\tilde{A}^\star_a\tilde{\mathfrak{u}}+\tilde{\mathfrak{u}}^{-1}(d\tilde{\mathfrak{u}})_a \quad (i.e., \tilde{F}_{ab}=\tilde{\mathfrak{u}}^{-1} \tilde{F}^\star_{ab} \tilde{\mathfrak{u}} ).
	\end{align}
	Then $(\tilde{g}^{ab},  \, \tilde{A}_a)$
	determines a solution of Einstein--Yang--Mills equations \eqref{eq-einstein}--\eqref{eq-maxwell-bianchi-o} in the temporal gauge $\tilde{A}_a\tilde{\nu}^a=0$ where $\tilde{\nu}^a = \tilde{\ulg}^{ab}(d\tau)_b$.
\end{lemma}
\begin{proof}
	The proof is similar to \cite[Theorem $2$]{Segal1979} (see also \cite[\S $6$]{ChoquetBruhat1991}), and follows from the observation that the condition
	\begin{align}\label{e:ggtr1}
		\tilde{\nu}^a(d\tilde{\mathfrak{u}})_a=-\tilde{\nu}^a\tilde{A}^\star_a\tilde{\mathfrak{u}}
	\end{align}
	is equivalent to the differential equation \eqref{gauge-ode}.
	Since, by assumption, we have a solution $\tilde{\mathfrak{u}}$ to the above differential equations, we can use it to define a gauge transformation under which the gauge potential transforms according to
	\begin{equation*}
		\tilde{A}_a=\tilde{\mathfrak{u}}^{-1}\tilde{A}^\star_a\tilde{\mathfrak{u}}+\tilde{\mathfrak{u}}^{-1}(d\tilde{\mathfrak{u}})_a
		.  
	\end{equation*}
	Contracting this with $\tilde{\nu}^a$, we find by \eqref{e:ggtr1} that $\tilde{\nu}^a\tilde{A}_a=0$.
	Moreover, since solutions of the Yang--Mills equations get mapped back into solutions of the Yang--Mills equations  under gauge transformations and the Yang--Mills curvature transforms as
	$\tilde{F}_{ab}=\tilde{\mathfrak{u}}^{-1} \tilde{F}^\star_{ab} \tilde{\mathfrak{u}}$,
	the proof follows. 
\end{proof}

\begin{lemma}\label{t:slrl2}
	Suppose $(\tilde{g}^{ab}, \, \tilde{A}^\star_a)$ is a solution of the Einstein--Yang--Mills equations \eqref{eq-einstein}--\eqref{eq-maxwell-bianchi-o} that satisfies the temporal and wave gauge conditions given by \eqref{phys-temp-gauge} and \eqref{e:prewg}, respectively, 
	$\tilde{\mfu}$ solves the differential equation \eqref{gauge-ode}, and let
	\begin{align}\label{e:AAconf}
		g^{ab}=e^{2\Psi}\tilde{g}^{ab},\quad F_{ab}= e^{-\Psi} \tilde{F}_{ab} \AND A_a = e^{- \frac{\Psi}{2} } \Ab_a,
	\end{align}
	where $\tilde{F}_{ab}$ and $\tilde{A}_a$ are defined by \eqref{e:ggtsf1}. Then 	
	$(g^{ab},A_a)$ determines a solution of the reduced conformal Einstein--Yang--Mills equations, consisting of \eqref{E:CONFEIN5} and \eqref{eq-conformal-Maxwell-div}--\eqref{eq-conformal-Maxwell-Bianchi}, that satisfies the temporal gauge $A_a\nu^a=0$ and wave gauge $Z^a=0$  conditions. 		 
\end{lemma}
\begin{proof} 
	By Lemma \ref{t:slrl1}, we know that $(\tilde{g}^{ab}, \tilde{A}_a)$ determines a solution of the Einstein--Yang--Mills equations \eqref{eq-einstein}--\eqref{eq-maxwell-bianchi-o} that satisfies the temporal gauge $\tilde{A}_a\tilde{\nu}^a=0$. Recalling the definitions of $\tilde{\nu}_a$ and $\nu_a$ given by \eqref{e:dtaut} and \eqref{E:CONFFAC}, respectively, we observe that
	\begin{align*}
		\tilde{\nu}_a=-\frac{H}{\sin(Ht)}\nu_a=-e^\Psi\nu_a \AND \tilde{\nu}^b=-e^{-\Psi}\nu^b.
	\end{align*}
	With the help of these relations, it then follows from \eqref{e:AAconf} that
	$0=\tilde{A}_a\tilde{\nu}^a= - e^{\frac{\Psi}{2}}A_a e^{-\Psi}\nu^a$, and hence, that $A_a\nu^a=0$. This establishes that the temporal gauge condition is satisfied.
	
	Next, setting
	\begin{align}\label{e:hX}
		\widehat{X}^a_{ef}= & -\frac{1}{2}\bigl(\tilde{\ulg}_{bf}\nb_e\tilde{\ulg}^{ab}+\tilde{\ulg}_{eb}\nb_f \tilde{\ulg}^{ab}-\tilde{\ulg}^{ac}\tilde{\ulg}_{ed}\tilde{\ulg}_{fb}\nb_c \tilde{\ulg}^{bd}\bigr) \notag \\
		= & \ts{\delta}{^a_f} \del{e}\Psi+ \ts{\delta}{^a_e} \del{f}\Psi-\ulg^{ac}\ulg_{fe}\del{c}\Psi,
	\end{align}
	where in obtaining the second equality we have used \eqref{E:DESIT2}, \eqref{E:CONFFAC} and the fact that $\nb$ is the Levi-Civita connection of $\ulg_{ab}$,
	we observe that
	\begin{equation}\label{e:gx}
		g^{fe}\widehat{X}^a_{ef}= 2 g^{ac}\del{c}\Psi-g^{fe}\ulg^{ac}\ulg_{fe}\del{c}\Psi.
	\end{equation}
	We also observe, by expressing  $\tilde{X}^a$ in terms of the conformal metric using \eqref{e:AAconf} and employing the relations \eqref{E:WAVEGA}--\eqref{E:XYZ} and \eqref{e:gx}, that
	\begin{align*}
		\tilde{X}^a
		=&-\tilde{\nb}_e \tilde{g}^{ae}+\frac{1}{2}\tilde{g}^{ae}\tilde{g}_{df}\tilde{\nb}_e\tilde{g}^{df} \notag  \\
		=&-e^{-2\Psi}\tilde{\nb}_e g^{ae}+\frac{1}{2}e^{-2\Psi} g^{ae} g_{df}\tilde{\nb}_e g^{df}  +(2 -n )e^{-2\Psi} g^{ae}  \del{e} \Psi   \notag  \\
		=&-e^{-2\Psi}(\nb_e g^{ae}+\widehat{X}^a_{ef}g^{fe}+\widehat{X}^e_{ef}g^{af})+\frac{1}{2}e^{-2\Psi} g^{ae} g_{df} (\nb_e g^{df} +\widehat{X}^d_{ec} g^{cf}+\widehat{X}^f_{ec} g^{dc})
		\notag  \\
		&  +(2 -n )e^{-2\Psi} g^{ae}  \del{e} \Psi   \notag  \\
		=&-e^{-2\Psi}(\nb_e g^{ae}+\widehat{X}^a_{ef}g^{fe} )+\frac{1}{2}e^{-2\Psi} g^{ae} g_{df} \nb_e g^{df}   +(2 -n )e^{-2\Psi} g^{ae}  \del{e} \Psi \notag  \\
		=&e^{-2\Psi}X^a-e^{-2\Psi} \widehat{X}^a_{ef}g^{fe}    +(2 -n )e^{-2\Psi} g^{ae}  \del{e} \Psi  \notag  \\ 
		=&e^{-2\Psi}Z^a + 2 e^{-2\Psi} (\ulg^{ac}-g^{ac})\del{c} \Psi  - n e^{-2\Psi}   (g^{ac}-\ulg^{ac})\del{c}\Psi\notag \\
		& +e^{-2\Psi}(g^{fe}-\ulg^{fe})\ulg^{ac}\ulg_{fe}\del{c}\Psi.
	\end{align*}
	Since the wave gauge condition \eqref{e:prewg} is satisfied by assumption, we conclude that the same is true for wave gauge condition $Z^a=0$.
	
	Now, to complete the proof, we observe that, with the help of  Lemma \ref{t:rdein}, it is straightforward to verify via a direct calculation that $(g^{ab}, A_a)$, which is determined by \eqref{e:AAconf}, solves the reduced conformal Einstein--Yang--Mills equations consisting of \eqref{E:CONFEIN5} and \eqref{eq-conformal-Maxwell-div}--\eqref{eq-conformal-Maxwell-Bianchi}.  
\end{proof}

In the next proposition, we establish quantitative bounds on the gauge transformation $\tilde{\mathfrak{u}}$ from Lemma \ref{t:slrl2}.

\begin{proposition}\label{t:uaest}
	Suppose $(\tilde{g}^{ab}, \tilde{A}^\star_a)$ is a solution of the Einstein--Yang--Mills equations that satisfies the temporal and wave gauge conditions defined by \eqref{phys-temp-gauge} and \eqref{e:prewg}, respectively, and $\tilde{\mathfrak{u}}$ satisfies the differential equation \eqref{gauge-ode} and the initial condition $\tilde{\mathfrak{u}}|_{\tau=0}=\mathds{1}$ such that the corresponding solution $(g^{ab}, A_a)$ of the reduced conformal Einstein--Yang--Mills equations from Lemma \ref{t:slrl2}, which satisfies both the $A_a\nu^a=0$ and wave gauge $Z^a=0$, yields a solution
	\begin{align} \label{Uhat-reg}
		\widehat{\mathbf{U}} \in C^0\Bigl(\Bigl(t_*, \frac{\pi}{2H}\Bigr],H^s(\Sigma;V)\Bigr) \cap C^1\Bigl(\Bigl(t_*, \frac{\pi}{2H}\Bigr],H^{s-1}(\Sigma;V)\Bigr)
	\end{align}
	of the Fuchsian system \eqref{e:FchEYM} with $s>\frac{n+1}{2}$ on
	$\bigl(t_*,\frac{\pi}{2H}\bigr]\times \Sigma$ for some $t_* \in \bigl[0,\frac{\pi}{2H}\bigr)$ that also satisfies
	$\lVert 	\widehat{\mathbf{U}}|_{t=\frac{\pi}{2H}}\rVert\leq \delta$ where $\delta \in (0,\delta_0]$ and $\delta_0>0$ is as given in Theorem \ref{t:Uglobal}. Then 
	\begin{equation}
		\lVert \tilde{A}_a|_{\tau=0}\rVert_{H^s } = \lVert  \tilde{A}^\star_a|_{\tau=0} \rVert_{H^s },   \label{e:L2}
	\end{equation}
	and there exists a constant $C>0$, independent of $\delta \in (0,\delta_0]$ and $t_* \in \bigl[0,\frac{\pi}{2H}\bigr)$, such that $\tilde{\mfu}(\tau,x)$, $\mfu(t,x)=\tilde{\mfu}(\tau(t),x)$, $\tilde{A}_a^\star(\tau,x)$ and $\tilde{A}_a(\tau,x)$,
	where $\tau$ and $t$ are related via \eqref{E:GDINV}, are bounded by 	
	\begin{gather}
		\lVert \tilde{\mfu}(\tau(t))\rVert_{H^s} = \lVert \mfu(t)\rVert_{H^s} \leq C, \label{e:L1}\\
		\lVert \tilde{\mfu}^{-1}(\tau(t))\rVert_{H^s} = \lVert \mfu^{-1}(t)\rVert_{H^s} \leq C, \label{e:L1a}\\
		\lVert \tensor{\tilde{\ulh}}{^c_b} (d\tilde{\mfu})_c(\tau(t))\rVert_{H^s}=\lVert \ts{\ulh}{^{c}_{b}} (d \mfu)_c(t)\rVert_{H^s}  \leq C \delta ,  
		\label{e:L1c} \\
		\lVert  \nu^a (d \tilde{\mathfrak{u}} )_a (\tau(t)) \rVert_{H^s} = \lVert  \nu^a (d  \mathfrak{u} )_a (t) \rVert_{H^s} \leq   C \delta^2,  \label{e:L1d} 
		\intertext{and}
		\lVert \tilde{A}^\star_a(\tau(t))\rVert_{H^s} \leq C\lVert \tilde{A}_a(\tau(t)) \rVert_{H^s }+ C \delta  \label{e:L3}
	\end{gather}
	for all	$t\in \bigl(t_*,\frac{\pi}{2H}\bigr]$.
	Moreover,  if
	\begin{equation}\label{e:tgAreg}
		\tilde g^{a b} \in \bigcap_{\ell=0}^{s}C^\ell([0,\tau_*),H^{s-\ell+1}(\Sigma ))\AND \tilde{A}^\star_a \in \bigcap_{\ell=0}^{s }C^\ell([0,\tau_*),H^{s-\ell}(\Sigma )),
	\end{equation}
	then $(g^{ab},A_a)$ satisfies
	\begin{equation}\label{e:gAreg}
		g^{a b} \in \bigcap_{\ell=0}^{s}C^\ell\Bigl(\Bigl(t_*, \frac{\pi}{2H}\Bigr],H^{s-\ell+1}(\Sigma )\Bigr)\AND  A_a \in \bigcap_{\ell=0}^{s }C^\ell\Bigl(\Bigl(t_*, \frac{\pi}{2H}\Bigr],H^{s-\ell}(\Sigma )\Bigr),
	\end{equation}
	where $t_*= \frac{1}{H}\left(\frac{\pi}{2}-\gd(H^{-1} \tau_*)\right)$, c.f.~\eqref{E:COOR1}. 
\end{proposition}
\begin{proof}	
	Before commencing with the proof, we remark that, in the following, the constant $C>0$, which can change from line to line, is independent of $\delta\in (0,\delta_0]$ and $t_* \in \bigl[0,\frac{\pi}{2H}\bigr)$.
	
	\bigskip
	
	\noindent \underline{$(1)$ Bounds on $\widehat{\mathbf{U}}$ and the renormalized conformal Yang--Mills fields:} Since, by assumption, the solution \eqref{Uhat-reg} of the Fuchsian system \eqref{e:FchEYM}  on $\bigl(t_*,\frac{\pi}{2H}\bigr]\times \Sigma$ satisfies
	$\lVert 	\widehat{\mathbf{U}}|_{t=\frac{\pi}{2H}}\rVert\leq \delta$ where $\delta \in (0,\delta_0]$, we know from the uniqueness of solutions to \eqref{e:FchEYM}  that it must agree, where defined,  the solution to \eqref{e:FchEYM}  from Theorem \ref{t:Uglobal} that is generated from the same initial data. In particular, this implies that $\widehat{\mathbf{U}}$ is bounded by
	\begin{equation} \label{Ubfhat-bnd}    
		\norm{\widehat{\mathbf{U}}(t)}_{H^s} \leq C \delta, \quad t_* < t\leq \frac{\pi}{2H}.
	\end{equation}
	This, in turn, implies via Corollary \ref{t:impv} and the uniqueness of solutions $(\widehat{\mathbf{U}},\mathring{\mathbf{V}})$ to the Fuchsian system comprised of \eqref{e:FchEYM}  and \eqref{e:newymeq1} that the renormalized conformal Yang--Mills--fields defined by \eqref{e:newYM} are bounded by 
	\begin{equation} \label{AEH-bnd}
		\lVert \mathring A_a (t)\rVert_{H^s} + \lVert\mathring E_a (t)\rVert_{H^s} + \lVert \mathring H_{a b} (t)\rVert_{H^s}  \leq C\delta, \quad t_* < t\leq \frac{\pi}{2H}.
	\end{equation}
	
	\bigskip
	
	\noindent \underline{$(2)$ Proof of \eqref{e:L1}--\eqref{e:L1a} and \eqref{e:L1c}--\eqref{e:L1d}:} Recalling that $\tilde{T}^a$ is defined by \eqref{def-T}, we contract both sides of \eqref{e:ggtsf1} with $\tilde{T}^a$ to get
	\begin{align*}
		0= \tilde{A}^\star_a \tilde{T}^a=\tilde{\mathfrak{u}}\tilde{A}_a\tilde{T}^a\tilde{\mathfrak{u}}^{-1}-(d\tilde{\mathfrak{u}})_a \tilde{T}^a \tilde{\mathfrak{u}}^{-1},
	\end{align*}
	which, after rearranging, yields
	\begin{equation}\label{e:AT}
		\tilde{T}^a(d\tilde{\mathfrak{u}})_a=\tilde{\mathfrak{u}}\tilde{A}_a\tilde{T}^a.
	\end{equation}
	From \eqref{e:dtaut}, \eqref{CONFG-g}--\eqref{E:CONFG2}, \eqref{E:NORMT}--\eqref{decom-g} and \eqref{def-T}, we observe that $\tilde{T}^a$ can be expressed as
	\begin{equation}\label{e:tT}
		\tilde{T}^a=(-\tilde{\lambda})^{-\frac{1}{2}}\tilde{\nu}_b\tilde{g}^{ab}=-e^{-\Psi}(-\lambda)^{-\frac{1}{2}}(\xi^a-\lambda\nu^a).
	\end{equation}
	Using this along with the temporal gauge condition \eqref{e:temgg} then allows us to write \eqref{e:AT} as 
	\begin{align}\label{e:tu1}
		\tilde{T}(\tilde{\mfu})  
		=-(-\lambda)^{-\frac{1}{2}}e^{-\frac{\Psi}{2}}\mfu A_a \xi^a,
	\end{align}
	or equivalently, as
	\begin{equation}\label{e:ueq}
		(-\lambda)\nu^a\nb_a \mfu + \xi^a\nb_a \mfu = \frac{  \ttb t \sqrt{Ht}}{\sqrt{\sin(Ht)}} \mfu \mathring{A}_a p^a. 
	\end{equation}
	Noting that $\lambda$, $\mathring{A}_a$, $\xi^a$ and $p^a$ are uniformly bounded for $t\in \bigl(t_*,\frac{\pi}{2H}\bigr]$ in $H^s(\Sigma)$ and that $\lambda$ is bounded away from $0$ on account of \eqref{Ubfhat-bnd} and \eqref{AEH-bnd},  and that $\frac{  \ttb t \sqrt{Ht}}{\sqrt{\sin(Ht)}}$ is also uniformly bounded for $t\in \bigl(t_*,\frac{\pi}{2H}\bigr]$, it follows that  \eqref{e:ueq} defines a linear symmetric hyperbolic equation for $\mfu$, which satisfies the condition $\mfu|_{t=\frac{\pi}{2H}}=\mathds{1}$ by assumption. 
	We can therefore conclude from the standard theory for linear symmetric hyperbolic equations that $\mfu$ is bounded by
	\begin{equation*}
		\lVert \mfu(t)\rVert\lVert _{H^s} \leq C, \quad 0<t_* < t \leq \frac{\pi}{2H}, 
	\end{equation*}
	for some constant $C>0$. This establishes the estimate \eqref{e:L1}. Noting that $\mfu^{-1}(\nb_a\mfu) \mfu^{-1}=-\nb_a\mfu^{-1}$, the estimate \eqref{e:L1a} can also be shown to hold by similar arguments.

	\bigskip

	Using \eqref{e:tT} to express \eqref{e:tu1} as
	\begin{align}\label{e:AT1}
		(\xi^a-\lambda \nu^a) \nb_a\mfu  
		= e^{\frac{\Psi}{2}} \mfu \xi^a   A_a,
	\end{align}		
	we get from applying $\ts{\ulh}{^{c}_{b}} \nb_c$ to this equation that
	\begin{equation}\label{e:AT1-a}
		\ts{\ulh}{^{c}_{b}} \nb_c \left((\xi^a-\lambda \nu^a) \nb_a\mfu\right)  
		=\ttb t e^{\frac{\Psi}{2}} \ts{\ulh}{^{c}_{b}} (d\mfu)_c p^a   A_a+e^{\frac{\Psi}{2}}  \mfu \ts{\ulh}{^{c}_{b}} \tp{a}{c}   A_a+ \ttb t e^{\frac{\Psi}{2}} \mfu p^a   \ts{\ulh}{^{c}_{b}} \nb_c A_a,
	\end{equation}
	where in deriving this we have used the definitions of \eqref{E:V} and \eqref{E:VD}.
	In the last term in the above expression, we use the identity \eqref{e:Hpq} and \eqref{decom-F} to re-express it as
	\[\ttb t e^{\frac{\Psi}{2}} \mfu p^a  \ts{\ulh}{^{c}_{b}} \nb_c A_a = \ttb t e^{\frac{\Psi}{2}} \mfu p^a   ( e^{\frac{\Psi}{2}} H_{b a} +   \ts{\ulh}{^{c}_{b}} \nb_a A_c - e^{\frac{\Psi}{2}} \ts{\ulh}{^{c}_{b}} [A_c, A_a]). \]
	Substituting this expression into \eqref{e:AT1-a},
	we find, with the help of \eqref{E:V}--\eqref{E:VD}, that
	\begin{align}\label{e:ddu1a}
		&(\xi^a-\lambda \nu^a)
		\nb_a \left(\ts{\ulh}{^{c}_{b}} (d\mfu)_c\right) -\xi^a   \nb_a (e^{\frac{\Psi}{2}}\mfu A_b) \notag \\
		=& \ttb t e^{\frac{\Psi}{2}} \ts{\ulh}{^{c}_{b}} (d\mfu)_c p^a   A_a+e^{\frac{\Psi}{2}}  \mfu \ts{\ulh}{^{c}_{b}} \tp{a}{c}   A_a - \ts{\ulh}{^{c}_{b}} \tp{a}{c} \ts{\ulh}{^{d}_{a}} (d\mfu)_d  \notag  \\
		& + \ts{\ulh}{^{c}_{b}} m_c \nu^a\nb_a\mfu   + \ttb t \mfu p^a  e^{ \Psi }H_{ba}  -\ttb t e^{\frac{\Psi}{2}}  p^c \ts{\ulh}{^{a}_{c}}  (d\mfu)_a A_b -\ttb t e^{\frac{\Psi}{2}} \mfu p^a  [A_b,A_a] .
	\end{align}
	
	On the other hand, the evolution equations \eqref{e:AT1} and \eqref{eq-YM-potential} for $\mfu$ and $\bar A_a$, respectively, imply that
	\begin{align*}
		\lambda\nu^d\nb_d(e^{\frac{\Psi}{2}}\mfu A_b)  = -\ttb t e^{ \Psi }  \mfu A_a p^a A_b + \ttb t e^{\frac{\Psi}{2}}  p^a \ts{\ulh}{^{d}_{a}} (d\mfu)_d A_b  - e^{ \Psi }\lambda\mfu E_b.  
	\end{align*}
	Adding this to  \eqref{e:ddu1a}, while employing the renormalized conformal Yang--Mills fields \eqref{e:newYM}, yields
	\begin{align}\label{e:hychi}
		-\lambda \nu^a\nb_a\chi_b+ \xi^a  \nb_a\chi_b    
		= \Delta^\star_b
	\end{align}
	where 
	\begin{equation*}
		\chi_b= \ts{\ulh}{^{c}_{b}} (d\mfu)_c-e^{\frac{\Psi}{2}} \mfu A_b=\ts{\ulh}{^{c}_{b}} (d\mfu)_c-e^{\frac{\Psi}{2}}\sqrt{t}\mfu \mathring{A}_b
	\end{equation*} 
	and 
	\begin{align*}
		\Delta^\star_b = & \ttb t^{\frac{3}{2}} e^{\frac{\Psi}{2}} \chi_b p^a   \mathring{A}_a -	 \ts{\ulh}{^{c}_{b}}\tp{a}{c}  \chi_a   - \ttb t (-\lambda)^{-1} \ts{\ulh}{^{c}_{b}} m_c p^a \chi_a  \notag \\
		&+ \ttb t^2 \mfu p^a  \bigl(e^{ \Psi } \mathring{H}_{ba}  -  e^{\frac{\Psi}{2}}  [\mathring{A}_b,\mathring{A}_a]  \bigr)  - t e^{ \Psi }\lambda\mfu \mathring{E}_b .
	\end{align*}
	In the following, we will interpret \eqref{e:hychi} as a linear symmetric hyperbolic equation for $\chi_b$.

	Next, we note that $t e^\Psi$ is uniformly bounded for  $t\in \bigl(t_*,\frac{\pi}{2H}\bigr]$ by \eqref{E:CONFFAC}. Using this, it is then not difficult to verify with the help of the bounds \eqref{Ubfhat-bnd} and \eqref{AEH-bnd} that $\lambda$,  $\xi$ and $\Delta_b^\star$ are uniformly bounded in $H^s(\Sigma)$ for $t\in \bigl(t_*,\frac{\pi}{2H}\bigr]$ and $\lambda$ is bounded away from $0$. 
	We also note $\tilde{\mfu}|_{\tau=0}=\mfu|_{t=\frac{\pi}{2H}}=\mathds{1}$ implies that 
	\begin{equation*}
		\chi_b|_{t=\frac{\pi}{2H}}=\bigl(\ts{\ulh}{^{c}_{b}} (d \mfu)_c\bigr) |_{t=\frac{\pi}{2H}} - \bigl( e^{\frac{\Psi}{2}}\mathfrak{u} A_b\bigr)|_{t=\frac{\pi}{2H}}=- \sqrt{H}  A_b |_{t=\frac{\pi}{2H}}
	\end{equation*} 
	from which we deduce that $\lVert \chi_b|_{t=\frac{\pi}{2H}}\rVert_{H^s} \leq C \delta$.
	We can therefore conclude from the standard theory for linear symmetric hyperbolic equations that $\chi_a$ is bounded by
	\begin{equation*}
		\lVert \chi_b(t)\rVert_{H^s} \leq C\delta, \quad 0<t_* < t \leq \frac{\pi}{2H}, 
	\end{equation*}
	for some constant $C>0$.
	With the help of the identities 
	\begin{equation} \label{projector-relations}
		\tensor{\tilde{\ulh}}{^c_b}=\tensor{\delta}{^c_b}+\tilde{\nu}^c\tilde{\nu}_b=\tensor{\delta}{^c_b}+\nu^c\nu_b=\ts{\ulh}{^{c}_{b}},
	\end{equation} which are a consequence of \eqref{E:NORMT}, \eqref{E:PRO} and \eqref{E:NORMT-a}--\eqref{e:def-h2}, it then follows from the above estimate and the bounds from  Corollary \ref{t:impv} that
	\begin{equation*}
		\lVert \tensor{\tilde{\ulh}}{^c_b}(d\tilde{\mfu})_c\rVert_{H^s}= \lVert \ts{\ulh}{^{c}_{b}} (d \mfu)_c\rVert_{H^s} \leq \lVert \chi_b\rVert_{H^s}+e^{\frac{\Psi}{2}}\sqrt{t} \lVert  \mfu \mathring{A}_b\rVert_{H^s} \leq C \delta, \quad  0<t_* < t \leq \frac{\pi}{2H}.
	\end{equation*}
	Making use of the above estimate for $\lVert \ts{\ulh}{^{c}_{b}} (d \mfu)_c\rVert_{H^s}$ and \eqref{e:ueq}, we obtain
	\[ \lVert  \nu^a (d \mathfrak{u} )_a \rVert_{H^s}  \leq  C \biggl(\lVert \xi^a\nb_a \mfu\rVert_{H^s} + \biggl\lVert \frac{  \ttb t \sqrt{Ht}}{\sqrt{\sin(Ht)}} \mfu \mathring{A}_a p^a\biggr\rVert_{H^s} \biggr)\leq C \delta^2, \quad  0<t_* < t \leq \frac{\pi}{2H}.  \]
	
	\bigskip
	
	\noindent \underline{$(3)$ Proof of \eqref{e:L2}:}	 
	Since  the initial condition $\tilde{\mfu}|_{\tau=0}= \mathds{1}$ (i.e., identity map on $\Sigma_{\tau=0}$),  it is clear that $(\tensor{\tilde{\ulh}}{^c_b} (d\tilde{\mfu})_c)|_{\tau=0}=0$. From this, we then have, with the help of \eqref{e:ggtsf1} and \eqref{e:ggtr1}, that
	\begin{align*}
		& \tilde{A}_b |_{\tau=0}= \bigl(\tilde{\mathfrak{u}}^{-1}\tilde{A}^\star_b\tilde{\mathfrak{u}}\bigr)|_{\tau=0}-  \tilde{\nu}_b\bigl( \tilde{\mathfrak{u}}^{-1}(d\tilde{\mathfrak{u}})_c\tilde{\nu}^c\bigr)|_{\tau=0}=\bigl( \tilde{\mathfrak{u}}^{-1}\tilde{A}^\star_b\tilde{\mathfrak{u}}\bigr)|_{\tau=0}+ \tilde{\nu}_b\bigl( \tilde{\mathfrak{u}}^{-1}\tilde{\nu}^a\tilde{A}^\star_a\tilde{\mathfrak{u}}\bigr)|_{\tau=0} \notag \\
		&\hspace{2cm}	=\bigl(\tilde{\mathfrak{u}}^{-1}\tensor{\tilde{\ulh}}{^c_b}\tilde{A}^\star_c\tilde{\mathfrak{u}}\bigr)|_{\tau=0}. 
	\end{align*}
	Thus by \eqref{projector-relations} (i.e.~ $\tensor{\tilde{\ulh}}{^c_b} =\ts{\ulh}{^{c}_{b}}$) and $\tilde{\mfu}|_{\tau=0}= \mathds{1}$, we conclude that
	\begin{equation*}
		\lVert \tilde{A}_a|_{\tau=0}\rVert_{H^s } =   \lVert \ts{\ulh}{^{c}_{a}} (\tilde{\mathfrak{u}}^{-1}\tilde{A}^\star_c \tilde{\mathfrak{u}})|_{\tau=0}\rVert_{H^s }  
		= \lVert  \tilde{A}^\star_a |_{\tau=0}\rVert_{H^s },
	\end{equation*}
	where in deriving the second equality we have used the fact that $\tilde{A}^\star_a |_{\tau=0}$ is a spatial tensor, which implies that $\lVert  \tilde{A}^\star_a |_{\tau=0}\rVert_{H^s } =\lVert  \ts{\ulh}{^c_a} \tilde{A}^\star_c |_{\tau=0}\rVert_{H^s }$.

	\bigskip
	
	\noindent \underline{$(4)$ Proof of \eqref{e:L3}:}
	Due to the temporal gauge condition $\tilde{A}^\star_a\tilde{T}^a=0$ and \eqref{projector-relations}, we deduce from the gauge transformation law \eqref{e:ggtsf1}--\eqref{e:ggtr1}  and the  estimates \eqref{e:L1}--\eqref{e:L1d} that 
	\begin{align*}
		\lVert \tilde{A}^\star_a \nu^a \rVert_{H^s } &
		\leq  \lVert  \nu^a (d\tilde{\mathfrak{u}})_a \tilde{\mathfrak{u}}^{-1}\rVert_{H^s } \leq C \delta^2
		\intertext{and}
		\lVert  \ts{\ulh}{^a_b} \tilde{A}^\star_a\rVert_{H^s } &
		\leq   \lVert \ts{\ulh}{^a_b} \tilde{\mathfrak{u}} \tilde{A}_a \tilde{\mathfrak{u}}^{-1}\rVert_{H^s } + \lVert  \ts{\ulh}{^a_b} (d\tilde{\mathfrak{u}})_a \tilde{\mathfrak{u}}^{-1}\rVert_{H^s } \notag  \\
		& \leq  C\lVert  \tilde{A}_a \rVert_{H^s } 
		+ C \lVert \tensor{\tilde{\ulh}}{^a_b} (d\tilde{\mathfrak{u}})_a\rVert_{H^{s } }  \leq C
		\bigl(\lVert  \tilde{A}_a \rVert_{H^s }+  \delta\bigr) 
	\end{align*}
	for all $t\in \bigl(t_*,\frac{\pi}{2H}\bigr]$, which leads to \eqref{e:L3} according to the Sobolev norms given in  \S\ref{s:norm}.
	
	\bigskip
	
	\noindent\underline{$(5)$ Solution regularity:} To complete the proof, assume that the solution $(\tilde{g}^{ab}, \tilde{A}^\star_a)$ of the Einstein--Yang--Mills equations that satisfies \eqref{e:tgAreg}. Noting by \eqref{e:ggtsf1} that
	\begin{equation*}
		A_a(t,x)=e^{-\frac{\Psi}{2}}\tilde{\mfu}^{-1}  (\tau(t),x)\tilde{A}^\star_a(\tau(t),x) \tilde{\mfu}(\tau(t),x)+e^{-\frac{\Psi}{2}} \tilde{\mfu}^{-1} (\tau(t),x)(d\tilde{\mfu})_a (\tau(t),x),
	\end{equation*}
	it then not difficult to verify from the definitions \eqref{e:tgAreg} and the estimate
	\eqref{e:L1}--\eqref{e:L1c} that the solution
	$(g^{ab},A_a)$ of the reduced conformal Einstein--Yang--Mills equations satisfies \eqref{e:gAreg}. 
\end{proof}

\section{Proof of Theorem \label{mainthm-proof}}

With the help of the technical results established in \S\ref{s:verif} and \S\ref{s:maprf}, we are now in a position to prove Theorem \ref{t:mainthm}, the main result of this article.  
We break the proof of this theorem into five distinct steps.

\bigskip

\noindent \underline{Step $1$: A local solution of the Einstein--Yang--Mills equations.}
We first recall the local well-posedness result from the companion paper \cite{LW2021b} and assume that the Einstein--Yang--Mills initial data, see \eqref{idata}, is chosen as in the statement of Theorem \ref{t:mainthm}, which in particular, implies that it satisfies Einstein--Yang--Mills constraints \eqref{E:constraintA}--\eqref{E:constraintB} and the gauge constraints \eqref{E:constraintC}--\eqref{E:constraintD}
on $\Sigma_0=\{0\}\times \Sigma$. Then by Theorem $1.1$ from \cite{LW2021b} there exists a $\tau_* > 0$ and a unique solution $( \tilde g^{a b}, \,\tilde{A}^\star_a)$ to the Einstein--Yang--Mills equations satisfying the temporal and wave gauge conditions \eqref{phys-temp-gauge} and\eqref{e:prewg}, respectively, and with the regularity
\begin{equation*}
	\tilde g^{a b} \in \bigcap_{\ell=0}^{s}C^\ell([0,\tau_*),H^{s-\ell+1}(\Sigma ))\AND \tilde{A}^\star_a,\,  \tilde{F}^\star_{a b} \in \bigcap_{\ell=0}^{s }C^\ell([0,\tau_*),H^{s-\ell}(\Sigma )).
\end{equation*}
Moreover, if
\begin{equation*}
	\Bigl\lVert  \Bigl( \tilde \nu^d \tensor{\tilde g}{^{ab}_d}, \, \tensor{\tilde \ulh}{^d_{c}} \tensor{\tilde g}{^{a b}_d}, \, \tilde g^{ab}, \, \tilde E^\star_a, \, \tilde A^\star_d, \, \tilde H^\star_{ab} \Bigr)\Bigr\rVert_{\Li ([0,\tau_*), W^{1,\infty})}<\infty,
\end{equation*}
then the solution $\bigl( \tilde \nu^d \tensor{\tilde g}{^{ab}_d}, \, \tensor{\tilde \ulh}{^d_{c}} \tensor{\tilde g}{^{a b}_d}, \, \tilde g^{ab}, \, \tilde E^\star_a, \, \tilde A^\star_d , \, \tilde H^\star_{ab}\bigr)$ can be uniquely continued, in the above temporal and wave gauge, as a classical solution with the same regularity to a larger time interval $\tau\in[0,\tau^*)$ where $\tau^*\in(\tau_*, + \infty)$.

\bigskip

\noindent \underline{Step $2$: A local solution of the reduced conformal EYM system.}
By Lemma \ref{t:slrl2} and Proposition \ref{t:uaest}, the solution $(\tilde{g}_{ab},\, \tilde{A}^\star_a)$ from Step 1 implies the existence of a unique solution 
\begin{equation*}\label{e:gAreg2}
	(g^{a b},A_a) \in \bigcap_{\ell=0}^{s}C^\ell\Bigl(\Bigl(t_*, \frac{\pi}{2H}\Bigr],H^{s-\ell+1}(\Sigma )\Bigr)\times \bigcap_{\ell=0}^{s }C^\ell\Bigl(\Bigl(t_*, \frac{\pi}{2H}\Bigr],H^{s-\ell}(\Sigma )\Bigr)
\end{equation*}
on $\bigl(t_*, \frac{\pi}{2H}\bigr]\times \Sigma$ of the reduced conformal Einstein--Yang--Mills equations, see \eqref{E:CONFEIN5} and \eqref{eq-conformal-Maxwell-div}--\eqref{eq-conformal-Maxwell-Bianchi}, that satisfies both the temporal and wave gauge conditions, that is, \eqref{e:temgg} and \eqref{E:CONSTR1}. 

\bigskip

\noindent \underline{Step $3$: A local solution of the Fuchsian system \eqref{e:FchEYM}.}
By Theorems \ref{thm-FOSHS-m} and \ref{thm-FOSHS-Maxwell}, the solution $(g^{ab}, A_a)$ of the reduced conformal Einstein--Yang--Mills equations from Step 2 determines via \eqref{decom-g}--\eqref{def-MYM} and \eqref{e:hatu} a solution 
\begin{align*}
	\widehat{\mathbf{U}} \in C^0\Bigl(\Bigl(t_*, \frac{\pi}{2H}\Bigr],H^s(\Sigma;V)\Bigr) \cap C^1\Bigl(\Bigl(t_*, \frac{\pi}{2H}\Bigr],H^{s-1}(\Sigma;V)\Bigr)\cap \Li \Bigl(\Bigl(t_*, \frac{\pi}{2H}\Bigr],H^s(\Sigma;V)\Bigr)
\end{align*}
on $\bigl(t_*, \frac{\pi}{2H}\bigr]\times \Sigma$ 
of the Fuchsian equation \eqref{e:FchEYM}. 

\bigskip

\noindent\underline{Step $4$: Fuchsian initial data bounds.} By assumption, the initial data set satisfies
\begin{align*}
	\lVert (\tilde{g}^{ab}-\tilde{\ulg}^{ab},\,\tilde{\nb}_{d}\tilde{g}^{ab},\, \tilde{A}^\star_a,\, d\tilde{A}^\star_{a b}, \, \tilde{E}^\star_{b})|_{\tau=0} \rVert_{H^s} < \delta.
\end{align*} 
Using this bound, it then not difficult to verify from the conformal transformations \eqref{CONFG-g}--\eqref{CONFG-A}, the gauge transformation \eqref{e:ggtsf1}, where we recall that $\tilde{\mathfrak{u}}|_{\tau=0}=\mathds{1}$,  and the bound \eqref{e:L2} from Proposition \ref{t:uaest}
that there exists a $\delta_0>0$ and a constant $C>0$, such that  
\begin{equation*}\mathbf{U}_0 = \mathbf{U}|_{t=\frac{\pi}{2H}},
\end{equation*} 
where $\mathbf{U}$ is defined by \eqref{def-U}
is bounded by
\begin{align*}
	\lVert \mathbf{U}_0\rVert_{H^s} \leq &C  \lVert (g^{ab}- \ulg^{ab},\nb_dg^{ab}, \, A_a,E_{a},H_{ab})|_{t=\frac{\pi}{2H}}\rVert_{H^{s }}    \notag \\ 
	\leq &  C  \lVert (\tilde{g}^{ab}- \tilde{\ulg}^{ab},\tilde{\nb}_d \tilde{g}^{ab},\tilde{A}_a,\tilde{E}_{a}, \tilde{H}_{ab} )|_{\tau=0}\rVert_{H^{s }}   \notag  \\ 
	\leq & C\lVert (\tilde{g}^{ab}-\tilde{\ulg}^{ab},\,\tilde{\nb}_{d}\tilde{g}^{ab},\, \tilde{A}^\star_a,\, d\tilde{A}^\star_{a b}, \, \tilde{E}^\star_{b})|_{\tau=0} \rVert_{H^s} 
	<   C \delta 
\end{align*}
for all $\delta \in [0,\delta_0)$, and $\tilde{E}_{c}$ and $\tilde{H}_{ab}$ are defined by $\tilde{E}_{c} = \tensor{\tilde{\ulh}}{^a_c} \tilde \nu^b \tilde{F}_{ab}$ and
$\tilde{H}_{ab}=  \tensor{\tilde{\ulh}}{^c_a}  \tensor{\tilde{\ulh}}{^d_b}\tilde{F}_{c d}$, 
respectively.
We further observe from the definition \eqref{e:hatu} of $\widehat{\mathbf{U}}$ that
\begin{equation}\label{wbfU0-bnd-a}
	\lVert \widehat{\mathbf{U}}\rVert_{H^s} \leq C \lVert \mathbf{U}\rVert_{H^s}
\end{equation}
for some constant $C>0$ that is independent of $\widehat{\mathbf{U}}$ and $\mathbf{U}$. From the above two inequalities, we conclude that
$\widehat{\mathbf{U}}_0 = \widehat{\mathbf{U}}|_{t=\frac{\pi}{2H}}$
is bounded by
\begin{equation} \label{wbfU0-bnd}
	\lVert \widehat{\mathbf{U}}_0\rVert_{H^s} \leq C\delta, \quad 0\leq  \delta < \delta_0.
\end{equation}

\bigskip

\noindent \underline{Step $5$: Global existence and stability.}
Assuming that $\delta$ is sufficiently small, the bound \eqref{wbfU0-bnd} implies via Theorem \ref{t:Uglobal} that $\widehat{\mathbf{U}}$ satisfies the energy estimate
\begin{equation}\label{e:ineq1-hatU}
	\lVert \widehat{\mathbf{U}}(t)\rVert_{H^s}^2 + \int_t^{\frac{\pi}{2H}}\frac{1}{\tau}\lVert \widehat{\Pbb} \widehat{\mathbf{U}}(\tau)\rVert^2_{H^s}d\tau \leq 
	C \delta^2, \quad t_*<t\leq \frac{\pi}{2H},
\end{equation}
where the constant $C$ is independent of $\delta \in [0,\delta_0)$ and $t_*\in \bigl[0,\frac{\pi}{2H}\bigr)$.

Next, we will show that the uniform bound \eqref{e:ineq1-hatU} on $\widehat{\mathbf{U}}$ implies a corresponding uniform bound on the physical variables $\tilde{g}^{ab} - \tilde \ulg_{a b}, \tilde{\nb}_{d}\tilde{g}^{ab},\tilde{F}^\star_{ab},\tilde{A}^\star_a$. To this end, we observe from \eqref{CONFG-g}  and \eqref{E:CONFFAC} that  
\begin{align*}
	\tilde{\nb}_c\tilde{g}^{ab}=&\tilde{\nb}_c(\tilde{g}^{ab}-\tilde{\ulg}^{ab})= \tilde{\nb}_c e^{-2\Psi}(g^{ab}- \ulg^{ab})+e^{-2\Psi}\tilde{\nb}_c (g^{ab}- \ulg^{ab}) \notag  \\
	=&  \frac{2e^{-2\Psi}}{\tan(Ht)} \nu_c (g^{ab}-\ulg^{ab})+e^{-2\Psi} [\nb_c g^{ab}+\widehat{X}^a_{cd}(g^{db}- \ulg^{db})+\widehat{X}^b_{cd}(g^{ad}- \ulg^{ad})]
\end{align*}
where we recall that $\widehat{X}^b_{cd}$ is defined by \eqref{e:hX}. Applying the $H^s$ norm to this expression yields
\begin{align}\label{e:nbg0}
	\lVert \tilde{\nu}^c\tilde{\nb}_c\tilde{g}^{ab}  \rVert_{H^s} + \lVert  \tensor{\tilde \ulh}{^c_d}\tilde{\nb}_c\tilde{g}^{ab}\rVert_{H^s} \leq  Ct^2\lVert g^{ab}- \ulg^{ab}\rVert_{H^s}+e^{-2\Psi}\lVert \nb_c g^{ab}\rVert_{H^s}.
\end{align}
From \eqref{E:W}--\eqref{E:q}, we also observe via a straightforward calculation that 
\begin{align}\label{e:g1}
	g^{ab}-\ulg^{ab}=& \bigl(e^{\frac{\tta tm-s}{n-3}}-1\bigr)s^{ab}+\bigl(e^{\frac{\tta tm-s}{n-3}}-1\bigr)\ulh^{ab}-\tta t m \nu^a\nu^b-2\ttb t p^{(b}\nu^{a)}
\end{align}
and
\begin{align}\label{e:nbg1}
	\nb_c g^{ab}= & \frac{m_c+\frac{\tta}{\ttj H} m \nu_c-s_c}{n-3}(s^{ab}-\ulh^{ab})e^{\frac{\tta tm-s}{n-3}}+e^{\frac{\tta tm-s}{n-3}} \ts{s}{^{ab}_c} -\Bigl(m_c+\frac{\tta}{\ttj H} m \nu_c\Bigr)\nu^a\nu^b\notag  \\
	& -\Bigl(\tp{b}{c}+\frac{B}{\ttk H}p^b\nu_c\Bigr)\nu^a-\Bigl(\tp{a}{c}+\frac{B}{\ttk H}p^a\nu_c\Bigr)\nu^b.
\end{align}
Then, with the help of \eqref{E:CONFFAC}, \eqref{def-U}, \eqref{wbfU0-bnd-a} and \eqref{e:nbg0}--\eqref{e:nbg1}, we see that the uniform bound \eqref{e:ineq1-hatU} implies that
\begin{align}
	& \lVert  (\tilde{g}^{ab}-\tilde{\ulg}^{ab})(\tau) \rVert_{H^{s+1}}+\lVert \tilde{\nu}^c\tilde{\nb}_c\tilde{g}^{ab}(\tau) \rVert_{H^s} \notag \\
	\leq &  Ct^2\lVert (g^{ab}- \ulg^{ab})(t)\rVert_{H^s}+e^{-2\Psi}\lVert \nb_c g^{ab}(t)\rVert_{H^s} \notag  \\
	\leq & C t^2 \lVert  \mathbf{U}(t) \rVert_{H^{s}} \leq  C \delta t^2 \label{decay-tilde g}
\end{align}
hold for all $\tau \in [0,\tau_*)$. We also observe as a consequence of the bound \eqref{e:L3} from Proposition \ref{t:uaest} and the bound from Corollary \ref{t:impv} on the renormalized conformal Yang--Mills fields \eqref{e:newYM} that
\begin{align}
	& \lVert \tilde{A}^\star_a (\tau)\rVert_{H^s}+ \lVert \tilde{E}^\star_{a}(\tau) \rVert_{H^s} +  \lVert \tilde{H}^\star_{ab}(\tau) \rVert_{H^s} \notag \\
	\leq & C\bigl(  \lVert  \tilde{A}_a(\tau) \rVert_{H^s}  + \delta +  \lVert \tilde{E}_{a}(\tau) \rVert_{H^s} +\lVert \tilde{H}_{ab}(\tau) \rVert_{H^s}\bigr)  \nnb \\
	\leq {}& C\bigl( \lVert \mathring{A}_a(t) \rVert_{H^s}   +  \delta + \lVert \mathring{E}_{a}(t) \rVert_{H^s} + \lVert \mathring{H}_{ab}(t) \rVert_{H^s}\bigr) \leq C \delta, \label{uniform-tilde YM}
\end{align}
for all $\tau \in [0,\tau_*)$,
where we have set
$\tilde{E}^\star_{c} = \tensor{\tilde h}{^a_c} \tilde T^b \tilde{F}^\star_{ab}$ and $\tilde{H}^\star_{ab}=  \tensor{\tilde h}{^c_a}  \tensor{\tilde h}{^d_b} \tilde{F}^\star_{c d}$.
Together the estimates \eqref{decay-tilde g} and \eqref{uniform-tilde YM} imply that 
\begin{equation}\label{e:ctprn}
	\lVert  ( \tilde{\nu}^d \tensor{\tilde{g}}{^{ab}_d}, \, \tensor{\tilde{\ulh}}{^d_{c}} \tensor{\tilde{g}}{^{a b}_d}, \, \tilde{g}^{ab}, \, \tilde{E}^\star_a, \, \tilde{A}^\star_d, \, \tilde{H}^\star_{ab} )\rVert_{\Li ([0,\tau_*), W^{1,\infty})}<\infty.
\end{equation}

Now, by way of contradiction, assume that  $\tau_*$ is the maximal time of existence for the solution $( \tilde g^{a b}, \,\tilde{A}^\star_a,)$ and that it is finite, i.e. $\tau_*<\infty$. Then since the bound \eqref{e:ctprn} implies by the continuation principle from Step 1 that the solution $( \tilde g^{a b}, \,\tilde{A}^\star_a,)$ can be continued to a larger time interval, we arrive at a contradiction. Thus,  we must have that $\tau_*=\infty$, which establishes that the solution $( \tilde g^{a b}, \,\tilde{A}^\star_a)$ exists globally on $[0,\infty)\times \Sigma$ and satisfies the uniform bounds \eqref{e:ctprn}.

\appendix

\section{Technical calculations}\label{s:App1}
\subsection{Geometric identities}\label{s:App2} In this appendix, we state a number of geometric identities that will be used throughout the article. Since the derivation of the identities is straightforward (although, it should be noted that some are quite lengthy), we, for the most part, omit the details.  
Below, we use the same notation for geometric objects as in \S\ref{s:ds}--\S\ref{sec:gauge-conditions}, \S\ref{sec-Fuchsian-field} and \S\ref{sec:prelim}.

\begin{lemma}\label{t:nbnuh}
	The vector field $\nu^a$ and projection operator $\tensor{\ulh}{^a_b}$ defined by \eqref{E:NORMT} and \eqref{E:PRO}, respectively, satisfy
	\al{NBNU}{
		\udn{a}\nu_b=0, \quad \udn{a}\nu^b=0, \quad \udn{c}\tensor{\ulh}{^a_b}=0,
	}
	and
	\al{RELBABLA}{
		-2\nabla^{(a}\nu^{b)}=\mathcal{L}_{\nu} g^{ab}=\nu^c\udn{c}g^{ab}.
	}		
\end{lemma}

\begin{lemma}\label{t:Rdop}
	The Ricci tensor $R^{ab}$ of the conformal metric $g_{ab}$ can be expressed in terms of the connection $\udl{\nabla}_a$ and curvature $\tensor{\udl{R}}{_{cde}^a}$ of the conformal de Sitter metric $\ulg_{ab}$, see \eqref{E:DESIT2}, by 
	\al{RICCI}{
		R^{ab}=\frac{1}{2} g^{cd}\udn{c}\udn{d}g^{ab}+\nabla^{(a}X^{b)}+\udl{R}^{ab}+P^{ab}(g^{-1})+Q^{ab}(g^{-1}, \udn{} g^{-1}),
	}
	where $X$ is given by \eqref{def-X},
	\als{
		P^{ab}(g^{-1})
		=& -\frac{1}{2} (g^{ac}-\ulg^{ac})\ulg^{de}\tensor{\udl{R}}{_{cde}^b} -\frac{1}{2} \ulg^{ac}(g^{de}-\ulg^{de})\tensor{\udl{R}}{_{cde}^b}  \\
		&-\frac{1}{2} (g^{ac}-\ulg^{ac})(g^{de}-\ulg^{de})\tensor{\udl{R}}{_{cde}^b}  -\frac{1}{2}(g^{bc}-\ulg^{bc})\ulg^{de}\tensor{\udl{R}}{_{cde}^a} \\
		&-\frac{1}{2}\ulg^{bc}(g^{de}-\ulg^{de})\tensor{\udl{R}}{_{cde}^a}-\frac{1}{2}(g^{bc}-\ulg^{bc})(g^{de}-\ulg^{de})\tensor{\udl{R}}{_{cde}^a}
	}
	and
	\als{
		Q^{ab}(g^{-1}, \udn{} g^{-1})   
		={}& -\frac{1}{4}\bigl(g^{ad}g^{bf}\udn{d}g_{ce}\udn{f}g^{ce}+g^{ad}g^{bf}\udn{f}g_{ce}\udn{d}g^{ce} \\
		& +g_{ef}g^{ac}\udn{c}g^{bd}\udn{d}g^{ef}  +g_{ef}g^{bd}\udn{d}g^{ac}\udn{c}g^{ef}\bigr)  \\
		&+\frac{1}{2}\big(g^{ae}g_{fc}\udn{e} g^{bd} \udn{d}g^{fc}+g^{ae} g^{bd} \udn{e} g_{fc} \udn{d} g^{fc} \\
		& -\udn{c}g^{ae}\udn{e}g^{cb}-\udn{e}g^{bc}\udn{c} g^{ea}\big) -g^{ac}\tensor{X}{^b_{cd}}X^d   \\
		&+g^{af}g^{bd}\tensor{X}{^{c}_{fd}}\tensor{X}{^e _{ce}}-g^{af}g^{bd}\tensor{X}{^{c}_{ed}}\tensor{X}{^e _{cf}} \\
		&+\tensor{X}{^e_{ed}} g^{af} \udn{f} g^{bd} -\tensor{X}{^e_{fd}}g^{af}\udn{e} g^{bd}-\tensor{X}{^e_{fd}}g^{bd}\udn{e} g^{af}.
	}
\end{lemma} 

\begin{lemma}\label{t:conf2}
	The conformal scalar $\Psi$ defined by \eqref{E:CONFFAC}  satisfies
	\gat{
		\nabla_a\Psi=-\frac{1}{\tan(Ht)} \nu_a,  \label{E:CAL1}\\
		\nabla_a\nabla_b\Psi=\frac{1}{\sin^2(Ht)}\nu_a\nu_b - \frac{1}{\tan(Ht)} \nabla_a\nu_b, \label{E:CAL2}
		\\
		\Box \Psi = \lambda \frac{1}{\sin^2(Ht)}+  \frac{1}{\tan(Ht)} X^c\nu_c , 
		\label{E:CAL3}
		\intertext{and}
		g^{ab}\nabla_a\Psi\nabla_b \Psi=\lambda \frac{1}{\tan^2(Ht)}. \label{E:CAL4}
	}
\end{lemma}

\begin{lemma}\label{lem-identity}
	The conformal metric $g^{a b}$ can be represented as
	\begin{equation}\label{E:g-h}
		g^{a b} = h^{a b} + \lambda \nu^a \nu^b - 2\xi^{(a} \nu^{b)}
	\end{equation}
	where $\nu^a$, $h^{ab}$ and $\xi^a$ are defined by \eqref{E:NORMT} and \eqref{decom-g}.
	Moreover the tensor $Q^{edc}$ defined by \eqref{E:BTENSOR} satisfies the following relations:
	\gat{
		Q^{edc} \nu_c =   h^{e d} - \lambda \nu^e \nu^d,  \quad
		Q^{edc} \ts{\ulh}{^{a}_{d}} =  \nu^e h^{c a} - \nu^c h^{e a}, \quad				Q^{edc} \nu_d=   2 \nu^e \xi^c - h^{e c} - \lambda \nu^c \nu^{e}, \notag 
		\\
		Q^{edc}\nu_d\nu_c =   \lambda \nu^e,  \quad
		\ts{\ulh}{^{b}_{e}} Q^{edc} \ts{\ulh}{^{a}_{d}} =   - \nu^c h^{a b}.   \quad
		\nu_e Q^{edc} \ts{\ulh}{^{a}_{d}}  =  - h^{a c}, \notag  
		\\
		\nu_e Q^{edc} \nu_d \nu_c= -\lambda,  \quad  \nu_e Q^{edc} \nu_d \ts{\ulh}{^{a}_{c}}  = -2 \xi^a, \quad
		\nu_e Q^{edc} \ts{\ulh}{^{a}_{d}} \nu_c=0, \nnb \\  \nu_e Q^{edc} \ts{\ulh}{^{a}_{d}}  \ts{\ulh}{^{b}_{c}} = - h^{a b}, \quad
		\ts{\ulh}{^{a}_{e}} Q^{edc} \ts{\ulh}{^{b}_{d}} \nu_c = h^{a b}, \quad \ts{\ulh}{^{a}_{e}} Q^{edc} \ts{\ulh}{^{b}_{d}} \ts{\ulh}{^{f}_{c}} = 0, \nnb\\
		\ts{\ulh}{^{a}_{e}} Q^{edc} \nu_d = - h^{a c}, \quad \ts{\ulh}{^{b}_{e}} Q^{edc} \nu_d \ts{\ulh}{^{a}_{c}}  =  - h^{a b} \quad \text{and}\quad  \ts{\ulh}{^{a}_{e}} Q^{edc} \nu_d \nu_c = 0. \nnb
	}
	
\end{lemma}

\begin{lemma}\label{lem-model-symm}
	Suppose $\zeta$  and $\zeta_d = \nb_d \zeta$ satisfy the equation
	\al{model-system-1}{
		\mathbf{A}^c\udn{c}\p{\zeta_d \\ \zeta }=\frac{1}{Ht} \mathcal{B} \p{\zeta_d\\ \zeta}+G 
	}
	where
	\begin{equation}
		\mathbf{A}^c=\p{ Q^{edc} & 0   \\
			0 & -\nu^c }. \nnb
	\end{equation} 
	Then $\nu^e \zeta_e$, $\ts{\ulh}{^e_{\he}}\zeta_e$ and $\zeta$ satisfy
	\gat{
		\bar{\mathbf{A}}^c\udn{c}\p{-\nu^e \zeta_e \\ \ts{\ulh}{^e_{\he}} \zeta_e \\ \zeta }=\frac{1}{Ht} \bar{\mathcal{B}} \p{-\nu^e \zeta_e \\ \ts{\ulh}{^e_{\he}} \zeta_e \\ \zeta } +\bar{G} \nnb
	}
	where
	\gat{
		\bar{\mathbf{A}}^c=  \p{\nu_e & 0 \\ \ts{\ulh}{^f_e} & 0 \\ 0 & 1 } A^c \p{\nu_d & \ts{\ulh}{^{\he}_d} & 0 \\0 & 0 & 1} = \p{ \nu_e Q^{edc}\nu_d &  \nu_e Q^{edc}\ts{\ulh}{^{\he}_d} & 0 \\
			\ts{\ulh}{^f_e}Q^{edc}\nu_d  &  \ts{\ulh}{^f_e}Q^{edc}\ts{\ulh}{^{\he}_d} & 0 \\ 0 & 0 & -\nu^{c}}, \nnb
	} 
	\gat{
		\bar{\mathcal{B}} = \p{\nu_e & 0 \\ \ts{\ulh}{^f_e} & 0 \\ 0 & 1 } \mathcal{B} \p{\nu_d & \ts{\ulh}{^{\he}_d} & 0 \\0 & 0 & 1} \quad \text{and}\quad
		\bar{G}=\p{\nu_e & 0 \\ \ts{\ulh}{^f_e} & 0 \\ 0 & 1 } G.  \nnb
	}
\end{lemma} 
\begin{proof}
	Noting the decomposition
	\als{
		\p{\zeta_d \\ \zeta}=Z \p{-\nu^e \zeta_e \\ \ts{\ulh}{^e_{\he}} \zeta_e \\ \zeta},
	}
	where	
	\als{
		Z=\p{\nu_d & \ts{\ulh}{^{\he}_d} & 0 \\0 & 0 & 1},
	}
	the proof follows from applying 
	\als{
		Z^{\tr}=\p{\nu_e & 0 \\ \ts{\ulh}{^f_e} & 0 \\ 0 & 1 }	}
	to \eqref{E:model-system-1} and using Lemmas \ref{t:nbnuh} and \ref{lem-model-symm} to obtain the stated equation.
\end{proof}

\subsection{Calculations for the proof of Theorem \ref{thm-FOSHS-m}}
The following three lemmas contain the detailed calculation needed to complete the proof of Theorem \ref{thm-FOSHS-m}.
\begin{lemma}\label{lem-FOSHS-p}
	The conformal Einstein equation \eqref{E:TSEQ} for $\xi$ can be expressed in the first order form as 
	\al{SYS1a}{
		-\bar{\mathbf{A}}_2^0 \nu^c \nb_{c} \p{-\nu^e \tp{a}{e} \\ \ts{\ulh}{^e_{\he}} \tp{a}{e} \\ p^a } + \bar{\mathbf{A}}_2^c \ts{\ulh}{^{b}_{c}}  \udn{b} \p{-\nu^e \tp{a}{e} \\ \ts{\ulh}{^e_{\he}} \tp{a}{e} \\ p^a } =\frac{1}{Ht}\bar{\mathcal{B}}_2  \p{-\nu^e \tp{a}{e} \\ \ts{\ulh}{^e_{\he}} \tp{a}{e} \\ p^a } +\bar{G}_2
	}
	where
	\gas{
		\bar{\mathbf{A}}_2^0 =  \p{-\lambda & 0 & 0\\ 0 & h^{f \he} & 0 \\ 0 & 0 & \ttf (-\lambda)}, \quad
		\bar{\mathbf{A}}_2^c \ts{\ulh}{^{b}_{c}}  =\p{ -2\xi^b & -\ts{h}{^{\he b}} & 0 \\
			-\ts{h}{^{f b}} & 0 & 0 \\ 0 & 0 & 0}, \quad \bar{G}_2=\p{\nu_e \triangle^{ea}_2(t, \mathbf{U}) \\ \ts{\ulh}{^f_e} \triangle^{ea}_2(t, \mathbf{U}) \\ 0 }\\
		\bar{\mathcal{B} }_2 =  \p{\left(n-2-\frac{1}{\ttk} \right)(-\lambda) & 0 & \left(1-\frac{1}{\ttk}\right)\left(\frac{1}{\ttk}-n+2\right)\frac{\ttb}{H}(-\lambda) \\
			0 &  \frac{1}{\ttk} \ts{h}{^{f \he}} & 0 \\
			\ttf \frac{H}{\ttb}(-\lambda) & 0 & \ttf \left(\frac{1}{\ttk}-1\right) (-\lambda)}, \nnb
	}
	$\ttf$ is a constant to be determined and the map  $\triangle^{ea}_2(t, \mathbf{U})$, which is analytic for $(t,\mathbf{U})\in \bigl(-\iota,\frac{\pi}{H}\bigr)\times B_R(0)$ for $\iota,R>0$ small enough, is given by
	\als{
		\triangle^{ea}_2(t, \mathbf{U}) = & \left(1-\frac{1}{\ttk}\right)\left(n-2-\frac{1}{\ttk} \right)
		\frac{\ttb}{H^2}  p^a \tta m\nu^e+ \left(n-2 - \frac{1}{\ttk} \right)\frac{\tta}{ H } m\nu^e\nu^d  \tp{a}{d}  \nnb  \\
		&- \frac{\ttb}{\ttk H}  p^d \nu^e\tp{a}{d} - \left(\frac{1}{\ttk}-1\right)\frac{\ttb}{\ttk}\frac{\tta}{H^2} m \nu^e p^a + \frac{1}{\ttk H} \tta m \nu^b \nu^e \tp{a}{b} - \frac{\ttb}{\ttk H} \nu^e p^b \tp{a}{b}   \nnb \\
		& - (n-2)\left( \frac{1}{Ht} -\frac{1}{ \tan(Ht)}\right) \nu^c \nu^e \tp{a}{c} - (n-2)\left(\frac{1}{ \tan(Ht)} - \frac{1}{Ht}\right) \frac{\ttb}{\ttk H} \nu^e p^a  \nnb  \\
		& - (n-2) \left( \frac{1}{ (Ht)^2} - \frac{1}{ \tan^2(Ht)}\right)\ttb t \nu^e p^a + 2 \tta \ttb \frac{t^2}{\sin^2(Ht)} m \nu^e p^a  \nnb \\
		& + (n-2) \ttb t \nu^e p^a  - 2 P^{cd}\nu^e\nu_c\ts{\ulh}{^a_d} - 2 Q^{cd}\nu^e\nu_c\ts{\ulh}{^a_d} - \frac{2}{n-2}X^c X^d\nu^e \nu_c\ts{\ulh}{^a_d}  \notag  \\
		& + 2\Bigl(\nu_a\ts{\ulh}{^e_b}g^{bd}g^{a\ha}-\frac{1}{2(n-2)}\nu_a\ts{\ulh}{^e_b}g^{ab} g^{d\ha}\Bigr) g^{c\hc}(H_{\ha\hc}-E_{\ha}\nu_{\hc}+\nu_{\ha}E_{\hc}) \notag \\
		&\times (H_{dc}-E_{d}\nu_{c}+\nu_{d}E_{c}).
	}
	
\end{lemma}
\begin{proof}
	Making use of the definitions \eqref{E:V} and \eqref{E:VD}, and noting the identity
	\als{
		&g^{cd}\udn{c}\udn{d} \xi^a =  g^{cd}\udn{c}\tp{a}{d}  + \frac{\ttb}{\ttk H}\left(\frac{1}{\ttk}-1\right)\frac{1}{Ht}\lambda p^a + \frac{\ttb}{\ttk H} p^b \tp{a}{b} - \frac{1}{\ttk Ht}\lambda \nu^b \tp{a}{b} \nnb \\
		=   & g^{cd}\udn{c}\tp{a}{d} - \left(\frac{1}{\ttk}-1\right)\frac{\ttb}{\ttk}\frac{1}{H^2 t}  p^a + \frac{1}{\ttk Ht} \nu^b \tp{a}{b}  + \left(\frac{1}{\ttk}-1\right)\frac{\ttb}{\ttk}\frac{1}{H^2} \tta m  p^a - \frac{1}{\ttk H} \tta m \nu^b \tp{a}{b} \notag \\
		&+ \frac{\ttb}{\ttk H} p^b \tp{a}{b},
	}
	we can express the conformal Einstein equation \eqref{E:TSEQ}  for $\xi^e$ as 
	\als{
		g^{cd}\udn{c}\tp{a}{d}  ={}& \frac{n-2}{ \tan(Ht)}\nu^c \tp{a}{c} - \frac{n-2}{ \tan(Ht)} \frac{\ttb}{\ttk H}p^a + \frac{n-2}{ \tan^2(Ht)}\ttb t p^a + 2 \tta \ttb \frac{t^2}{\sin^2(Ht)} m  p^a  \nnb \\
		& - \frac{1}{\ttk Ht} \nu^b \tp{a}{b}+ \frac{\ttb}{\ttk H}\left(\frac{1}{\ttk}-1\right)\frac{1}{Ht} p^a - \left(\frac{1}{\ttk}-1\right)\frac{\ttb}{\ttk}\frac{\tta}{H^2}  m  p^a + \frac{\tta}{\ttk H} m \nu^b \tp{a}{b}  \nnb \\
		& - \frac{\ttb}{\ttk H} p^b \tp{a}{b} + (n-2) \ttb t p^a - 2 P^{cd}\nu_c\ts{\ulh}{^a_d} - 2 Q^{cd}\nu_c\ts{\ulh}{^a_d} - \frac{2}{n-2}X^c X^d\nu_c\ts{\ulh}{^a_d}  \notag  \\
		& + 2\Bigl(\nu_a\ts{\ulh}{^e_b}g^{bd}g^{a\ha}-\frac{1}{2(n-2)}\nu_a\ts{\ulh}{^e_b}g^{ab} g^{d\ha}\Bigr)  g^{c\hc}(H_{\ha\hc}-E_{\ha}\nu_{\hc}+\nu_{\ha}E_{\hc}) \notag \\
		&	\times (H_{dc}-E_{d}\nu_{c}+\nu_{d}E_{c}).
	}
	Multiplying this equation by $\nu^e$, we find, after rearranging, that
	\als{
		\nu^e g^{cd}\udn{c}\tp{a}{d}  ={} & \left(n-2-\frac{1}{\ttk} \right)\frac{1}{ Ht} \nu^e \nu^d \tp{a}{d}  - \left(\left(n-1-\frac{1}{\ttk}\right)\frac{1}{\ttk}-n+2\right)\frac{\ttb}{H^2 t} \nu^e p^a  + \triangle^{ea}_{21}
	}
	where
	\als{
		\triangle^{ea}_{21} = {}&\left(1-\frac{1}{\ttk}\right)\frac{\ttb}{\ttk}\frac{\tta}{H^2} m \nu^e p^a + \frac{\tta}{\ttk H} m \nu^b \nu^e \tp{a}{b} - \frac{\ttb}{\ttk H} \nu^e p^b \tp{a}{b}  \nnb \\
		&- (n-2)\left( \frac{1}{Ht} -\frac{1}{ \tan(Ht)}\right) \nu^c \nu^e \tp{a}{c} - (n-2)\left(\frac{1}{ \tan(Ht)} - \frac{1}{Ht}\right) \frac{\ttb}{\ttk H} \nu^e p^a  \nnb  \\
		&- (n-2) \left( \frac{1}{ (Ht)^2} - \frac{1}{ \tan^2(Ht)}\right)\ttb t \nu^e p^a + 2 \tta \ttb \frac{t^2}{\sin^2(Ht)} m \nu^e p^a  \nnb \\
		&+ (n-2) \ttb t \nu^e p^a - 2 P^{cd}\nu^e\nu_c\ts{\ulh}{^a_d} - 2 Q^{cd}\nu^e\nu_c\ts{\ulh}{^a_d} - \frac{2}{n-2}X^c X^d\nu^e \nu_c\ts{\ulh}{^a_d}  \notag  \\
		&+ 2\Bigl(\nu_a\ts{\ulh}{^e_b}g^{bd}g^{a\ha}-\frac{1}{2(n-2)}\nu_a\ts{\ulh}{^e_b}g^{ab} g^{d\ha}\Bigr)  g^{c\hc}(H_{\ha\hc}-E_{\ha}\nu_{\hc}+\nu_{\ha}E_{\hc}) \notag \\
		&\times (H_{dc}-E_{d}\nu_{c}+\nu_{d}E_{c}).
	}
	From the above expressions and the  identities \eqref{QCOMMUN2}, we then get
	\al{TSEQ3-3}{
		&Q^{edc} \udn{c}\tp{a}{d} - \left( \nu^d g^{ec} \udn{c}\tp{a}{d} -\nu^c g^{ed} \udn{c}\tp{a}{d} \right) \notag \\
		= &  \left( n-2 -\frac{1}{\ttk} \right)\frac{1}{Ht}  Q^{fgc}\nu_c(\ts{\delta}{^e_f}\ts{\delta}{^d_g}-\ts{\ulh}{^e_f}\ts{\ulh}{^d_g})\tp{a}{d}\nnb  \\
		& + \left( \left(n-1-\frac{1}{\ttk}\right)\frac{1}{\ttk}-n+2 \right)\frac{\ttb}{H^2 t}   Q^{edc}\nu_c\nu_d p^a   + \triangle^{ea}_{22}
	}
	where
	\als{
		\triangle^{ea}_{22} = {}& \left(n-2-\left(n-1-\frac{1}{\ttk}\right)\frac{1}{\ttk}\right)\frac{\tta\ttb}{H^2} m p^a \nu^e + \left(n-2- \frac{1}{\ttk} \right)\frac{\tta}{ H } m\nu^e\nu^d  \tp{a}{d} + \triangle^{ea}_{21}.
	}

	Next, we observe, with the help of \eqref{E:ID3-0}, that
	\als{
		& \nu^d g^{ec} \udn{c}\tp{a}{d} -\nu^c g^{ed} \udn{c}\tp{a}{d}=\nu^d g^{ec} (\udn{c}\tp{a}{d} - \udn{d}\tp{a}{c})  \nnb \\
		= {} & \nu^d g^{ec} \left( \udn{c} \udn{d} \xi^a - \udn{c}\frac{\ttb p^a}{\ttk H } \nu_d -\udn{d} \udn{c} \xi^a + \udn{d}\frac{\ttb p^a}{\ttk H} \nu_c \right)  \nnb \\
		= {} & \nu^d g^{ec} \left( \udn{c} \udn{d} \xi^a - \udn{d} \udn{c} \xi^a \right)  + \frac{\ttb}{\ttk H } g^{ec}  \udn{c}p^a  +  \frac{\ttb}{\ttk H} \nu_c \nu^d g^{ec} \udn{d}p^a   \nnb \\
		= {} & - \nu^d g^{ec} \ts{\udl{R}}{_{cdb}^a} \ttb t p^b  + \frac{1}{\ttk H t} g^{ec}  \ts{\ulh}{^d_c}\tp{a}{d}
		=  \frac{1}{\ttk H t} \tp{a}{d}h^{de} -\frac{1}{\ttk H} \ttb p^d \nu^e\tp{a}{d},
	}
	and that $\frac{1}{\ttk H t} \tp{a}{d}h^{de} = \frac{1}{\ttk H t} \tp{a}{d} Q^{abc}\nu_c \ts{\ulh}{^e_a} \ts{\ulh}{^d_b}$.
	Using these relations along with $n-2-\left(n-1-\frac{1}{\ttk}\right)\frac{1}{\ttk} = \left(1-\frac{1}{\ttk}\right)\left(n-2-\frac{1}{\ttk} \right)$ allows us  to express \eqref{E:TSEQ3-3} as
	\al{TSEQ3}{
		&Q^{edc} \udn{c}\tp{a}{d}   = \left( n-2 -\frac{1}{\ttk} \right)\frac{1}{Ht}  Q^{fgc}\nu_c(\ts{\delta}{^e_f}\ts{\delta}{^d_g}-\ts{\ulh}{^e_f}\ts{\ulh}{^d_g})\tp{a}{d}  \nnb  \\
		&\qquad + \frac{1}{\ttk H t} Q^{fgc}\nu_c \ts{\ulh}{^e_f} \ts{\ulh}{^d_g} \tp{a}{d} + \left(1-\frac{1}{\ttk}\right)\left(\frac{1}{\ttk}-n+2 \right)\frac{\ttb}{H^2 t}   Q^{edc}\nu_c\nu_d p^a   +\triangle^{ea}_2
	} 
	where
	\gat{
		\triangle^{ea}_2 =  -\frac{\ttb}{\ttk H}  p^d \nu^e\tp{a}{d} + \triangle^{ea}_{22}.\nnb
	} 
	
	On the other hand, by Lemma \ref{lem-identity}, we have that $Q^{ebc} \nu_b\nu_c \udn{e} = \lambda \nu^e \udn{e}$. Using this, we can write \eqref{E:ID3-0} as   
	\ali{
		\ttf Q^{ebc}\nu_b\nu_c\udn{e} p^a =\ttf Q^{ebc}\nu_b\nu_c \left(\frac{1}{\ttk}-1\right)\frac{1}{Ht} p^a \nu_e + \ttf Q^{ebc}\nu_b\nu_c \frac{1}{\ttb t} \tp{a}{e}   \label{E:ID3}
	}
	where $\ttf$ is a constant to be determined.  
	We then collect \eqref{E:TSEQ3} and \eqref{E:ID3} together to get the system
	\al{SYS2A}{
		\mathbf{A}_2^c\udn{c}\p{\tp{a}{d} \\ p^a }=\frac{1}{Ht} \mathcal{B}_2 \p{\tp{a}{d} \\ p^a}+\p{ \triangle^{ea}_2 \\ 0 }
	}
	where
	\gat{
		\mathbf{A}_2^c=\p{ Q^{edc} & 0   \\
			0 & \ttf Q^{cbe}\nu_b\nu_e }, \nnb
	}  
	and
	\als{
		&\mathcal{B}_2  = Q^{abc}\nu_c\p{
			\left(n-2-\frac{1}{\ttk}\right)(\ts{\delta}{^e_a}\ts{\delta}{^d_b}-\ts{\ulh}{^e_a}\ts{\ulh}{^d_b})+\frac{1}{\ttk} \ts{\ulh}{^e_a}\ts{\ulh}{^d_b} & \left(1-\frac{1}{\ttk}\right)\left(\frac{1}{\ttk}-n+2\right) \frac{\ttb}{H} \ts{\delta}{^e_a}\nu_b  \\
			\ttf \frac{H}{\ttb}\ts{\delta}{^d_a}\nu_b  &   \ttf \left(\frac{1}{\ttk}-1\right) \nu_b\nu_a
		}  \nnb \\
		&\scriptsize=\p{ Q^{abc} \nu_c & 0 \\ 0 &    \ttf Q^{abc}\nu_b\nu_c} \p{
			\left(n-2-\frac{1}{\ttk}\right)(\ts{\delta}{^e_a}\ts{\delta}{^d_b} -\ts{\ulh}{^e_a}\ts{\ulh}{^d_b})+\frac{1}{\ttk} \ts{\ulh}{^e_a}\ts{\ulh}{^d_b}  & \left(1-\frac{1}{\ttk}\right)\left(\frac{1}{\ttk}-n+2\right) \frac{\ttb}{H} \ts{\delta}{^e_a}\nu_b  \\
			\frac{H}{\ttb}\ts{\delta}{^d_a}   &   \left(\frac{1}{\ttk}-1\right) \nu_a
		}.
	}
	To complete the proof, we apply Lemma \ref{lem-model-symm} to \eqref{E:SYS2A} to get the system \eqref{E:SYS1a}.
\end{proof}

\begin{lemma}\label{lem-FOSHS-3}
	The conformal Einstein equation \eqref{E:SSEQTRFR3} for $\mathfrak{h}^{ab}- \ulh^{ab}$ can be expressed in first order form as
	\al{SYS1-3}{
		-\bar{ \mathbf{A}}_3^0\nu^c \nb_{c} \p{-\nu^e \tss{\ha\hb}{e} \\ \ts{\ulh}{^e_{\he}} \tss{\ha\hb}{e} \\ s^{\ha \hb} } + \bar{ \mathbf{A}}_3^c \ts{\ulh}{^{b}_{c}}  \udn{b} \p{-\nu^e \tss{\ha\hb}{e} \\ \ts{\ulh}{^e_{\he}} \tss{\ha\hb}{f} \\ s^{\ha \hb} } =\frac{1}{Ht} \bar{\mathcal{B}}_3\p{-\nu^e \tss{\ha\hb}{e} \\ \ts{\ulh}{^e_{\he}} \tss{\ha\hb}{e} \\ s^{\ha \hb} }+\bar{G}_3
	}
	where
	\als{
		\bar{ \mathbf{A}}_3^0 ={}& \p{-\lambda & 0 & 0\\ 0 & h^{f \he} & 0 \\ 0 & 0 & -\lambda}, \quad
		\bar{ \mathbf{A}}_3^c \ts{\ulh}{^{b}_{c}}  = \p{ -2\xi^b & -\ts{h}{^{\he b}} & 0 \\
			-\ts{h}{^{f b}} & 0 & 0 \\ 0 & 0 & 0}, \nnb \\
		\bar{\mathcal{B} }_3 ={}& \p{-\lambda\left(n-2\right) & 0 & 0 \\
			0 & 0 & 0 \\
			0 & 0 & 0}, \quad
		\bar{G}_3 = \p{ \nu_e \triangle^{e \ha \hb}_3(t, \mathbf{U}) \\ \ts{\ulh}{^f_e} \triangle^{e \ha \hb}_3(t, \mathbf{U})  \\  \lambda \nu^d \tss{\ha\hb}{d} }, \nnb
	}
	and the map $\triangle^{eab}_3(t, \mathbf{U})$, which is analytic for $(t,\mathbf{U})\in \bigl(-\iota,\frac{\pi}{H}\bigr)\times B_R(0)$ for $\iota,R>0$ small enough, is given by
	\als{
		&\triangle^{eab}_3(t, \mathbf{U}) \notag \\ = {}& \frac{n-2}{H}\tta m\nu^e\nu^d\tss{ab}{d} - \left(\frac{n-2}{Ht}-\frac{n-2}{\tan(Ht)}\right) \nu^c\nu^e \tss{ab}{c} + \nu^e g^{cd}\udn{c}(S^{-1}\ts{\mathcal{L}}{^{ab}_{gf}})\udn{d} h^{gf}  \nnb \\
		& - 2 \nu^e S^{-1}\ts{\mathcal{L}}{^{ab}_{gf}}\ts{\ulh}{^g_a}\ts{\ulh}{^f_b}P^{ab}  - 2 \nu^e S^{-1}\ts{\mathcal{L}}{^{ab}_{gf}}\ts{\ulh}{^g_a}\ts{\ulh}{^f_b}Q^{ab} \nnb\\
		& - \frac{2}{n-2} \nu^e S^{-1}\ts{\mathcal{L}}{^{ab}_{gf}}\ts{\ulh}{^g_a}\ts{\ulh}{^f_b}X^a X^b - 2(n-2) S^{-1} \nu^e \ts{\mathcal{L}}{^{ab}_{cd}} s^{c d}  \nnb \\
		& + 2 \nu^e S^{-1}\ts{\mathcal{L}}{^{ab}_{\hat{e}f}} \ts{\ulh}{^{\hat{e}}_a}\ts{\ulh}{^f_b}g^{bd}g^{a\ha}  g^{c\hc}(H_{\ha\hc}-E_{\ha}\nu_{\hc}+\nu_{\ha}E_{\hc})(H_{dc}-E_{d}\nu_{c}+\nu_{d}E_{c}).
	} 		
\end{lemma}
\begin{proof}
	Using the definitions \eqref{E:U}--\eqref{E:UD} for $\ts{s}{^{a b}_d}$ and $s^{ab}$, we can express the conformal Einstein equation \eqref{E:SSEQTRFR3} for $\mathfrak{h}^{ab}- \ulh^{ab}$ as
	\begin{align*}	
		& g^{cd}\udn{c}\tss{ab}{d} \notag \\
		={} &\frac{n-2}{\tan(Ht)}\nu^c \tss{ab}{c} +  g^{cd}\udn{c}(S^{-1}\ts{\mathcal{L}}{^{ab}_{ef}})\udn{d} h^{ef} - 2 S^{-1}\ts{\mathcal{L}}{^{ab}_{ef}}\ts{\ulh}{^e_a}\ts{\ulh}{^f_b}P^{ab}   \nnb\\
		&- 2 S^{-1}\ts{\mathcal{L}}{^{ab}_{ef}}\ts{\ulh}{^e_a}\ts{\ulh}{^f_b}Q^{ab} - \frac{2}{n-2} S^{-1}\ts{\mathcal{L}}{^{ab}_{ef}}\ts{\ulh}{^e_a}\ts{\ulh}{^f_b}X^a X^b - 2 (n-2) S^{-1} \ts{\mathcal{L}}{^{ab}_{cd}} s^{c d}  \notag  \\
		&+ 2 S^{-1}\ts{\mathcal{L}}{^{ab}_{ef}} \ts{\ulh}{^e_a}\ts{\ulh}{^f_b}g^{bd}g^{a\ha}  g^{c\hc}(H_{\ha\hc}-E_{\ha}\nu_{\hc}+\nu_{\ha}E_{\hc})(H_{dc}-E_{d}\nu_{c}+\nu_{d}E_{c}),
	\end{align*}
	which, in turn, implies by  \eqref{E:BTENSOR} that 
	\al{SSEQTRFR6a}{
		&Q^{edc}\udn{c}\tss{ab}{d} = \nu^d g^{ec}\udn{c}\tss{ab}{d} - \nu^c g^{ed} \udn{c} \tss{ab}{d}  + \frac{n-2}{Ht}\nu^d\nu^e \tss{ab}{d}  +\triangle^{eab}_{31}
	}
	where
	\als{
		\triangle^{eab}_{31} ={}&-\left(\frac{1}{Ht}-\frac{1}{\tan(Ht)}\right)(n-2)\nu^c\nu^e \tss{ab}{c} +  \nu^e g^{cd}\udn{c}(S^{-1}\ts{\mathcal{L}}{^{ab}_{gf}})\udn{d} h^{gf} \nnb  \\
		& -2 \nu^e S^{-1}\ts{\mathcal{L}}{^{ab}_{gf}}\ts{\ulh}{^g_a}\ts{\ulh}{^f_b}P^{ab} - 2 \nu^e S^{-1}\ts{\mathcal{L}}{^{ab}_{gf}}\ts{\ulh}{^g_a}\ts{\ulh}{^f_b}Q^{ab}  \nnb\\
		& -\frac{2}{n-2} \nu^e S^{-1}\ts{\mathcal{L}}{^{ab}_{gf}}\ts{\ulh}{^g_a}\ts{\ulh}{^f_b}X^a X^b - 2 (n-2) S^{-1} \nu^e \ts{\mathcal{L}}{^{ab}_{cd}} s^{c d}   \notag  \\
		& + 2 \nu^e S^{-1}\ts{\mathcal{L}}{^{ab}_{\hat{e}f}} \ts{\ulh}{^{\hat{e}}_a}\ts{\ulh}{^f_b}g^{bd}g^{a\ha}  g^{c\hc}(H_{\ha\hc}-E_{\ha}\nu_{\hc}+\nu_{\ha}E_{\hc})(H_{dc}-E_{d}\nu_{c}+\nu_{d}E_{c}).
	}
	However, due to \eqref{QCOMMUN2} and the commutator identity
	\als{
		\nu^d g^{ec}\udn{c}\tss{ab}{d}-\nu^c g^{ed} \udn{c} \tss{ab}{d}
		=  & \nu^d g^{ec}(\udn{c}\udn{d}s^{ab}- \udn{d} \udn{c}s^{ab})
		\notag \\
		=   & -\nu^d g^{ec} (\ts{\udl{R}}{_{cdf}^a} s^{fb} +\ts{\udl{R}}{_{cdf}^b} s^{fa} ) =0,
	}
	it follows that \eqref{E:SSEQTRFR6a} can be equivalently written as
	\al{SSEQTRFR6}{
		Q^{edc}\udn{c}\tss{ab}{d}
		= \frac{n-2}{Ht}Q^{fgc}\nu_c(\ts{\delta}{^e_f}\ts{\delta}{^d_g}-\ts{\ulh}{^e_f}\ts{\ulh}{^d_g}) \tss{ab}{d}+ \triangle^{eab}_3				}
	where
	\als{
		\triangle^{eab}_3 = {} & \frac{n-2}{H}\tta m\nu^e\nu^d\tss{ab}{d} +  \triangle^{eab}_{31}.
	} 
	
	On the other hand, by Lemma \ref{lem-identity}, we have that  $Q^{ebc} \nu_b\nu_c \udn{e} = \lambda \nu^e \udn{e}$. Using this, we see from \eqref{E:ID4-0} that
	\be\label{eq-sab}
	Q^{ebc}\nu_b\nu_c\udn{e}s^{fg}=Q^{ebc}\nu_b\nu_c \tss{fg}{e}.
	\ee		
	Collecting \eqref{E:SSEQTRFR6} and \eqref{eq-sab} together gives 
	\al{SYS1-3a}{
		\mathbf{A}_3^c\udn{c}\p{\tss{\ha\hb}{d} \\ s^{\ha \hb} }=\frac{1}{Ht} \mathcal{B}_3 \p{\tss{\ha\hb}{d} \\ s^{\ha \hb}}+ \p{\triangle^{e \ha \hb}_3 \\ Q^{ebc}\nu_b\nu_c \tss{\ha \hb}{e} }
	}
	where
	\gas{
		\mathbf{A}_3^c=  \p{ Q^{edc} & 0  \\
			0 & Q^{cbe}\nu_b\nu_e }
		\intertext{and}
		\mathcal{B}_3  =   Q^{abc}\nu_c \p{
			\left(n-2\right)(\ts{\delta}{^e_a}\ts{\delta}{^d_b}-\ts{\ulh}{^e_a}\ts{\ulh}{^d_b})  & 0 \\
			0  &   0}
		=   \p{ - \lambda \left(n-2\right) \nu^e \nu^d  & 0 \\
			0  &   0}.  		
	}
	Then by applying Lemma \ref{lem-model-symm} to \eqref{E:SYS1-3a}, we get
	\gat{
		\bar{ \mathbf{A}}_3^c\udn{c} \p{-\nu^e \tss{\ha\hb}{e} \\ \ts{\ulh}{^e_{\he}} \tss{\ha\hb}{e} \\ s^{\ha \hb} }=\frac{1}{Ht}\bar{\mathcal{B}}_3  \p{-\nu^e \tss{\ha\hb}{e} \\ \ts{\ulh}{^e_{\he}} \tss{\ha\hb}{e} \\ s^{\ha \hb} } +\bar{G}_3\nnb
	}
	where
	\gat{
		\bar{ \mathbf{A}}_3^c = \p{ \nu_e Q^{edc}\nu_d &  \nu_e Q^{edc}\ts{\ulh}{^{\he}_d} & 0 \\
			\ts{\ulh}{^f_e}Q^{edc}\nu_d  &  \ts{\ulh}{^f_e}Q^{edc}\ts{\ulh}{^{\he}_d} & 0 \\ 0 & 0 & Q^{cbe}\nu_b\nu_e}, \quad
		\bar{G}_3 = \p{\nu_e & 0 \\ \ts{\ulh}{^f_e} & 0 \\ 0 & 1 } \p{ \triangle^{e \ha \hb}_3 \\ Q^{ebc}\nu_b\nu_c \tss{\ha \hb}{e} }, \nnb
	}
	and
	\gat{
		\bar{\mathcal{B} }_3= \p{\nu_e & 0 \\ \ts{\ulh}{^f_e} & 0 \\ 0 & 1 } \mathcal{B}_3  \p{\nu_d & \ts{\ulh}{^{\he}_d} & 0 \\0 & 0 & 1} = \p{ - \lambda \left(n-2\right) & 0 & 0 \\
			0 & 0 & 0 \\
			0 & 0 & 0}. \nnb
	}
	We then complete the proof by setting $\bar{ \mathbf{A}}_3^0=  \bar{ \mathbf{A}}_3^c \nu_c $, which allows us to write in the stated form \eqref{E:SYS1-3}.
\end{proof}

\begin{lemma}\label{lem-FOSHS-4}
	The conformal Einstein equation \eqref{E:SSEQ2}  for $q$ can be expressed in first order form as
	\als{
		-\bar{\mathbf{A}}_4^0\nu^c \nb_{c} \p{-\nu^e s_{e} \\ \ts{\ulh}{^e_{\he}}s_{e} \\ s } + \bar{\mathbf{A}}_4^c \ts{\ulh}{^{b}_{c}}  \udn{b} \p{-\nu^e s_{e} \\ \ts{\ulh}{^e_{\he}}s_{e} \\ s } =\frac{1}{Ht} \bar{\mathcal{B}}_4 \p{-\nu^e s_{e} \\ \ts{\ulh}{^e_{\he}}s_{e} \\ s } +\bar{G}_4
	}
	where
	\als{
		\bar{\mathbf{A}}_4^0 ={}& \p{-\lambda & 0 & 0\\ 0 & h^{f \he} & 0 \\ 0 & 0 & -\lambda}, \quad
		\bar{\mathbf{A}}_4^c \ts{\ulh}{^{b}_{c}}  =\p{ -2\xi^b & -\ts{h}{^{\he b}} & 0 \\
			-\ts{h}{^{f b}} & 0 & 0 \\ 0 & 0 & 0}, \nnb \\
		\bar{\mathcal{B} }_4 ={}& \p{-\lambda\left(n-2\right) & 0 & 0 \\
			0 & 0 & 0 \\
			0 & 0 & 0}, \quad
		\bar{G}_4 =\p{ \nu_e\triangle^{e}_4 (t, \mathbf{U}) \\ \ts{\ulh}{^f_e}\triangle^{e}_4 (t, \mathbf{U}) \\  \lambda \nu^d s_d }, \nnb
	}
	and the map $\triangle^e_4(t, \mathbf{U})$, which is analytic for $(t,\mathbf{U})\in \bigl(-\iota,\frac{\pi}{H}\bigr)\times B_R(0)$ for $\iota,R>0$ small enough, is given by
	\als{
		\triangle^e_4(t, \mathbf{U}) = {}& \frac{(n-2)}{H} \tta m\nu^e\nu^d  s_d - 2 P^{ab}\nu^e \nu_a\nu_b - 2 Q^{ab}\nu^e\nu_a\nu_b - \frac{2}{n-2}X^a X^b\nu^e \nu_a\nu_b   \nnb  \\
		& - \left(\frac{1}{Ht} -\frac{1}{ \tan(Ht)} \right) (n-2) \nu^c \nu^e s_c + 2 \nu^e \tta^2  \frac{t^2}{\sin^2(Ht)} m^2 - 2 (n-2) \nu^e \tta t m \nnb  \\
		& + \frac{3-n}{ (n-1)} \nu^e g^{cd}\udn{c}h_{gf} \udn{d} h^{gf} - \frac{2(3-n)}{n-1} \nu^e h_{ab} P^{ab} - \frac{2(3-n)}{n-1} \nu^e h_{ab} Q^{ab} \nnb\\
		& - \frac{2(3-n)}{(n-1)(n-2)}\nu^e h_{ab} X^a X^b + 2  \nu^e  \Bigl(\frac{3-n}{n-1} h_{ab}  g^{bd}g^{a\ha}-\frac{3-n+\lambda}{2(n-2)}   g^{d\ha}  \notag  \\
		& - \nu_a\nu_bg^{bd}g^{a\ha}\Bigr) g^{c\hc}(H_{\ha\hc}-E_{\ha}\nu_{\hc}+\nu_{\ha}E_{\hc})(H_{dc}-E_{d}\nu_{c}+\nu_{d}E_{c}).
	}
\end{lemma}

\begin{proof}
	Using the definitions \eqref{E:W} and \eqref{E:Q}--\eqref{E:QD}, we can express the conformal Einstein equation \eqref{E:SSEQ2} for $q$ as
	\als{
		g^{cd}\udn{c}s_d   ={}&\frac{(n-2)}{ H t} \nu^c s_c + \triangle_4
	}
	where
	\als{
		\triangle_4 = {} & - \left(\frac{1}{Ht} -\frac{1}{ \tan(Ht)} \right) (n-2) \nu^c s_c + \frac{2 \tta^2 t^2}{\sin^2(Ht)} m^2 - 2 (n-2) \tta t m \nnb  \\
		& + \frac{3-n}{ (n-1)} g^{cd}\udn{c}h_{ef} \udn{d} h^{ef} - \frac{2(3-n)}{n-1} h_{ab} P^{ab} - \frac{2(3-n)}{n-1} h_{ab} Q^{ab} \nnb\\
		& - \frac{2(3-n)}{(n-1)(n-2)}h_{ab} X^a X^b
		- 2 P^{ab}\nu_a\nu_b - 2 Q^{ab}\nu_a\nu_b - \frac{2}{n-2}X^a X^b\nu_a\nu_b   \notag  \\
		& + 2\Bigl(\frac{3-n}{n-1} h_{ab}  g^{bd}g^{a\ha}-\frac{3-n+\lambda}{2(n-2)}   g^{d\ha} +\nu_a\nu_bg^{bd}g^{a\ha}\Bigr) \nnb \\
		&\qquad \qquad \cdot g^{c\hc}(H_{\ha\hc}-E_{\ha}\nu_{\hc}+\nu_{\ha}E_{\hc})(H_{dc}-E_{d}\nu_{c}+\nu_{d}E_{c}). \nnb
	}
	With the help of \eqref{QCOMMUN2} and the commutator identity
	\[	\nu^d g^{ec} \udn{c}s_d -\nu^c g^{ed} \udn{c}s_d
	=\nu^d g^{ec} ( \udn{c}\udn{d}s  -  \udn{d}\udn{c}s  )=0, \] 
	it then follows that
	\als{
		Q^{edc} \udn{c} s_d - \left( \nu^d g^{ec} \udn{c}s_d -\nu^c g^{ed} \udn{c}s_d \right)
		=  \frac{(n-2)}{H}  \frac{1}{t}Q^{fgc}\nu_c(\ts{\delta}{^e_f}\ts{\delta}{^d_g}-\ts{\ulh}{^e_f}\ts{\ulh}{^d_g})s_d+\triangle^e_4
	}
	where
	\als{
		\triangle^e_4 ={} & \frac{(n-2)}{H} \tta m\nu^e\nu^d  s_d + \nu^e \triangle_4.
	}

	On the other hand, by Lemma \ref{lem-identity}, we have that $Q^{ebc} \nu_b\nu_c \udn{e} = \lambda \nu^e \udn{e}$. Using this, we see from \eqref{E:ID4-0} that
	\gat{
		Q^{ebc}\nu_b\nu_c \udn{e}s= Q^{ebc}\nu_b\nu_c s_e.  \nnb 
	}		
	Collecting together the above two equations gives
	\gat{
		\mathbf{A}_4^c\udn{c}\p{s_d \\ s  }=\frac{1}{Ht}\mathcal{B}_4 \p{s_d\\ s }+\p{ \triangle^{e}_4 \\  Q^{ebc}\nu_b\nu_c s_e}, \nnb
	}
	where
	\gas{
		\mathbf{A}_4^c= \p{ Q^{edc} & 0   \\
			0 &  Q^{cbe}\nu_b\nu_e }  
		\intertext{and}
		\mathcal{B}_4 =   Q^{abc}\nu_c\p{
			\left(n-2\right)(\ts{\delta}{^e_a}\ts{\delta}{^d_b}-\ts{\ulh}{^e_a}\ts{\ulh}{^d_b})  & 0 \\
			0  &   0}
		=  \p{ - \lambda \left(n-2\right) \nu^e \nu^d  & 0 \\
			0  &   0}.
	}
	To complete the proof, we can then proceed in the same way as in final step of the proof of Lemma \ref{lem-FOSHS-4}. 		
\end{proof}

\section{Expansions and inequalities}\label{sec-matrix}
\subsection{Expansions}
$\;$

\bigskip

\noindent We recall the well-known Neumann series expansion.
\begin{lemma} \label{E:EXPANSIONOFINVERSE2}
	If $A$ and $B$ are $n\times n$ matrices with $A$ invertible,  then there exists an $\epsilon_0>0$ such that the map
	\begin{equation*}
		(-\epsilon_0,\epsilon_0) \ni \epsilon \longmapsto (A+\epsilon B)^{-1} \in \mathbb{M}_{n\times n}
	\end{equation*}
	is analytic and admits the series representation
	\begin{align*}
		(A+\epsilon B)^{-1}=A^{-1}+\sum_{n=1}^\infty (-1)^n\epsilon^n (A^{-1}B)^nA^{-1}, \quad  |\epsilon|< \epsilon_0.
	\end{align*}
\end{lemma}

\noindent In the following proposition, we use the Neumann series expansion to derive some useful geometric expansion formulae. All of the geometric objects are as defined in \S\ref{s:ds} and \S\ref{sec-Fuchsian-field} .
\begin{proposition}\label{prop-analy}
	There exists a constant $\epsilon_0>0$ such that if $|g^{ab}-\ulg^{ab}|<\epsilon_0$, then
	\begin{align}
		g_{ab} =&{} \ulg_{ab}+\mathcal{S}_{ab}(t,m,p^d,s,s^{\ha\hb}), \label{eq-g-analy}\\
		\nb_c h^{\ha\hd} =&{} \tensor{\mathcal{S}}{^{\ha\hd}_{c}} (t, m, m_d, s, s_d, s^{a b},\tss{a b}{d}),  \label{eq-g-analy2}
		\intertext{and}
		\nb_c g^{a b}=&{} \tensor{\mathcal{S}}{^{\ha\hd}_{c}} (t, m, m_d, s, s_d, p^a, \tensor{p}{^a_d}, s^{a b}, \tss{a b}{d} ), \label{eq-g-analy3}
	\end{align}
	where the $\mathcal{S}$-maps are analytic in all their variables and satisfy $\mathcal{S} (t,0)=0$.
\end{proposition}

\begin{proof} 
	By \eqref{E:W}, \eqref{E:Q} and \eqref{E:q}, we have
	\begin{equation*}
		S=\exp\Bigl(\frac{q-(\lambda+1)}{3-n}\Bigr)=\exp\Bigl(\frac{s-\tta tm}{3-n}\Bigr),
	\end{equation*}
	and hence, after differentiating, we get
	\begin{align*}
		\nb_c S={}&S\nb_c\ln S=S\nb_c\Bigl(\frac{q-(\lambda+1)}{3-n}\Bigr) \\  
		={}&\frac{1}{3-n}\Bigl(s_c -m_c-\frac{\tta}{\ttj H} m\nu_c\Bigr)\exp\Bigl(\frac{s-\tta tm}{3-n}\Bigr).
	\end{align*}	
	With the help of these relations, we then observe that
	\begin{align*}
		\nb_c h^{\ha\hd} =	{} & \nb_c(S\mathfrak{h}^{\ha\hd} )=S\nb_c \mathfrak{h}^{\ha\hd}+ \mathfrak{h}^{\ha\hd}\nb_c S  \notag  \\
		= {} &\frac{1}{3-n}\Bigl(s_c -m_c-\frac{\tta}{\ttj H} m\nu_c\Bigr)\exp\Bigl(\frac{s-\tta tm}{3-n}\Bigr)s^{\ha\hd} \notag \\
		& + \frac{1}{3-n}\Bigl(s_c -m_c-\frac{\tta}{\ttj H} m\nu_c\Bigr)\exp\Bigl(\frac{s-\tta tm}{3-n}\Bigr)\ulh^{\ha\hd} \notag  \\
		& +  \exp\Bigl(\frac{s-\tta tm}{3-n}\Bigr)\tss{\ha\hd}{c}.  
	\end{align*}
	Letting $\tensor{\mathcal{S}}{^{\ha\hd}_{c}} (t, m, m_d, s, s_d, s^{a b},\tss{a b}{d})$ 
	denote the right hand side of the above expression, it is clear that $\tensor{\mathcal{S}}{^{\ha\hd}_{c}}$ is analytic in all its variables and satisfies $\tensor{\mathcal{S}}{^{\ha\hd}_{c}}(t,0)=0$, which establishes the validity of \eqref{eq-g-analy2}.
	
	Next, we observe that
	\begin{align*}
		g^{cd}-\ulg^{cd}={}&(\lambda+1)\nu^c\nu^d-\xi^d\nu^c-\xi^c\nu^d+(S\mathfrak{h}^{cd}-\ulh^{cd}) \notag  \\
		={}&\tta tm\nu^c\nu^d-\ttb t p^d\nu^c-\ttb t p^c\nu^d \\
		&+\exp\Bigl(\frac{s-\tta tm}{3-n}\Bigr) s^{cd}+\Bigl[\exp\Bigl(\frac{s-\tta tm}{3-n}\Bigr)-1\Bigr]\ulh^{cd}.
	\end{align*}
	With the help of this expression,  \eqref{eq-g-analy} follows from applying Lemma \ref{E:EXPANSIONOFINVERSE2} to
	$g_{ab}= (g^{cd})^{-1}=\bigl[\ulg^{cd}+(g^{cd}-\ulg^{cd})\bigr]^{-1}$.
	
	Finally, to verify the remaining expression \eqref{eq-g-analy3}, we differentiate $g^{ab}=  h^{ab}- 2\nu^{(a}  \xi^{b)} +\lambda \nu^a\nu^b $ to get
	\begin{align*}
		\nb_c g^{ab}={}&\nb_c h^{ab}-2\nu^a\nb_c \xi^b+\nu^a\nu^b\nb_c\lambda \\
		={}&\nb_c h^{ab}- 2 \nu^{(a} \Bigl(\tensor{p}{^{b)}_c}+\frac{\ttb}{\ttk H}p^{b)} \nu_c\Bigr)+\nu^a\nu^b\Bigl(m_c+\frac{\tta}{\ttj H} m \nu_c\Bigr),
	\end{align*}
	where in deriving this we have used \eqref{decom-g}, \eqref{E:W}, \eqref{E:V} and \eqref{E:VD}. Letting \\ $\tensor{\mathcal{S}}{^{\ha\hd}_{c}} (t, m, m_d, s, s_d, p^a, \tensor{p}{^a_d}, s^{a b}, \tss{a b}{d} )$ 
	denote the right hand side of the above expression, it is clear that $\tensor{\mathcal{S}}{^{\ha\hd}_{c}}$ is analytic in all its variables and satisfies $\tensor{\mathcal{S}}{^{\ha\hd}_{c}}(t,0)=0$, which establishes the validity of \eqref{eq-g-analy3} and  completes the proof. 	
\end{proof}

\subsection{Young's inequality}\label{s:young}	
\begin{lemma}[Young's inequality for scalars]\label{t:young}
	Suppose $a,b\in \mathbb{R}$ and $\epsilon>0$. Then
	\begin{equation*}
		|ab|\leq \frac{\epsilon }{2}a^2+\frac{1}{2\epsilon}b^2.
	\end{equation*}
\end{lemma}

\begin{lemma}[Young type inequalities for tensors]\label{t:young2}
	Suppose $\ulh_{ab}\in T_2^{0} \Sigma $ is a positive definite metric and $\ulh^{ab}\in T^2_{0} \Sigma $ is the inverse of $\ulh_{cd}$, $X_{dc}\in T^{0}_2 \Sigma $ is anti-symmetric, $Y^a \in T \Sigma $, $Z_c\in T^\ast \Sigma$, $A^d\in T \Sigma$, $A^{abc}\in T^{3}_0 \Sigma $, $\tensor{A}{^d_c} \in T^{1}_1 \Sigma $  and $\delta>0$. Then
	\begin{align}
		| X_{dc}A^d Y^c | \leq & \frac{\delta }{2}\ulh^{ad}\ulh^{bc}X_{ab} X_{dc}+\frac{1}{2\delta}\ulh_{bd} \tensor{A}{^d} A^b \ulh_{ac}Y^aY^c, \label{e:chy1} \\
		| X_{dc}A^{adc} Z_a | \leq & \frac{\delta }{2}\ulh^{ad}\ulh^{bc}X_{ab} X_{dc}+\frac{1}{2\delta} A^{abe}A^{fdc} \ulh_{bd} \ulh_{ec} Z_aZ_f, \label{e:chy2} 
		\intertext{and}
		| Z_{d} \tensor{A}{^d_c} Y^c | \leq & \frac{\delta }{2}\ulh^{ad}Z_{d} Z_{a}+\frac{1}{2\delta}\ulh_{ad} \tensor{A}{^a_e} \tensor{A}{^d_b } Y^eY^b. \label{e:chy3}
	\end{align} 	
\end{lemma}
\begin{proof}  
	The inequality \eqref{e:chy1} is a direct consequence of the calculation
	\begin{align*}
		& \delta \ulh^{da}\ulh^{bc} X_{ab} X_{dc}\pm 2A^d Y^cX_{dc}+ \frac{1}{\delta} A^d Y^c \ulh_{bd}\ulh_{ac}A^b Y^a  \notag   \\
		= & \Bigl(\sqrt{\delta}\ulh^{da}\ulh^{bc} X_{ab} \pm \frac{1}{\sqrt{\delta}} A^d Y^c\Bigr)\Bigl(\sqrt{\delta}X_{dc} \pm \frac{1}{\sqrt{\delta}} \ulh_{\hb d}\ulh_{\ha c}A^{\hb}Y^{\ha}\Bigr) \notag   \\
		= &  \ulh^{da}\ulh^{bc} \Bigl(\sqrt{\delta}X_{ab} \pm \frac{1}{\sqrt{\delta}}\ulh_{bf}\ulh_{ae}A^eY^f\Bigr)\Bigl(\sqrt{\delta} X_{dc} \pm \frac{1}{\sqrt{\delta}}\ulh_{\hb d}\ulh_{\ha c}A^{\hb} Y^{\ha}\Bigr)\geq 0,
	\end{align*}
	where the inequality in the last line is due to positive definiteness of the metric $\ulh_{ab}$.
	We similarly conclude the validity of the inequalities \eqref{e:chy2}--\eqref{e:chy3} from the following related calculations:
	\begin{align*}
		0\leq& \ulh^{ad}\ulh^{bc}\Bigl(\sqrt{\delta}X_{ab} \pm \frac{1}{\sqrt{\delta}} \ulh_{\ha a}\ulh_{\hb b} A^{e\ha\hb}Z_e\Bigr)\Bigl(\sqrt{\delta}X_{dc} \pm \frac{1}{\sqrt{\delta}} \ulh_{\hd d}\ulh_{\hc c}A^{f\hd\hc}Z_f\Bigr) \notag  \\
		=&\delta \ulh^{ad} \ulh^{bc}X_{ab}X_{dc} \pm 2X_{ab} A^{fab} Z_f+ \frac{1}{\delta} A^{edc}Z_e \ulh_{\hd d} \ulh_{\hc c} A^{f\hd\hc} Z_f
	\end{align*}
	and
	\begin{align*}
		0\leq& \ulh^{dc}\bigl(\sqrt{\delta}Z_d\pm \frac{1}{\sqrt{\delta}}\ulh_{ad} \tensor{A}{^a_e} Y^e\Bigr)\Bigl(\sqrt{\delta}Z_c \pm \frac{1}{\sqrt{\delta}}\ulh_{bc} \tensor{A}{^b_f} Y^f\Bigr) \\
		&=  \delta Z_d Z_c \ulh^{dc} \pm 2 Z_aY^f \tensor{A}{^a_f} + \frac{1}{\delta} \tensor{A}{^a_e} Y^e \ulh_{ba} \tensor{A}{^b_f} Y^f.
	\end{align*}
\end{proof}

\section{Index of notation} \label{s:index}

\bigskip

\begin{longtable}{ll}
	$(\widetilde{\mathcal{M}}^n, \, \tilde g_{a b})$ & $n$-dimensional, connected Lorentzian  manifold; \S\ref{S:INTRO}  \\
	$G$ & compact and connected Lie group;  \S \ref{S:INTRO} \\
	$\mathcal{G}$ & Lie algebra of $G$;  \S \ref{S:INTRO} \\
	$\phi \cdot \psi$ & inner-product for $\phi,\psi \in \mathcal{G}$;  \S \ref{S:INTRO} \\
	$|\cdot|$ & norm on $ \mathcal{G}$;  \S \ref{S:INTRO} \\
	$\tilde{g}_{ab}$ & physical spacetime metric; \S \ref{S:INTRO} \\
	$\tilde F_{ab} $ & physical  Yang–Mills curvature; \S \ref{S:INTRO} \\
	$\Ab_a$ & physical gauge potential; \S \ref{S:INTRO} \\
	$\tilde \nabla_a$ &  covariant derivative associated to $\tilde{g}_{a b}$;  \S\ref{S:INTRO}\\
	$\tilde D_a $ &   gauge covariant derivative of a $\mathcal{G}$-valued tensor;  \S\ref{S:INTRO}\\
	$\tau$ & Cartesian time coordinate function on $\Rbb$; \S \ref{s:ds} \\
	$\ulh_{ab}$ &  standard metric on $\mathbb{S}^{n-1}$; \S \ref{s:ds}\\
	$\gd(x)$ & Gudermannian function;  \S \ref{s:ds} \\
	$t$ & conformal compactified time function; \S \ref{s:ds}, eqn. \eqref{E:COOR1} \\
	$\tilde{\ulg}_{ab}$ &  de Sitter metric; \S \ref{s:ds}, eqn. \eqref{E:DESMETR}\\
	$\Psi$ &  conformal factor; \S \ref{s:ds}, eqn. \eqref{E:CONFFAC}\\
	$\udl{g}_{ab}$  & conformal de Sitter metric; \S \ref{s:ds}, eqn. \eqref{e:cfds}\\
	$\nu_a $, $\nu^a$ & unit normal and conormal to $\Sigma_t $ w.r.t $\ulg_{ab}$; \S \ref{s:ds}, eqn. \eqref{E:NORMT}  \\	
	$\tensor{\ulh}{^a_b}$ &  projection w.r.t $\nu^a$ and $\ulg_{ab}$; \S \ref{s:ds}, eqn. \eqref{E:PRO} \\
	$\tilde{\nu}_a$, $\tilde{\nu}^a$ & unit normal and conormal to  $\Sigma_{\tau}$ w.r.t.   $\tilde{\ulg}_{ab}$; \S \ref{s:ds}, eqn. \eqref{E:NORMT-a}  \\	
	$\tilde{T}_a$, $\tilde{T}^a$ & unit normal and conormal to  $\Sigma_{\tau}$ w.r.t.   $\tilde{g}_{ab}$; \S \ref{s:ds}, eqn. \eqref{def-T}  \\	
	$\tilde{\lambda} $  &  $\tilde{g}^{ab} \tilde{\nu}_a \tilde{\nu}_b$; \S \ref{s:ds}, eqn. \eqref{def-T} \\
	$\tensor{\tilde{\ulh}}{^c_d} $, $\tensor{\tilde{h}}{^c_d}$  &  projections w.r.t $\tilde{\nu}^a$ and $\tilde{T}^a$, respectively; \S \ref{s:ds}, eqn. \eqref{e:def-h2}--\eqref{e:def-h2-phy} \\
	$\tilde{E}_b$, $ \tilde{H}_{db}$ &  $3+1$ decomposition of  $\tilde{F}_{a p}$ w.r.t. $\tensor{\tilde{h}}{^a_b}$ and $\tilde{T}^p $; \S \ref{s:ds}, eqn. \eqref{Et-Ht-def} \\
	$\tilde{Z}^{a} $  & physical wave gauge vector fields; \S \ref{sec:gauge-conditions} \\
	$\tilde{X}^a $  & physical contracted Christoffel symbols; \S \ref{sec:gauge-conditions} \\
	$\tilde{\nb} $  & covariant derivative associated to $\tilde{\ulg}_{ab}$; \S \ref{sec:gauge-conditions} \\
	$\tilde{h}^{ab}$ & $\tensor{\tilde{\ulh}}{^a_c}\tensor{\tilde{\ulh}}{^b_d}\tilde{g}^{cd}$ ; \S \ref{s:dt0} \\
	$g_{ab}$, $g^{ab}$ & conformal spacetime metric; \S \ref{sec-Fuchsian-field}, eqn. \eqref{CONFG-g} \\
	$F_{ab} $ & conformal Yang–Mills curvature; \S \ref{sec-Fuchsian-field}, eqn. \eqref{CONFG-F}  \\
	$A_a$ & conformal gauge potential; \S \ref{sec-Fuchsian-field}, eqn. \eqref{CONFG-A}  \\	
	$\lambda$, $\xi^c$, $h^{ab}$ & $3+1$ decompositions of $g^{ab}$ w.r.t. $\nu_a$ and $\tensor{\ulh}{^c_b}$; \S \ref{sec-Fuchsian-field}, eqn. \eqref{decom-g}  \\	
	$m,  m_d$ & Fuchsian fields; \S \ref{sec-Fuchsian-field}, eqn. \eqref{E:W}--\eqref{E:QD}  \\
	$p^a,  \tp{a}{d}$ & Fuchsian fields; \S \ref{sec-Fuchsian-field}, eqn. \eqref{E:W}--\eqref{E:QD}  \\	
	$s^{ab}, \tss{ab}{d}, s, s_d$ & Fuchsian fields; \S \ref{sec-Fuchsian-field}, eqn. \eqref{E:W}--\eqref{E:QD}  \\	
	$\tta, \ttb, \ttj, \ttk$ &  dimension dependent constants; \S \ref{sec-Fuchsian-field} \\
	$q, \mathfrak{h}^{ab}$ & modified metric variables; \S \ref{sec-Fuchsian-field}, eqn. \eqref{E:q} \\
	$\ulh^{ab}$ & $ \ulh^{a}{}_{c}\ulh^{b}{}_{d}\ulg^{cd}$ ; \S \ref{sec-Fuchsian-field}, eqn. \eqref{e:S} \\ 	
	$\alpha,\ula $ &  determinant variables;  \S \ref{sec-Fuchsian-field}, eqn. \eqref{e:S} \\
	$\bar A_b, E_b$ & $A_a \tensor{\ulh}{^a_b} $ and $- \nu^p F_{p a} \tensor{\ulh}{^a_b}$, respectively;  \S \ref{sec-Fuchsian-field}, eqn. \eqref{decom-F}  \\
	$\tE^a, H_{db}$ & $-h^{a b} E_{b}$ and $\tensor{\ulh}{^c_d} F_{c a} \tensor{\ulh}{^a_b}$, respectively;  \S \ref{sec-Fuchsian-field}, eqn. \eqref{def-tE-1}, \eqref{def-MYM}  \\
	$\mathbf{U}$ & Fuchsian field vector;  \S \ref{sec-Fuchsian-field}, eqn. \eqref{def-U} \\
	$\norm{\cdot}_{W^{k,p}} $ & Sobolev norms; \S\ref{s:norm} \\
	$ \norm{\cdot}_{L^\infty(I,W^{s,p})}$ & Sobolev norms; \S\ref{s:norm} \\
	$\nb$ & covariant derivatives associated to   $\underline{g}_{ab}$; \S\ref{sec-Einstein}  \\
	$\nabla$ &  covariant derivative associated to  $g_{ab}$; \S\ref{sec-Einstein} \\
	$\Box $ &  wave operator associated to $g_{a b}$; \S\ref{sec-Einstein} \\
	$\tensor{\tilde{\udl{R}}}{_{cde}^a},\tilde{\udl{R}}_{ab},\tilde{\udl{R}}$ & Riemannian/Ricci/scalar curvature tensors of $\tilde{\underline{g}}_{ab}$; \S\ref{sec-Einstein} \\ 
	$ \tensor{\udl{R}}{_{cde}^a},\udl{R}_{ab},\udl{R}$ & Riemannian/Ricci/scalar curvature tensors of  $\underline{g}_{ab}$; \S\ref{sec-Einstein} \\ 
	$\tensor{\tilde{R}}{_{cde}^a},\tilde{R}_{ab},\tilde{R}$  &  Riemannian/Ricci/scalar curvature tensors of  $\tilde{g}_{ab}$; \S\ref{sec-Einstein} \\
	$\tensor{R}{_{cde}^a},R_{ab},R$ &  Riemannian/Ricci/scalar curvature tensors of  $g_{ab}$; \S\ref{sec-Einstein} \\
	$Z^a$ & conformal wave gauge vector field;  \S \ref{s:REE}, eqn. \eqref{E:WAVEGA}\\
	$X^a$ & conformal contracted Christoffel symbols;  \S \ref{s:REE}, eqn. \eqref{def-X}\\
	$Y^a,\eta^a$ & gauge source vector field;  \S \ref{s:REE}, eqn. \eqref{E:XYZ}\\ 
	$\ts{A}{^{ab}_c}$ &   wave gauge tensor;  \S \ref{s:REE}\\
	$P^{ab}, Q^{ab}$  &   lower order terms in $R^{ab}$;  \S \ref{s:REE}, \S\ref{s:App2}, eqs. \eqref{E:CONFEIN5}, \eqref{E:RICCI}  \\
	$\ts{\mathcal{L}}{^{ab}_{cd}} $  &   projection operator;  \S \ref{s:REE}, eqs. \eqref{E:LG}   \\
	$Q^{edc}$ & symmetrizing tensor; \S\ref{s:REE}, eqn \eqref{E:BTENSOR} \\
	$\tte, \ttf$  &  dimension dependent constants; \S \ref{s:fucein}\\
	$ \widehat{\mathbf{U}}$  &  $3+1$ decomposition of $\mathbf{U}$; \S\ref{s:verif}, eqn. \eqref{e:hatu}\\
	$\mathring{\mathbf{V}}, \p{\mathring{\tE}^e, \mathring{E}_{d},\mathring{H}_{a b},\mathring{A}_s}$ & renormalized conformal Yang--Mills fields; \S\ref{S:YMipv}, eqn. \eqref{e:newYM}  \\
	$\mathbf{\mathring{U}}$  & renormalized $\mathbf{U}$; \S\ref{S:YMipv}, eqn. \eqref{e:rgu} \\
	$\tilde{\mfu} $ & gauge element of $G$; \S\ref{s:maprf}, eqn. \eqref{gauge-ode}\\
	$\tilde{A}^\star_a,  \tilde{F}^\star_{a b}, E^\star_a,  \tilde H^\star_{ab} $ & YM fields under temporal gauge \eqref{phys-temp-gauge}; \S\ref{mainthm-proof}  \\
	$\mathbf{U}_0$ & initial data for $\mathbf{U}$; \S\ref{mainthm-proof}, Step $4$ \\
	$\widehat{\mathbf{U}}_0 $ & initial data for $\widehat{\mathbf{U}}$; \S\ref{mainthm-proof}, Step $4$ 
\end{longtable}



	\section*{Acknowledgement}
	We thank the anonymous referees for their comments. 
	C.L. is supported by the Fundamental Research Funds for the Central Universities, HUST: $5003011036$.
	J.W. is supported by NSFC (Grant No. $12271450$ and $11701482$). 



 \bibliographystyle{amsplain}
 \bibliography{Reference_Chao}

\end{document}